\documentclass{amsart}
\usepackage{enumerate}
\usepackage{amsmath}
\usepackage{amsthm}
\usepackage{amsfonts}
\usepackage{amssymb}
\usepackage{mathrsfs}
\usepackage{stmaryrd}
\usepackage{tikz} 
\usetikzlibrary{matrix,arrows,decorations.pathmorphing}
\usepackage[pdfpagescrop={90 60 520 725}, bookmarks, colorlinks, breaklinks]{hyperref}  
\hypersetup{linkcolor=blue,citecolor=blue,filecolor=black,urlcolor=blue}

\newtheorem{theorem}{Theorem}[section]
\newtheorem{proposition}[theorem]{Proposition}

\newtheorem{lemma}[theorem]{Lemma}
\newtheorem{corollary}[theorem]{Corollary}

\theoremstyle{remark}

\newtheorem{remark}[theorem]{Remark}
\theoremstyle{plain}

\newcommand\wt\widetilde
\newcommand\sheaf\mathcal
\newcommand\complex\mathscr
\newcommand{\exterior}[1]{\ensuremath{{\textstyle\wedge}^{\! #1 }}}
\newcommand\lto\longrightarrow

\newcommand{\Sp}[1]{\ensuremath{\text{Splittings} #1  }}

\newcommand\Hom{\ensuremath{\text{Hom}}}
\newcommand\Ker{\ensuremath{\text{Ker}}}
\newcommand\Pic{\ensuremath{\text{Pic}}}

\newcommand\rank{\ensuremath{\text{rank}}}
\newcommand\Gr{\ensuremath{\text{Gr}}}
\newcommand\Id{\ensuremath{\text{Id}}}

\newcommand\Tot{\ensuremath{\text{Tot}}}

\newcommand\Aut{\ensuremath{\text{Aut}}}
\newcommand\Der{\ensuremath{\text{Der}}}
\newcommand\Isom{\ensuremath{\text{Isom}}}

\newcommand\PP{\ensuremath{\mathbb P}}
\newcommand\CC{\ensuremath{\mathbb C}}

\newcommand\ZZ{\ensuremath{\mathbb Z}}

\renewcommand\AA{\ensuremath{\mathbb A}}

\newcommand\M{{\ensuremath{\mathcal{M}}}}
\newcommand\D{{\ensuremath{\mathcal{D}}}}
\newcommand\sM{{\ensuremath{\mathcal{SM}}}}
\newcommand\SM{{\ensuremath{\mathfrak{M}}}}
\newcommand\A{{\ensuremath{\mathcal{A}}}}
\renewcommand\S{{\ensuremath{   {\mathcal{S}}}}}
\newcommand\OO{{\ensuremath{\mathcal{O}}}}
\def\hat{\widehat}

\def\O{{\mathcal O}}

\def\Z{\ZZ}
\def\T{{\mathcal T}}
\def\L{{\mathcal L}}
\def\bar{\overline}
\def\red{{\mathrm{red}}}
\def\C{{\Bbb C}}
\def\tilde{\widetilde}
\begin{document}
\bibliographystyle{hep}

\parindent=0cm
\parskip=0.5cm

\title[Supermoduli Space is Not Projected]{Supermoduli Space is Not Projected}
\author{Ron Donagi}
\address{Department of Mathematics, University of Pennsylvania, Philadelphia, PA 19104}
\email{donagi@math.upenn.edu}
\author{Edward Witten}
\address{Institute for Avanced Study, Princeton NJ  08540}
\email{witten@ias.edu}

\date{\today}

\begin{abstract}
We prove that for genus $g\geq 5$, the moduli space of super Riemann surfaces is not projected (and in particular is not split): it cannot be holomorphically projected to its
underlying reduced manifold.  Physically, this means that certain approaches to superstring perturbation theory that are very
powerful in low orders have no close analog in higher orders.  Mathematically, it means that the moduli space of super Riemann
surfaces  cannot be constructed in an elementary way starting with the moduli space of ordinary Riemann surfaces. It has a life of its own.
\end{abstract}
\maketitle
\newpage
\tableofcontents
\thispagestyle{empty}
\newpage

\section{Introduction}
Ordinary geometry has a generalization in $\ZZ_2$-graded supergeometry, which is the arena for supersymmetric
theories of physics.  In this generalization, ordinary manifolds are replaced by supermanifolds, which are endowed 
with $\ZZ_2$-graded rings of functions.  
In addition to a vast physics literature, supermanifolds have also been much studied mathematically; for example,
see  \cite{Green, Berezin, Manin, Vaintrob, Rothstein, Onishchik}.   A basic example of a supermanifold is a super Riemann surface, 
which for our purposes  is a complex supermanifold
of dimension $(1|1)$ with a superconformal structure, a notion that we explain in section 3.  See for example \cite{BS1, BS2, F, M2,RSV,SRSnotes} (and see \cite{DRS} for
the generalization to super Riemann surfaces of dimension $(1|n)$, $n>1$).

Mathematically, the theory of super Riemann surfaces and their moduli spaces generalizes the theory of ordinary Riemann
surfaces and their moduli spaces in a strikingly rich way.
Physically, the main importance of super Riemann surfaces is their role in superstring perturbation theory.
Perturbative calculations in superstring theory are carried out by integration over the moduli space $\SM_g$ of super Riemann
surfaces, and its analogs for super Riemann surfaces with punctures.  See for example \cite{DPh,Revisited}.   We will write $\SM_{g,n}$ for
the moduli space of genus $g$ super Riemann surfaces with $n$ marked points or Neveu-Schwarz punctures (we do not consider in this
paper the more general moduli spaces of super Riemann surfaces with Ramond punctures).

Both physically and mathematically, one of the most basic questions about $\SM_g$ is whether it can be projected holomorphically
to its reduced space, which is the moduli space $\sM_g$ that parametrizes ordinary Riemann surfaces with a spin structure.  (By such a projection,
one means a holomorphic map that is left inverse to the natural inclusion of $\sM_g$ in $\SM_g$.)
A complex supermanifold that can be projected holomorphically to its reduced space is said to be projected; if the fibers of the projection
are linear in a certain sense that will be described later, the supermanifold is said to be split.  

Mathematically, if $\SM_g$ is split, this means that it can be reconstructed from the purely bosonic moduli space $\sM_g$ in an elementary fashion and, in a sense,
need not be studied independently.  Physically, if $\SM_g$ is split -- or at least projected -- then a possible strategy in 
superstring perturbation theory is to integrate over $\SM_g$ by first integrating over the fibers of its projection to $\sM_g$.
Indeed, practical calculations in superstring perturbation theory -- such as the $g=2$ calculations that are surveyed in \cite{DPhgold} --
are usually done in this way.  

However, there has been no evidence that $\SM_g$ is split, or even projected,
in general.  The validity of superstring perturbation theory does not depend on a projection, so the existence of string theory gives no hint that
$\SM_g$ is projected. The existence of holomorphic projections for
small $g$ follows from the cohomological nature of the obstructions to splitting and the nature of the reduced spaces $\sM_g$ for
small $g$, and gives little indication of what happens for larger $g$.
The goal of the present paper is to show that actually $\SM_g$ is not projected or split in general.

In fact, we show the following:
\begin{theorem}\label{Main1}
The supermanifold $\SM_g$ is non-projected, and in particular non-split, for $g \geq 5$.
\end{theorem}
(We suspect that this result may hold for $g\geq 3$.)

Our second main result is:
\begin{theorem}\label{Main2}
The supermanifold $\SM_{g,1}$ is non-projected for $g \geq 2$, for the case of an even spin structure.
\end{theorem}

Once this is established, a simple argument gives
\begin{theorem}\label{Main3}
The supermanifold $\SM_{g,n}$ is non-projected for $g \geq 2$ and $g-1 \geq n \geq 1$.  This holds for both even and odd spin structures
if $g$ is odd and for  even spin structure if $g$ is even.
\end{theorem}
We do not resolve the question of projectedness of $\SM_{g,n}$ for the case of even $g$, odd spin structure, and $n\geq 1$.

\remark{Though $\SM_g$ and $\SM_{g,n}$, like their bosonic counterparts, have natural Deligne-Mumford compactifications, we do not
consider the compactifications in the present paper.  Our assertion is that $\SM_g$ and $\SM_{g,n}$ are non-projected without
any consideration of the compactification. Asking for the compactification to be projected would only be a stronger condition.}\label{secondremark}

\remark{Although the upper bound on $n$ in Theorem \ref{Main3} is probably not optimal, one should expect to require some sort of upper bound precisely
because we do not consider the compactification of $\SM_{g,n}$ and thus we require the $n$ marked points to be distinct.  For $n>>g$, requiring the $n$ points
to be distinct very likely kills the obstructions to splitness, though we would expect the compactification of $\SM_{g,n}$ to be non-split.}\label{newremark}

This paper is organized as follows.
In section 2, we review the basic notions of supermanifold theory. We start with several constructions of new supermanifolds from old: as submanifolds, coverings, branched coverings,  blowups and blowdowns.  We explain the concepts of projection and splitting
and the cohomology classes $\omega_i$ that obstruct a splitting.   A useful result is the Compatibility Lemma \ref{mila}, comparing the leading obstruction
$\omega_2$ for a supermanifold $S$ and a submanifold $S'$. In fact, the part of this discussion that involves $\omega_2$  is considerably simpler than the full story, and is the only part that will be used in the rest of the paper. The reader interested  only in the main results of this work may safely skip our analysis of the higher obstructions.

In section 3, we introduce our main objects of study, namely super Riemann surfaces and their moduli spaces.  In sections 3.1 and 3.2 we review some of the basics of super Riemann surfaces and examine in particular the effects on them of the branched covering, blowup and blowdown constructions considered previously for supermanifolds. In section 3.3, we exhibit an explicit and basic example of a non-split supermanifold $X_{\eta}$. This has dimension $(1|2)$, and can be thought of as a family of super Riemann surfaces parametrized by a single odd parameter.

The  main theorems are proved in sections 4 and 5.
Let us outline some of the ideas of the proofs.

The  above-mentioned non-split supermanifold $X_{\eta}$ embeds naturally into $\SM_{g,1}$. Using the Compatibility Lemma and some standard algebraic geometry, we use this embedding to prove that  the leading obstruction to a projection of $\SM_{g,1,{\text{even}}}$ is non-vanishing, proving Theorem \ref{Main2}. The argument fails for odd spin structure.  However, by considering unramified covers of $X_\eta$, it is then possible to
deduce Theorem \ref{Main3}.

To prove Theorem \ref{Main1}, we follow a similar path. We describe a covering space $\wt{\SM}_{{g_0},1} \to \SM_{{g_0},1}$ parametrizing super Riemann sufaces of genus $g_0$ together with a particular type of branched covering, and we find 
an explicit embedding of supermoduli spaces, mapping $\wt{\SM}_{{g_0},1}$ into $\SM_g$.  In fact, we can find such an embedding for every $g \geq 5$. The normal bundle sequence for this
embedding splits, so non-splitness of $\SM_g$ follows from that of 
$\wt{\SM}_{g_0,1,{\text{even}}}$. In fact, we see that $\wt{\SM}_{g_0,1,{\text{even}}}$ is itself reducible. Conveniently, its two components go to the two components of $\SM_g$, so we are able to deduce non-splitness for both of the latter.

In the appendix we discuss in more detail the simplest instance of our family of branched coverings, namely the case $g_0=2, ~ g = 5$. We give an elementary construction of the families involved, check that both parities on $\sM_5$ arise from even parity on $\sM_{2,1}$, and analyze the parameter spaces in $\sM_5$ and in the ordinary (non-spin) moduli space $\M_5$. Somewhat surprisingly, they both turn out to be curves of genus 0.

In a sequel to this work \cite{DWtwo}, we 
will provide some further interpretations of  the first non-trivial class  $\omega := \omega_2$.   We show that this class is, in a certain sense, a superanalog of what in ordinary algebraic
geometry is the Atiyah class of the tangent bundle.  In the case of the moduli space of super Riemann surfaces, we give a concrete description of $\omega$ in terms
of sheaves on $C\times C$, where $C$ is an ordinary Riemann surface, and use this to give an alternative proof of Theorem \ref{Main2}.

\remark{As explained in section \ref{stacks}, the phrase ``moduli space'' is oversimplified in the present context.
The object  of study,  $\SM_g$, is really the moduli {\it stack} of super Riemann surfaces.
The first obstruction to splitting or projection of this stack is a cohomology class $\omega_2$ that can be restricted to any substack $S'$.
We consider a particular substack $S'$ of dimension $1|0$.    It has a finite cover $\pi:S\to S'$ that is a smooth compact curve that parametrizes a
family $X\to S$ of smooth  compact split super Riemann surfaces. We show that $\omega_2(\SM_g)$ is
nonzero by showing that $\pi^*(\omega_2(\SM_g)|_{S'})\not=0$. Passing from $S'$ to $S$ eliminates
the automorphisms of the objects parametrized by $S'$.


\section{Supermanifolds} \label{supermanifolds}

A {\em  supermanifold}, like an ordinary manifold, is a locally ringed space which is  locally isomorphic to a certain local model. 
A locally ringed space is a pair $(M,\OO)$ consisting of a topological space $M$ and a sheaf of algebras $\OO=\OO_{M}$ on it, whose stalk $\OO_x$ at each point $x \in M$ is a {\em local ring}. 
One type of example is affine space $\AA^m$, for which $M$ is an $m$-dimensional vector space, while $\OO$ can be the sheaf of functions on $M$ that are continuous, differentiable,  or analytic. (Here and elsewhere, we will always work in characteristic 0, in fact over the real or complex numbers.)
A manifold is then a commutative locally ringed space which is  locally isomorphic to one of these local models. In the complex case, we  also have the option of taking the functions in $\OO$ to be algebraic,
but we must then allow a larger collection of local models, namely all non-singular affine varieties, i.e. closed non-singular algebraic subvarieties of affine space.\footnote {In each of these cases, the stalk $\OO_x$ is the ring of germs of functions, i.e. functions defined on any neighborhood of $x$, where two functions are identified if their restrictions to some open subset coincide. This is indeed a local ring: the germs of functions that vanish at $x$ form the unique maximal ideal.}

A $\ZZ/2$-graded sheaf of algebras $A=A_0 \oplus A_1$ is {\em  supercommutative} if it is commutative up to the usual sign rule:
for $f \in A_i,   g \in A_j$, the rule is $fg = (-1)^{ij} gf$. 
Given a manifold $M$ and a vector bundle $V$ over it, we define the supercommutative locally ringed space 
$S :=S(M,V)$ to be the pair $(M,\OO_{S})$, where $ \OO_{S}$   is the sheaf of $\OO_{M}$-valued sections of the exterior algebra $\exterior{\bullet} V^{\vee} $ on the dual bundle $V^{\vee}. $  (If we interpret $V$ as a locally free sheaf of 
$\OO_M$-modules,  then
$ \OO_{S} $ is simply $\exterior{\bullet} V^{\vee}. $)  This sheaf $\OO_S$ is $\ZZ/2$-graded and  supercommutative, and its stalks are local rings.  The simplest example is {\em affine superspace} 
$\AA^{m|n} = (\AA^m, \OO_{ \AA^{m|n} })=S(\AA^m, {\OO_{\AA^m}}^{\oplus n})$, 
where $M$ is ordinary affine $m$-space and $V = \OO_M^{\oplus n}$ is the trivial rank $n$ bundle on it. 
A {\em  supermanifold} is then a supercommutative locally ringed space which is  locally isomorphic to some local model $S(M,V)$.
\footnote {In the continuous, differentiable or analytic settings, we may as well restrict our local models to the affine superspaces 
$\AA^{m|n}.$ In the algebraic setting we must allow all $S(M,V)$ with affine $M$, as in the bosonic algebraic case.}
It is {\em split} if it is globally  isomorphic to some $S(M,V)$. 

It is important to note that the isomorphisms above are isomorphisms of $\ZZ/2$-graded algebras (over the real or complex numbers).  They need not preserve the $\ZZ$-grading of $\exterior{\bullet} V^{\vee} $. So if $z$ is a function on $M$ and $\theta_i$ are fiber coordinates on $V$, then $z$ cannot go to $z + \theta_1$, but it can go to $z + \theta_1 \theta_2$. 

The definition of a supermanifold endows  the sheaf $\OO_S$ with the structure of a sheaf of (real or complex) algebras, and also with  a surjective homomorphism to $\OO_M$ (more on this below). It does not endow the sheaf $\OO_S$ with the structure of a sheaf of $\OO_M$-modules: multiplication by $\OO_M$ need not commute with the gluing isomorphisms. We say that the supermanifold $S$ is {\em projected} when the sheaf $\OO_S$ can be given the structure of a sheaf of $\OO_M$-algebras commuting with the given projection to $\OO_M$  . (It may be possible to do this in more than one way.)  Clearly, every split supermanifold is projected, but we will see that not every supermanifold is projected, and likewise there are obstructions for a projected supermanifold to be split.

The structure sheaf $\OO_S$ of a supermanifold $S$ contains the ideal $J$ consisting of all nilpotents. It is the ideal generated by all odd functions. Given a supermanifold $S$, we can recover its underlying manifold $(M,\OO_M)$, which we call the reduced space  ${S}_{{\mathrm{red}}}$, as well as a bundle $V$ on it: $\OO_M$ is recovered as $\OO_S/J$, while the dual $V^{\vee}$ is recovered as $J/{J^2}$. In fact, the supermanifold $S$ determines the split supermanifold $\Gr(S)$ whose reduced space is $M$ and whose structure sheaf is $\Gr_J(\OO_S) := \bigoplus_{i=0}^{\infty} (J^i/J^{i+1})$, and the latter  determines and is uniquely determined by the pair $(M,V)$:
\[
S  \implies \Gr(S) \Longleftrightarrow (M,V).   
\]
We  say that a supermanifold $S$ with $\Gr(S)\cong S(M,V)$
is modeled on $M,V$, where  $S(M,V)$ is the split supermanifold determined by the pair $M,V$ as above.
We will discuss in  section \ref{obstructions} how to characterize all supermanifolds $S$ that are modeled on $M,V$.

The definition of a supermanifold  leads more or less immediately to various standard notions: a morphism of supermanifolds, a submanifold (a term we use instead of the clunky `sub-supermanifold'), immersion and submersion, a product, a fiber product (or pullback), a fiber. (The fiber of a submersion of supermanifolds is a supermanifold. The fiber of a general morphism of supermanifolds is a possibly singular locally ringed superspace. In the algebraic world, it is a superscheme.) By a family of supermanifolds parametrized by a supermanifold $P$, we mean a submersion $\pi: S \to P$.  
If $0 \in P_{{\mathrm{red}}}$ is a point, the family is interpreted as a deformation of the fiber $S_0 := \pi^{-1}(0)$. It is also straightforward to define vector bundles, sheaves of $\OO_S$-modules, and so on. 

We have noted that a supermanifold $S$ determines its reduced space $M$.
Moreover, $M$ has a natural embedding in $S$; this corresponds to the  existence of
a projection from $\OO_S$ to $\OO_M=\OO_S/J$. In the other direction, a projection $S\to M$ would be equivalent to 
an embedding of $\OO_M$ in $\OO_S$, and also to endowing $\OO_S$ with an $\OO_M$-algebra structure. 
We emphasize that in general these structures do not exist: a general supermanifold may not be projected.

We note that the product of two split supermanifolds is split, and for any morphism $f:M' \to M$, if $S$ is a split or projected supermanifold with reduced space $M$, then $S' = f^*S$ is well-defined and it too is split or projected. 
Some of the simplest examples of supermanifolds are coverings and submanifolds, discussed below in section \ref{FirstExamples}. If $\wt{S} \to S$ is a finite covering map, it is not too hard to see (Corollary \ref{CoroCovering}) that $\wt{S}$ is split  if and only if so is ${S}$ (and projected if $S$ is).
However, a submanifold  $S'$ of a split supermanifold $S$ need not in general be split as a supermanifold. 

A sheaf $\mathcal{F}$ on a supermanifold $S$ can be viewed simply as a sheaf on the reduced space $M$, and sheaf cohomology on $S$ is defined to be the cohomology of the corresponding sheaf on the reduced space: 
$H^*(S, \mathcal{F}) := H^*(M, \mathcal{F})$.   This is the same definition that one uses for sheaf cohomology on a non-reduced scheme
in ordinary algebraic geometry, and as in that case,
 it is usually more illuminating to denote this
cohomology as $H^*(S,\mathcal{F})$ rather than $H^*(M,\mathcal{F})$.  For example, it is much more natural to think of a sheaf of $\OO_S$-modules
as a sheaf on $S$ rather than as a sheaf on $M$.

The tangent bundle $TS$ of the supermanifold $S$ is an example of such a sheaf of $\OO_S$-modules. It can be defined in terms of derivations: $TS := {\mbox{Der}}(\OO_S)$. 
It is a $\ZZ_2$-graded vector bundle, or locally free sheaf of $\OO_S$-modules. When $S = \AA^{m|n}$, $TS$ is the free $\OO_S$-module generated by even tangent vectors $\partial / \partial x_i, i = 1, \ldots, m$ and odd tangent vectors $\partial / \partial \theta_j, j = 1, \ldots, n$.
In general, $TS$ need not have such distinguished complementary even and odd subbundles: the even and odd parts are not sheaves of $\OO_S$-modules. But the  restriction $TS _{|M}$ of $TS$ to the reduced space $M=S_{{\mathrm{red}}}$ does split.
(By restriction to $M=S_{{\mathrm{red}}}$ we mean pullback, under the natural inclusion $M \to S$, 
of a sheaf of $\OO_S$ modules to a sheaf of $\OO_M$ modules. In this case, this is accomplished 
by setting all odd functions to $0$.)
We refer to the graded pieces of $TS _{|M}$ as the 
even and odd tangent bundles of $S$. Explicitly, these are given by:
\[
T_+S := TM,  \quad  \quad  T_-S :=V.
\]
$S$ is then modeled on the pair $M,V$, though of course it may not be isomorphic to the split model $S(M,V)$.
By definition, the {\em dimension} of $S$ is the pair $(m|n)$, where $m,n$ are the ranks of $T_{\pm}S$: $m :=\dim (M)$ and $n:=\rank (V)$.
We note that specifying a map $f_v: \CC^{0|1} \to S$, where $\CC^{0|1}$ is the $(0|1)$ dimensional affine superspace, is equivalent to specifying an odd tangent vector $v \in T_-S$, i.e. a point $p \in M$ and an odd tangent vector $v \in T_{-,p}S$ at that point.

\subsection{Examples of supermanifolds} \label{FirstExamples}
The simplest supermanifolds are the affine superspaces $\AA^{m|n}$, defined above. In the algebraic case, their global function ring is the polynomial ring  in $m$ commuting (even) variables $x_i, i = 1, \ldots, m$, and $n$ anticommuting (odd) variables $\theta_j, j = 1, \ldots, n$. In the other cases, the even part is extended to allow continuous, differentiable or analytic functions of the $x_i$, but the odd part is unchanged. In these cases, it is often convenient to illustrate arguments about supermanifolds by referring to these local coordinates $x_i, \theta_j$. As noted above, in the algebraic case a supermanifold is usually not locally isomorphic to affine superspace, but we can still use analytic local coordinates.

A less trivial example of  a complex supermanifold is complex projective superspace $\PP^{m|n}$, for $m,n \ge 0$. We can think of it globally, as a quotient, or locally, as pieced together from affine charts. The global description involves a quotient by the  purely even group $\CC^\times$, so we use homogeneous coordinates $x_0\ldots x_{m}|\theta_1\ldots\theta_n$,
subject to an overall scaling of all $x$'s and $\theta$'s by the same nonzero even complex parameter $\lambda \in \CC^\times$, and with a requirement that not all the bosonic coordinates
$x_{\alpha}$ are  allowed to   vanish simultaneously.  
 The local description specifies $\PP^{m|n}$ 
as the union of its open subsets $U_\alpha$, for $\alpha=0,\ldots,m$,  defined by the condition
 $x_\alpha\not=0$.
Each $U_\alpha$ can be identified with affine superspace $\AA^{m|n}$ by the ratios $x_{\beta}/x_{\alpha}$, $\beta \not=\alpha$, and $\theta_j/x_\alpha$, for $j=1,\ldots,n$. The gluing relations are  the obvious ones.  Note that for $m=0$ there is a unique open set $U_\alpha$, so $\PP^{0|n}$ is the same thing as $\AA^{0|n}$.

Next we describe three ways of constructing supermanifolds: as submanifolds, blowups, and branched covers.

\subsubsection{\it Submanifolds of supermanifolds}

One way to construct new supermanifolds from a given supermanifold is  by imposing
one (or more) equations. E.g. in $\PP^{m|n}$ impose:
 $P(z^0\dots z^n |\theta^1 \dots\theta^n)=0$, where $P$ is a homogeneous polynomial in the homogeneous coordinates of $\CC\PP^{m|n}$ that is either even or odd.  
If $P$ is even and sufficiently generic, this will give a complex supermanifold of dimension $m-1|n$.  For suitable odd $P$, it gives
a complex supermanifold of dimension 
$m|n-1$.
We reserve the name  {\em divisor} to the case of codimension $(1|0)$, i.e. when the defining polynomial $P$  is even.

\subsubsection{\it Coverings of supermanifolds}\label{Covers}

Given  a supermanifold $S=(M,\OO_S)$ it is straightforward to lift a finite unramified 
covering $f: \wt{M} \to M$ of the underlying reduced space $M$ to a finite unramified 
covering $F: \wt{S} \to S$. The reduced space is of course $\wt{M}$, and  the sheaf  $\OO_{\wt{S}}\to \wt{M}$ is the pullback $f^{-1}(\OO_S)$ of the sheaf $\OO_S$
over $M$.  This pullback has the local structure of the sheaf of functions on a supermanifold since the covering map $f$ is a local isomorphism.

\subsubsection{\it Branched coverings of supermanifolds}\label{covers}

A variant of the above allows us to construct branched coverings as well. Start with:
\begin{itemize}
\item a supermanifold $S=(M,\OO_S)$,
\item a  branched covering $f: \wt{M} \to M$ of the underlying reduced space $M$, with smooth branch divisor $B \subset M$, and
\item  a divisor $D \subset S$ whose intersection with $M$ is $B$. 
\end{itemize}

We construct a supermanifold $\wt{S}$ and a morphism $F: \wt{S} \to S$ whose branch divisor is $D$ and whose reduced version is $F_{{\mathrm{red}}} = f$. For the moment, assume that $S$ has global coordinates $(z_1,\dots,z_m|\theta_1,\dots,\theta_n)$, where $(z_1,\dots,z_m)$ are coordinates on $M$, the branch divisor $B$ is given by $z_1=0$, and the corresponding coordinates on $\wt{M}$ are $(w_1,z_2,\dots,z_m)$,
so that  pulling back by the branched covering $f: \wt{M} \to M$ sends $z_1$ to $(w_1)^k$ for some $k$. (We can always achieve this after restricting to sufficiently small open subsets.) Now the divisor $D$ is given by the vanishing of some even function $z'=z'(z,\theta)$ whose image modulo the $\theta$'s is $z_1$. This implies that the coordinate ring $(z_1,\dots,z_m|\theta_1,\dots,\theta_n)$ is also generated by $(z',z_2,\dots,z_m|\theta_1,\dots,\theta_n)$.  We can therefore construct $\OO_{\wt{S}}$ as the sheaf of $\OO_S$ algebras generated by $w'$, which is defined to be the $k$-th root of $z'$. The reduced space of the resulting $\wt{S}$ is naturally identified with $\wt{M}$: just send $w'$ to $w$, $z_i$ to themselves and the $\theta_j$ to $0$. Now this sheaf is unique up to an isomorphism which itself is unique up to a $k$-th root of unity (=a deck transformation of the covering). This allows us to patch the open pieces to obtain the desired global branched covering $F: \wt{S} \to S$.

The above construction extends to families of supermanifolds. In fact, a branched covering of a family of supermanifolds is a special case of a branched covering of a single supermanifold: given a family $\pi:S \to B$ and the above data, we momentarily forget $\pi$, so we get the branched covering $F: \wt{S} \to S$ of the total space of the family, and then we remember that $S$ (and hence also $\wt{S}$) are families over $B$.

\subsubsection{\it Blowups of supermanifolds}\label{blowup}

Starting with a supermanifold $X$ and its codimension $(k|l)$ submanifold $Y$, we construct the blowup $\wt{X}$ of $X$ along $Y$. Let $y_1,\ldots,y_{m-k}$ and $\eta_1,\ldots,\eta_{n-l}$ be coordinates on an open set $W \subset Y$, while $x_1,\ldots,x_k$ and $\theta_1,\ldots,\theta_l$ are normal coordinates to $Y$ in $X$, so that $x,y,\theta,\eta$ together form coordinates on an open $U \subset X$. We cover $U$ by affine open subsets $U_{\alpha}, \quad \alpha=1,\dots,k,$ given by $x_{\alpha} \neq 0$,
and replace these by new affines $\wt{U}_{\alpha}$ with coordinates:
\[
 y_1,\ldots,y_{m-k}, \quad 
\eta_1,\ldots,\eta_{n-l}, \quad 
x_{\beta}/x_{\alpha}, \beta \neq \alpha, \quad
x_{\alpha}, \quad
\theta_j/x_{\alpha}, j=1,\ldots,l.
\]
The $\wt{U}_{\alpha}$ are glued in the obvious way to give the blowup $\wt{X}$. (As in ordinary algebraic geometry, one can describe the blowup more intrinsically using the Proj construction, by means of which one can more generally blow up an arbitrary sheaf of ideals.)
As in  ordinary algebraic geometry, $\wt{X}$  comes with an {\em exceptional divisor} $E$ and a map $\pi: \wt{X} \to X$ such that $E$ becomes a bundle over $Y$ with fiber $\PP^{(k-1)|l}$ while $\wt{X} \setminus E$ maps isomorphically to $X-Y$. Particularly interesting is the case $k=1$, in which the underlying reduced manifold remains unchanged by the blowup, and only the odd directions are modified. We will encounter this in section \ref{SRS blowup}.

\subsection{Obstructions to splitting} \label{obstructions}

\mbox {   }
The obstructions to splitting of a supermanifold have been analyzed by a number of authors, including Green, Berezin, Manin, Vaintrob, Rothstein, Onishchik and others, cf. \cite{Green, Berezin, Manin, Vaintrob, Rothstein, Onishchik}. In this section  we will describe the space of all supermanifolds modeled on a given
$M,V$ in terms of the cohomology of a certain sheaf $G$ of non-abelian  groups, and explain how the condition for a supermanifold $S$ to be split is equivalent to vanishing of a certain sequence of abelian
cohomology classes.  We go into much more detail in this section than is needed for our later applications, which only depend on the leading
obstruction $\omega_2$.

\subsubsection{\it Green's cohomological description and the obstruction classes}
Let $V$ be a rank $n$ vector bundle on a manifold $M$, and  let $S(M,V)$  be the corresponding split supermanifold. 
In great generality, the set of all objects of some kind that are locally isomorphic to some model object is given as the first cohomology of the sheaf of automorphism groups of the model. When these  automorphism groups are non abelian, the first cohomology is not a group, only a pointed set: the ``point" corresponds to the model object itself.
We start by applying this principle to all supermanifolds with given reduced space $M$ and given 
odd dimension; all of these are locally isomorphic to $S(M,V)$, since any two vector bundles of the same rank on $M$ are locally isomorphic. 
We then restrict to obtain cohomological descriptions of 
(1)  isomorphism classes of pairs consisting of a supermanifold 
together with an isomorphism of its graded version with $S(M, V ),$  and 
(2) isomorphism classes of supermanifolds whose graded version is globally isomorphic 
to the given $S(M,V)$ (Green's theorem).

 By definition, a supermanifold $S$ with reduced space $M$ and  odd dimension $n$ is a sheaf of $\ZZ_2$-graded algebras on $M$ that is locally isomorphic to $\OO_{S(M,V)} = \exterior{\bullet} V^{\vee}$. 
Consider the sheaf of these local isomorphisms, namely the sheaf $\Isom(S(M,V),S)$ on $M$ whose sections on an open $U \subset M$ are the isomorphisms between $S(M,V)_{|U}$ and $S_{|U}$. In case $S=S(M,V)$, this becomes the sheaf of non-abelian groups:
\[
\Isom(S(M,V),S(M,V)) = \Aut(\OO_{S(M,V)})  = \Aut(\exterior{\bullet} V^{\vee}) ,   
\]
where $\Aut(\exterior{\bullet} V^{\vee})$ denotes the sheaf of automorphisms of the  $\ZZ_2$-graded sheaf of algebras $\exterior{\bullet} V^{\vee}$.  These automorphisms send $J$ to itself, so they act on $\OO_M
=\OO_{S(M,V)}/J$, but this action is trivial since these are automorphisms of sheaves on $M$.
We will describe the structure of $\Aut(\exterior{\bullet} V^{\vee})$ below: 
by \eqref{Gdef} it maps onto $\Aut(V)$ with  a kernel $G$, 
which in turn is filtered by subgroups $G^i$ 
with graded pieces given (cf. \eqref{drt}) by:
$G^i / G^{i+1} \cong
 T_{(-)^i}M \otimes \exterior{i}V^{\vee}$.

In general, since $S$ is locally isomorphic to $S(M,V)$,
$\Isom(S(M,V),S)$ is locally isomorphic to $\Aut(\exterior{\bullet} V^{\vee})$, i.e. it is 
a torsor (=principal homogeneous space) over  $\Aut(\exterior{\bullet} V^{\vee})$. Conversely, every such torsor determines a corresponding supermanifold.
The set of isomorphism classes of supermanifolds $S$ with a given reduced space $M$ and given odd dimension $n$ is therefore given by the first cohomology set 
\[
H^1(M,\Aut(\OO_{S(M,V)})) = H^1(M,\Aut(\exterior{\bullet} V^{\vee})).
\]
As noted above, since the group involved is non-abelian, this cohomology is not a group but only a set with a base point. The  base point corresponds to $S(M,V)$. (So far, we could have used a different rank $n$  bundle $V'$; this would have yielded another description of the same set of isomorphism classes of supermanifolds $S$ with the given reduced space $M$ and given odd dimension $n$, the only difference being that the base point would now be $S(M,V')$.)

Let $G = G_V$ be the kernel of the map that sends an automorphism of  $\exterior{\bullet} V^{\vee}$ to the induced automorphism of $V^{\vee} = \Gr^1(\exterior{\bullet} V^{\vee})$, or equivalently to the (transpose inverse) automorphism of of $V$:
\begin{equation}\label{Gdef}
1 \to G \to \Aut(\exterior{\bullet} V^{\vee}) \to \Aut(V) \to 1.  
\end{equation}
Equivalently, $G$ is the group of those automorphisms of the split model $S(M,V)$ that preserve both $M$ and $V$.

Consider a supermanifold $S=(M,\OO_S)$ with an isomorphism $\rho: V \cong  T_{-}S = \Gr^1(\OO_S)^{\vee} $, or equivalently
an isomorphism of $\ZZ$-graded sheaves of algebras:
\begin{equation}\label{zolg}\rho: \exterior{\bullet} V^{\vee} \cong \Gr(\OO_S). \end{equation}
We can compare the pair $(S,\rho)$ to $S(M,V)$, which comes with the natural isomorphism $\Id: V \cong T_{-}S(M,V)$. We find that the sheaf 
\[\
Isom( \  (S(M,V),\Id),~(S,\rho) \  ),
\]
consisting of those local isomorphisms that send $\Id$ to $\rho$, is  a torsor over $G$. 
The terms of the long exact sequence of cohomology sets of \eqref{Gdef}: 
\begin{align*}
1 &\to H^0(M,G) \to H^0(M,\Aut(\exterior{\bullet} V^{\vee})) \to H^0(M,\Aut(V)) \\
&\to H^1(M,G) \to H^1(M,\Aut(\exterior{\bullet} V^{\vee})) \to H^1(M,\Aut(V))
\end{align*}
therefore have the following interpretations:
\begin{itemize}
\item $H^1(M,\Aut(V)$ is the set of isomorphism classes of rank $n$ bundles on $M$, with the base point corresponding to $V$.
\item As  noted above, $H^1(M,\Aut(\exterior{\bullet} V^{\vee}))$ is the set of isomorphism classes of supermanifolds $S$ with reduced space $M$ and  odd dimension $n$. The map to $H^1(M,\Aut(V))$ sends $S$ to $T_{-}S$.
\item $H^1(M,G)$ is the set of isomorphism classes of pairs $(S,\rho)$ where $S$ is a supermanifold with reduced space $M$ and $\rho$ is an isomorphism $V \cong T_{-}S$.
\item The set of isomorphism classes of supermanifolds $S$ modeled on $M,V$ (the isomorphism is required to be the identity on the reduced space $M$, but can act on the odd directions $V$) is therefore identified with the quotient of $H^1(M,G)$ by $H^0(M,\Aut(V))$. The base point corresponds to $S(M,V)$. This is the main result of \cite{Green}. In the present paper we will use only the previous identification of $H^1(M,G)$ itself.
\end{itemize}

The group $G$ has a descending filtration by normal subgroups $G^i \quad ( i=2,3,\ldots )$. We give three descriptions of these subgroups: algebraic, geometric, and analytic. We then use these groups to describe an obstruction theory for the splitting of a supermanifold.

Algebraically,  we define:
\[
G^i = \{g \in G \quad | \quad g(x) - x \in J^i  \quad \forall x \in \wedge^{\bullet} V^\vee \}.
\]
One has $G^2 = G$, while $G^i$ is trivial if $i$ exceeds the odd dimension $n$ of $S$.
Modulo higher order terms, each $g \in G$ is a $\exterior{i}V^{\vee}$-valued, even derivation. 
In other words, there is a natural isomorphism for  $i \ge 2$:
\begin{equation}\label{drt}
G^i / G^{i+1} \cong
 T_{(-)^i}M \otimes \exterior{i}V^{\vee}.
\end{equation}
On the right hand side, $ T_{(-)^i}M \otimes \exterior{i}V^{\vee}$ is understood simply as a sheaf of abelian groups under addition.  This isomorphism is easiest to see from the geometric or  analytic descrptions of the $G^i$, which we give next.

Geometrically, we interpret the $G^i$ in terms of a filtration of  $S$ itself.
Given a supermanifold $S=(M,\OO_S)$ with nilpotent subsheaf $J\subset \OO_S$, it is convenient to introduce
\begin{equation}\label{Si}
S^{(i)} := (M, \OO_{S^{(i)}} =\OO_S/J^{i+1}).
\end{equation}
These $S^{(i)}$ are locally ringed subspaces of $S$, though they are not supermanifolds, except for the extremes $i=0,n$.  (They are superanalogs of non-reduced schemes in ordinary algebraic geometry. ) They form an increasing filtration of $S$:   
\[
S_{red} = S^{(0)} \subset S^{(1)} \subset \dots  \subset S^{(i-1)} \subset S^{(i)} \subset \dots \subset S^{(n)}  =S
\] 
Recall that automorphisms of the exterior algebra preserve $J$, hence they preserve$ J^i$ so they preserve the filtration of $S$ by the $S^{(i)}$. In the above, we can in particular take $S$ to be the split model $S(M,V)$. An equivalent definition of the $G^i$ is as those automorphisms of $S(M,V)$ that act as the identity on $S(M,V)^{(i-1)}$.

Analytically, it is natural to interpret these groups in terms of vector fields on the split model  $S(M,V)$.
Concretely, the (Lie algebra $\mathfrak{g}$ of the) group $G$ is generated by vector fields on   $S(M,V)$ that are schematically of the form
$\theta^{2k}\partial_x$ or $\theta^{2k+1}\partial_\theta$, $k\geq 1$.  These expressions are shorthand for vector
fields on $S(M,V)$ that in local coordinates $x_1,\dots,x_m|\theta_1,\dots,\theta_n$ take the form
\begin{equation}\label{zelbo} \sum_{a_1,\dots, a_{2k}=1}^n\sum_{b=1}^m f_{a_1,\dots,a_{2k};b}(x_1,\dots,x_m)\theta_{a_1}\dots \theta_{a_{2k}}\frac{\partial}{\partial x^b}\end{equation} or 
\begin{equation}\label{welbo} \sum_{a_1,\dots, a_{2k+1}=1}^n \sum_{s=1}^n f_{a_1,\dots,a_{2k+1};s}(x_1,\dots,x_m)\theta_{a_1}\dots \theta_{a_{2k+1}}\frac{\partial}{\partial \theta^s},\end{equation} 
respectively.  
In these terms, the (Lie algebra $\mathfrak{g}^i$ of the)  subgroup $G^i$ is generated by vector fields on   
$S(M,V)$ that are schematically of the form
$\theta^j\partial_x$ or $\theta^j\partial_\theta$, depending on the parity of $j$, for $j \geq i$. The abelian
$G^i / G^{i+1} \cong T_{(-)^i}M \otimes \exterior{i}V^{\vee}$
can be identified with its Lie algebra. In agreement with \eqref{drt}, it 
is just the sheaf of vector
fields on $S(M,V)$ that are schematically $\theta^i\partial_x$ or $\theta^i\partial_\theta$, depending on the parity of $i$. 
Since $\mathfrak{g}$ is nilpotent, the exponential map $\exp: \mathfrak{g} \to G$ is a bijection, inducing  bijections on global sections: $\exp: H^0(\mathfrak{g}) \to H^0(G)$  and   $\exp: H^0(\mathfrak{g}^i) \to H^0(G^i)$. But since $\exp$ does not respect the group structure of the two sheaves $\mathfrak{g}, G$, the bijection on $H^0$'s need not be an  isomorphism of groups, and there is no induced bijection on $H^1$'s.

Using either of these equivalent descriptions of the $G^i$, 
we obtain natural interpretations for their cohomologies $H^0$ and $H^1$.
by a {\em splitting} of the supermanifold $S$ we mean 
an isomorphism from the split supermanifold $S(M,V)$ to $S$
that induces the identity on both the underlying reduced space $M$ and the odd tangent bundle $V$.
The family of all splittings of $S$ is parametrized by 
$$\Sp(S) := \Isom_{M,V}(S(M,V),S).$$
A bit more generally, we have the notion of a splitting of  the superspace $S^{(i)}$, i.e. 
an isomorphism from $S(M,V)^{(i)}$ to $S^{(i)}$
that induces the identity on both  $M$ and $V$, and
the parameter space  $\Sp(S^{(i)})$ of all  such splittings. 
For the split $S=S(M,V)$, we have identifications 
$\Sp(S) \cong H^0(G)$ and $\Sp(S^{(i-1)}) \cong H^0(G/G^i)$. 
For a general $S$, we get instead that $\Sp(S^{(i-1)})$ is an $H^0(G/G^i)$-torsor, 
which is non-empty if and only if $S$ is split.
Not all splittings of $S^{(i-1)}$ lift to splittings of $S$. The variety 
$\Sp(S^{(i-1)})^S$ of those splittings that do lift to $S$ is (a torsor over)
$ H^0(G)/H^0(G^i) \subset  H^0(G/G^i)$.
Similarly, $ H^0(G^i) $ itself parametrzes those splittings of $S=S(M,V)$ that 
induce the identity splitting of $S^{(i-1)}$.
In full generality, we may consider splittings of $S^{(j-1)}$ that lift to $S^{(k-1)}$  and 
induce the identity splitting of $S^{(i-1)}$, 
whenever $2 \leq i \leq j \leq k$. 
In terms of the $G^i$, this is:
\begin{equation}\label{dwsrxe}
\Sp( S^{(j-1)})_{S^{(i-1)}}^{S^{(k-1)}} \quad\cong \quad H^0(G^i/G^k) \quad / \quad H^0(G^j/G^k).
\end{equation}
Another useful case is when $j=i+1$ and $k = \infty$ or $k=j$: 
the splittings of $S^{(i)}$ that lift to $S$ and are trivial on $S^{(i-1)}$
are parametrized by $H^0(G^i) / H^0(G^{i+1})$, while
all splittings of $S^{(i)}$ that are trivial on $S^{(i-1)}$
are parametrized by
$H^0(G^i/ G^{i+1}) \cong H^0(M,T_{(-)^i}M \otimes \exterior{i}V^{\vee})$. 
The latter is a vector space, and we will see below (in the proof of Corollary \ref{FF}) 
that the former is actually a linear subspace.

The obstruction theory for splitting of the supermanifold $S$ 
is based on filtering  $H^1(M,G)$ by the images of the $H^1(G^i)$.
The geometric interpretation of 
$H^1(G^i)$ is as the set of isomorphism classes of pairs $\varphi_{i-1}=(S,\rho_{i-1})$, 
where $S$ is a supermanifold with reduced space $M$, and 
$\rho_{i-1} \in \Sp(S^{(i-1)})$ is an isomorphism between $S(M,V)^{(i-1})$ and $S^{(i-1)}$.
In order for a class $\varphi = (S,\rho)$ in $H^1(M,G)$ to vanish, 
it is necessary and sufficient that, for each $i\geq 2$, 
this class should be  the image of some $\varphi_{i-1} \in H^1(M,G^i)$.  
This is clear, since  for sufficiently high $i$ 
the $G^i$ vanish and the $S^{(i)}=S$.

There is nothing to check for $i=2$, since $G^{2} = G$, $\rho_1 = \rho$ and $\varphi_{1} = \varphi$.  If a given class $(S,\rho)$ is in the image for some  $i$, then to decide if it is in the image for $i+1$, we look at the exact sequence
\begin{equation}\label{proz} H^1(M,G^{i+1})\to H^1(M,G^i) \stackrel{\omega}{\to} H^1(M,  T_{(-)^i}M \otimes \exterior{i}V^{\vee}).\end{equation}
The obstruction for a class  $\varphi _{i-1}\in H^1(M,G^i)$ to come from some $\varphi_{i}\in H^1(M,G^{i+1})$
is that its image  
\begin{equation}\label{poiudz}
\omega_i = \omega(\varphi_{i-1}) \in H^1(M,T_{(-)^i}M \otimes \exterior{i}V^{\vee})
\end{equation}
must vanish.  This class $\omega_i$ is called the $i$-th obstruction class for splitting of $S$, and we refer to $\varphi_{i}\in H^1(M,G^{i+1})$ as a level $i$ trivialization (or level $i$ splitting) of $S$. 
The condition for lifting a given $\varphi _{i-1}$ to an  isomorphism  $\varphi_{i}: S^{(i)} \to S(M,V)^{(i)}$
is that $\omega_i= \omega({\varphi_{i-1}})=0$.
This was the basis for the original definition of the classes  $\omega_i$ in \cite{Berezin, Manin}; for the interpretation we have  described above via cohomology of the sheaf of non-abelian groups $G$, see \cite{Onishchik}.

Note that we have
defined  $\omega_i$  only if the $\omega_j$ vanish, for all $2\leq j<i$, and only after a level $i-1$ trivialization $\varphi _{i-1} $ has been chosen. We discuss
these issues below. The splitness of $S$ is equivalent to the existence of a level $i$ trivialization $\varphi _i $ for some $i \geq n$, and hence
can be investigated recursively using the obstruction classes $\omega_i$.

 We summarize the limitations in the definition of the higher obstructions $\omega_i$  as follows:
\begin{enumerate}
\item As noted above, $\omega_i$  is defined only if the $\omega_j$ vanish, for all $2\leq j<i$. 
\item Even then, $\omega_i= \omega({\varphi_{i-1}})$ may depend not only on the supermanifold $S$ but also on the choice of trivialization ${\varphi_{i-1}}$. In  section \ref{cwsdf} below we give an example showing that even for the split model $S(M,V)$, a non-standard choice of trivialization $\varphi _2$ may lead to $\omega_3 \neq 0$, while of course the standard  choice of trivialization $\varphi _2$ gives $\omega_3 = 0$. So the non-vanishing of some higher $\omega_i$ for a particular choice of $\varphi_{i-1}$ 
is not sufficient to deduce that a supermanifold is non-split.
\item On the other hand, we will show, also in section  \ref{cwsdf},  that for even $i$, while $\omega_i$ may depend on  $\varphi _{i-2}$, it is independent of how $\varphi_{i-2}$ is lifted to  $\varphi _{i-1}$.
\item In section \ref{cwsdfr} we will describe analogous classes $\omega_i^-$, for even $i$, that obstruct the existence of a projection  $\varphi^-_{i}: S^{(i)} \to S^{(0)}$. The $\omega_i^-$ for odd $i$ vanish identically: the projection $\varphi^-_{2k-2}$ lifts uniquely to $\varphi^-_{2k-1}$. The obstruction $\omega_{2k}^-$ depends only on $\varphi^-_{2k-2}$. When  $\varphi^-_{2k-2}$ is taken as the image of  $\varphi_{2k-2}$ (i.e. it is the projection determined by the level $i$ trivialization  $\varphi_{2k-2}$)   and $\omega_{2k}$ is defined, it equals  $\omega_{2k}^-$. In this sense,  $\omega_{2k}$ can be made to depend only on the projection data.
\item These obstructions to projection also depend on previous choices. For instance, Proposition 4.9.5 of \cite{Berezin} shows that when $\omega_2 =0$, so that $\varphi_{2} = \varphi_{2}^- = \varphi_{3}^-$ can be chosen, but $\omega_3 \neq 0$, the next class $\omega_4^-$ depends linearly and non trivially on the choice of $\varphi_2$.

\end{enumerate}

None of these points  affect the first obstruction class $\omega_2$, which is an invariant of any supermanifold $S$ and obstructs a splitting
or projection.  Our proofs that various moduli spaces are non-split and in fact non-projected will boil down to showing that $\omega_2$ is nonzero.

\subsubsection{\it Illustrations}\label{cwsdf}

Consider for example a split supermanifold $S=S(M,V)$ of dimension $1|3$. The filtration is
\[
1 = G^4 \subset G^3 \subset G^2 = G,
\]
and we have a short exact sequence
\begin{equation}\label{poioi}
1  \to G^3  \to  G  \to  G^2/G^3 \to  1    
\end{equation}
with
\begin{align*}
G^3=G^3/G^4 &= Hom(\wedge^3T-,T-) = T_{-} \otimes \wedge^3T_{-}^* \\
G^2/G^3 &= Hom(\wedge^2T-,T+) = T_{+} \otimes \wedge^2T_{-}^*.
\end{align*}
The  trivial class $1 \in H^1(G)$ has the standard trivialization  $\varphi_2 =1 \in H^1(G^3)$, whose obstruction $\omega_3(1)$ vanishes.
Since $G^4=1$, the map $\omega$ in sequence \eqref{proz} for $i=3$ is injective.
So any exotic lift $\varphi_2 \neq 1$ of the trivial class $1 \in H^1(G)$ must be obstructed; it cannot be extended to a splitting of $S$.
We will show  that the coboundary map:
\begin{equation}\label{potoi}
H^0(Hom(\wedge^2T_{-},T_{+})) \to H^1(Hom(\wedge^3T_{-},T_{-}))     
\end{equation}
can be non-zero. Any  non-trivial  $\varphi_2 \neq 1$ in its image would then be an obstructed, exotic level 2 trivialization of $S(M,V)$, with $ \omega(\varphi_2) \neq 1$.

For a split supermanifold $S=S(M,V)$ of arbitrary dimension $m|n$, the Lie algebra $\mathfrak{g}$ is the subalgebra of $\Der(\wedge^{\bullet}V)$ consisting of even derivations sending $J^i$ to $J^{i+2}$ for each $i$. The action of each $g \in 
\mathfrak{g}$ gives a first-order differential operator $ V^{\vee} \to \wedge^3 V^{\vee}$, so we have a sheaf map
\begin{equation}\label{poipoi}
\mathfrak{g} \to D^1(V^{\vee}, \wedge^3 V^{\vee}).
\end{equation}

Several simplifications occur when $n=3$: $\mathfrak{g}$ becomes abelian, the exponential map $\exp: \mathfrak{g} \to G$ sends $x \mapsto 1+x$ and induces an isomorphism of groups. In particular it induces a bijection $\exp: H^1(\mathfrak{g}) \to H^1(G)$. The map \eqref{poipoi} becomes an isomorphism, which in fact takes the Lie algebra variant of short exact sequence \eqref{poioi} to the symbol sequence:
\begin{equation}\label{totoi}
0 \to  \wedge^2 V^{\vee}   \to D^1(V^{\vee}, \wedge^3 V^{\vee}) \stackrel{\sigma}{ \to}   T_{+} \otimes \wedge^2 V^{\vee}   \to 0.
\end{equation}
The coboundary map in \eqref{potoi} can therefore be identified as cup product with the extension class of the symbol sequence \eqref{totoi}, which is induced from the Atiyah class of $V$. 

For simplicity, consider the case that $V$ is the direct sum of three line bundles $L_i$. Our coboundary decouples as the sum of three maps:
\[
c_1(L_i) : H^0(T_{+} \otimes L_j^{\vee} \otimes L_k^{\vee})  \to H^1(   L_j^{\vee} \otimes L_k^{\vee}),
\]
where $\{ i,j,k \} $ is a permutation of $ \{ 1,2,3 \} $, and we have identified the Atiyah class of a line bundle with its first Chern class. This is clearly non-zero for general choices. For example, this is the case for the super line $\PP^{(1|3)}$, where $T_{+} = \OO_{\PP^1}(2)$ and the $L_i$ are $\OO_{\PP^1}(1)$; or for a superelliptic curve where two of the $L_i $ are trivial, as is $T_{+}$, while the third $L_i $ has non-zero degree.

The situation is very different though for even $i$. In that case, $\omega_i$ depends on  $\varphi _{i-2}$, but is 
independent of the choice of its lift $\varphi _{i-1}$.
This is seen by considering the rescaling action along the fibers of $V$. Under this action, the $\theta$ have degree 1, $\partial_{\theta}$ has degree $-1$, and the $x$'s and $\partial_{x}$ are neutral. We see that the vector fields in equations \eqref{zelbo} and \eqref{welbo} both have the same degree $2k$. The Lie algebra $\mathfrak{g}$ is therefore graded, with a two-dimensional graded piece for each even weight $2k$. It follows that  the coboundary map:
\begin{align*}\label{potpi}
 H^0(Hom(\wedge^{2k-1}T_{-},T_{-}))     &\to   H^1(Hom(\wedge^{2k}T_{-},T_{+})) \\
H^0(G^{2k-1}/G^{2k}) &\to H^1(G^{2k}/G^{2k+1})
\end{align*}
goes between pieces of different weights, so it must vanish, and therefore the ambiguity, given by the composition
\[
H^0(G^{2k-1}/G^{2k}) \to H^1(G^{2k})  \stackrel{\omega}{\to}  H^1(Hom(\wedge^{2k}T_{-},T_{+}))
\]
vanishes as well. Actually, a stronger result holds: in an appropriate sense, the even $\omega_i$ can be chosen to be independent of 
$\varphi _{j}$ for all odd $j<i$. We give a direct construction of these choice-independent $\omega_{2k}$ next.

\subsubsection{\it Analog for projections}\label{cwsdfr}

The group $G$ has a subgroup $G_-$ that is generated by the vector fields $\theta^{2k+1}\partial_\theta$.
This is the subgroup of $G$ that preserves a projection.  Accordingly,
supermanifolds $S$ with a chosen isomorphism $\rho$ as in eqn. (\ref{zolg})
and that are  projected but not necessarily split are labeled by a class in $H^1(M,G_-)$. 

In general, a supermanifold $S$ modeled on $M,V$ is projected if and only if the class  $x_S\in H^1(M,G)$ that represents it is in the image of $H^1(M,G_-)$.
If the subgroup $G_-$ of $G$ were normal, we would have an exact sequence $H^1(M,G_-)\to H^1(M,G)\to H^1(M,G/G_-)$
and then $S$ would be projected if and only if  $x_S$ maps to zero in $H^1(M,G/G_-)$.  
But  $G_-$ is not normal, and it does not
appear to be possible to express the obstruction to  a projection in such simple terms. 

Instead, we define $G^i_-$ to be the subgroup of $G$ generated by $G^i$ and $G_-$. This gives a sequence of subgroups
\[
G  := G^2_- \supset G^3_- =  G^4_- \supset G^5_- =  G^6_- \supset G^7_- =\dots
\]
In order for a class  $x_S \in H^1(M,G)$ to come  from $H^1(M,G_-)$, it is necessary and sufficient that, for each (odd) $i\geq 2$, this class should be in the image of $H^1(M,G^i_-)$. Is it possible to convert this to the vanishing of a series of obstructions? For that, we need each $G^i_-$ to be a normal subgroup of its predecessor $G^{i-2}_-$. This turns out to be true.

We can see the non-normality for $G_-$ as well as the normality for $G^i_-$  quite explicitly from the analytic description in terms of the Lie algebra $\mathfrak g$, which is generated by the vector fields on the split model  $S(M,V)$ that are schematically of the form
$\theta^{2k}\partial_x$ or $\theta^{2k+1}\partial_\theta$, $k\geq 1$, as in \eqref{zelbo} and \eqref{welbo}.
The Lie algebra $\mathfrak g_-$ of $G_-$ is the subalgebra generated by the vector fields  $\theta^{2k+1}\partial_\theta$.
The condition for a subgroup to be normal is that its Lie algebra should be an ideal in the ambient Lie algebra.
The Lie bracket of vector fields on $S(M,V)$  schematically gives formulas such as
\begin{equation}\label{yros} \left[\theta^2\partial_x,\theta^3\partial_\theta\right]=\theta^4\partial_x +\theta^5\partial_\theta.\end{equation}
(Recall that our schematic notation suppresses the coefficients, which depend on $x$.)
This shows that $\mathfrak g_-$ is not an ideal in $\mathfrak g$ since the bracket of  $\theta^2\partial_x \in \mathfrak g$ with $\theta^3\partial_\theta \in \mathfrak g_-$ has a non-vanishing $\theta^4\partial_x$ term, and is therefore not in $\mathfrak g_-$. 

The Lie subalgebra $\mathfrak g^i_-$ corresponding to $G^i_-$, for odd $i$, is generated by all the above vector fields except  $\theta^{2}\partial_x, \dots, \theta^{i-1}\partial_x$. The same equation \eqref{yros} shows that for $i \geq 5$, $\mathfrak g^i_-$ is not an ideal in $\mathfrak g$, since the bracket of  $\theta^2\partial_x \in \mathfrak g$ with $\theta^3\partial_\theta \in \mathfrak g^i_-$ has a non-vanishing $\theta^4\partial_x$ term, and is therefore not in $\mathfrak g^i_-$. (For $i=3$, $\mathfrak g^3_-$  {\em{is}} normal in $\mathfrak g$, as it equals $\mathfrak g^3$.)

On the other hand, $\mathfrak g^{i}_-$ is always an ideal in $\mathfrak g^{i-2}_-$. As above, $\mathfrak g^{i-2}_-$  is generated by all the vector fields except  $\theta^{2}\partial_x, \dots, \theta^{i-3}\partial_x$. The bracket of any two vector fields in $\mathfrak g_-$ is in $\mathfrak g_-$. The smallest mixed term with a $\partial_x$  factor from $\mathfrak g^{i-2}_-$ and a $\partial_\theta$ factor from $\mathfrak g^{i}_-$ is 
$\left[\theta^{i-1}\partial_x,\theta^3\partial_\theta\right]=\theta^{i+1}\partial_x +\theta^{i+2}\partial_\theta,$ 
which is in $\mathfrak g^i_-$, as is the remaining bracket 
$\left[\theta^{i-1}\partial_x,\theta^{i+1}\partial_x\right]=\theta^{2i}\partial_x.$ 

We conclude that each $G^{2k+2}_- = G^{2k+1}_-$ is a normal subgroup of its predecessor $G^{2k}_- = G^{2k-1}_-$.  As in \eqref{proz}, if a given class $x_S$ is in the image of $H^1(G^{2k}_-)$ for some  $k$, then to decide whether it is in the image of $H^1(G^{2k+1}_-)$, we look at the exact sequence
\begin{equation}\label{projz} H^1(M,G^{2k+1}_-)\to H^1(M,G^{2k}_-) \stackrel{\omega^-}{\to} H^1(M,  T_{+}M \otimes \exterior{2k}V^{\vee}).\end{equation}
The obstruction for a class  $\varphi _{2k-1}^- \in H^1(M,G^{2k}_-)$ to come from $H^1(M,G^{2k+1}_-)$
is that its image  
\[
\omega_{2k}^- = \omega^-(\varphi_{2k-1}^-) \in H^1(M,T_{+}M \otimes \exterior{2k}V^{\vee})
\]
must vanish.  We may as well call this class $\omega_{2k}^-$   the $2k$-th obstruction class for projectedness of $S$. It has properties, and limitations, analogous to those of the $\omega_i$. By construction, it depends only on the odd 
$\varphi_{2k-1}^-$, and these are the same as the even $\varphi_{2k-2}^-$. When both are defined, $\omega_{2k}^-$ and $\omega_{2k}$ clearly agree. This is the sense in which
$\omega_{2k}$ can be made to depend only on the previous even choices.

Note in particular that $\omega_2^- = \omega_2$, so the vanishing of $\omega_2$ is a necessary condition for a projection. (This follows more directly  from the fact that  $G_-$ is contained in $G^3$.) This is the only necessary criterion for projectedness that we will use in the present paper.

\subsubsection{\it Some immediate applications}

\begin{corollary}
Any $C^\infty$ supermanifold $S$ is split.
\end{corollary}
\begin{proof}
A $C^\infty$ locally-free sheaf is fine, hence its $H^1$ vanishes.
\end{proof}

\vspace {0.2 in}

\begin{corollary} \label{m|1}
Any supermanifold $S$ of dimension $(m|1)$ is split.
\end{corollary}
\begin{proof}
$(\exterior{i}V^{\vee})_0 =0, \quad i \ge 2$.
\end{proof}

\vspace {0.2 in}

\begin{corollary} \label{m|2}
A supermanifold $S$ of dimension $(m|2)$ is determined by the triple $(M,V,\omega)$, where 
$\omega=\omega_2 \in H^1(M,\Hom (\exterior{2}T_-, T_+))$, and any such triple arises from some $S$.  A supermanifold of dimension $(m|2)$
is projected if and only if it is split.\end{corollary}

The last statement reflects the fact that $\omega_2$, which is the only obstruction to a splitting when the odd dimension is 2, is also
an obstruction to a projection.

Similarly,

\begin{corollary} \label{m|}
For any supermanifold $S$, the following conditions are equivalent: $\omega_2(S)\not=0$; $S^{(2)}$ is not split; $S^{(2)}$ is not projected.  Moreover, if $\omega_2(S)\not=0$,
then $S$ is not projected.\end{corollary}

Indeed, with $M,V$ given, $S^{(2)}$ is classified up to isomorphism by $\omega_2(S)$, which obstructs both a splitting and a projection of $S^{(2)}$.  
A projection of $S$  could be restricted to a projection of $S^{(2)}$, so $\omega_2(S)\not=0$ implies that $S$ is not projected.

One can easily construct explicit examples of non-split supermanifolds, beginning in dimension $(1|2)$. For example, 
 a non-degenerate conic in $\PP^{2|2}$, e.g.
\begin{equation}
x^2  + y^2  + z^2 + \theta^1 \theta^2 =0,
\end{equation}
is non-split.
The obstruction class $\omega_2$ for this supermanifold is evaluated in \cite{Manin}. We will see another explicit example in section \ref{local non-splitness}.

\vspace {0.2 in}
\begin{lemma} \label{FF}
The space $\Sp(S)$ of all splittings of a split supermanifold $S$ has the structure of an iterated fibration by affine spaces.
The action of the automorphism group $\Aut(S)$ preserves this fibration and induces an affine action on the affine fibers.
(Similarly if  $S$ is replaced by $S^{(i)}$ for some $i$.) 
\end{lemma}
\begin{proof}

Let $M,V$ be the reduced space and odd tangent bundle of $S$.
The  group $\Aut(S)$ acts on the manifold $M$, and $V$ is an $\Aut(S)$-equivariant vector bundle on $M$. 
Let $S_0:=S(M,V)$ be the split version of $S$. 
It inherits an action of $\Aut(S)$, and the distinguished splitting is $\Aut(S)$-invariant.

Recall that $\Sp(S) := \Isom_{M,V}(S_0,S)$ denotes the variety of splittings of the split supermanifold $S$, i.e. isomorphisms of $S_0$ with $S$ that act trivially on $M$ and $V$.
The actions of $\Aut(S)$ on $S,S_0$ induce an action of $\Aut(S)$ on  $\Sp(S)$. 
(This action combines the $\Aut(S)$ action on $S$ with the inverse action on $S_0$.)

Similarly, we have the variety 
$\Sp(S^{(i)})^S = \Sp(S^{(i)})^S_{ S^{(1)}}$ 
of those splittings of $S^{(i)}$ that lift to $S$ (and restrict to the identity on $S^{(1)}$, as we always require.) 
The iterated fibration is:

\begin{equation} \label{fkyrd}
\Sp(S) \ldots \to   \Sp(S^{(i)})^S \to \Sp(S^{(i-1)})^S \to \ldots
\end{equation}

For $S_0$ itself, $\Sp(S_0)$ is a group, isomorphic to $H^0(G)$. The distinguished splitting, corresponding to the unit element, is a fixed point of the above natural $\Aut(S)$-action.  As a special case of \eqref{dwsrxe}, the iterated fibration \eqref{fkyrd} becomes:

\begin{equation} \label{fkzrd}
H^0(G) \ldots \to   H^0(G)/H^0(G^{i+1}) \to H^0(G)/H^0(G^{i}) \to \ldots
\end{equation}

The typical step here is a surjective group homomorphism:
\[
\pi: H^0(G)/H^0(G^{i+1}) \to H^0(G)/H^0(G^{i}))
\]
whose kernel is $H^0(G^{i}))/H^0(G^{i+1})$ and whose other fibers are cosets of this kernel. The sequence 
$ 0 \to G^{i+1} \to G^{i} \to G^{i}/G^{i+1} \to 0$
shows that this kernel is the image of the group homomorphism
\[
h: H^0(G^{i})) \to H^0(G^{i}/G^{i+1}).
\]
This image is clearly a subgroup of the vector space 
$$H^0(G^{i}/G^{i+1}) = H^0(M,  T_{(-)^{i}}M \otimes \exterior{i}V^{\vee}).$$ 
It is also invariant under homotheties: if $h(x)=y$, then (because the exponential map is bijective) $x$ can be extended to a 1-parameter subgroup of $H^0(G^{i}))$, and $h$ of this 1-parameter subgroup is a 1-parameter subgroup of the vector space $H^0(G^{i}/G^{i+1})$, i.e. it is a line. So $\ker(\pi)$, which is the image of $h$, is a subgroup closed under homotheties, i.e. it is a vector subspace of $H^0(G^{i}/G^{i+1})$, and $\Aut(S)$ acts on it linearly. The general fibers of $\pi$ are cosets of this, i.e. affine spaces with affine $\Aut(S)$ action.

Compare this with $S$, which is split but perhaps not $\Aut(S)$-invariantly so. Then 
\break $\Sp(S)$ is a $\Sp(S_0)$-torsor, and iterated fibration \eqref{fkyrd} is a torsor over iterated fibration \eqref{fkzrd}. (This means that the objects of  \eqref{fkyrd} are torsors over corresponding objects of  \eqref{fkzrd}, and the maps in   \eqref{fkyrd} are torsor maps compatible with the corresponding maps of groups in   \eqref{fkzrd}.)
There are induced actions of $\Aut(S)$ on all terms of the two iterated fibrations, and 
the crucial observation is that these $\Aut(S)$-actions are compatible with the above torsor structure. 
In particular, the fibers of $\Sp(S^{i})^S \to \Sp(S^{i-1})^S$ in the  iterated fibration \eqref{fkyrd} are affine spaces modeled on the vector space $\ker(\pi)$,  the $\Aut(S)$ action on them is affine, and the linearization of this action is the linear $\Aut(S)$ action on $\ker(\pi)$.

\end{proof}

\vspace {0.2 in}
\begin{corollary} \label{F}
Let $S$ be an algebraic supermanifold 
(i.e. the reduced space $M$ is a complex algebraic variety and $V$ is an algebraic  vector bundle),
and let $F$ be a reductive subgroup of its automorphism group $\Aut(S)$. 
If  $S$, or $S^{(i)}$ for some $i$, is  split, then the splitting can be chosen to be $F$-invariant.   \end{corollary}
\begin{proof}\footnote{We thank P. Deligne for help with this proof.}
An $F$-invariant splitting of $S$ is a fixed point $p$ of the $F$-action on 
\break $\Sp(S)$.
We find a fixed point $p^i$ of the action of $F$ on  $\Sp(S^{(i)})^S $ for each $i$, inductively. 
For sufficiently high $i$ this will give us the desired fixed point of the action of $F$ on  $\Sp(S)$.
For $i=1$ there is nothing to do, since  
$\Sp(S^{(1)})^S$ 
is a single point. 
So we assume inductively that we have a fixed point $p^{i-1}  \in \Sp(S^{(i-1)})^S$ 
and we need to lift it to a fixed point $p^i \in \Sp(S^{(i)})^S$ in the fiber above $p^{i-1}$, 
which was seen in Lemma \ref{FF} to be an affine space with an affine action of $F$. 
The inductive step, and the proof of the Corollary, now follow from the observation that 
an algebraic affine action of a reductive group  on an affine space must have a fixed point. 

If $F$ is finite, the existence of fixed points is clear: 
the average of a finite set of points in an affine space (in characteristic zero) is well defined, and 
the average of the points in any $F$-orbit is clearly $F$-invariant. 
When the vector space in question, 
$A:=H^0(M,T_{(-)^i}M \otimes \exterior{i}V^{\vee}),$
is  finite dimensional, e.g. when $M$ is compact, 
the reductive case follows similarly, 
by averaging with respect to an invariant measure on a compact form $F_c$ of the group. 
Alternatively, the affine action of $F$ on the vector space $A$ 
extends to a linear action on the vector space $B:=A \oplus C$, where $C=\CC$, 
and the action on $B$ commutes with the projection to $C$. 
(The affine action is recovered on the slice $A^1$ above $1 \in \CC$,
while its linearization is realized on the slice $A^0$ above $0 \in \CC$.)
This $B$ contains $A=A^0$, with a 1-dimensional quotient $C$ on which $F$ acts trivially. 
In the finite dimensional case, 
complete reducibility of the action  of our reductive  $F$ gives an $F$-invariant embedding of $C$ in $B$, 
and points of its image are fixed by $F$.

It is true in general that any algebraic linear representation of $F$, such as its action on 
$B=A \oplus C$, is the union of its finite dimensional subrepresentations. 
The existence of fixed points for the affine action in the general case then follows: 
choose any point $p$ in the affine fiber  $A^1$, think of it as a point of $B$ above $1 \in \CC$, 
and find a finite dimensional subrepresentation of $B$ containing this point. 
Using the finite dimensional case, we find a fixed point $p^i$ in this subrepresentation.

In our case, with $A=H^0(M,T_{(-)^i}M \otimes \exterior{i}V^{\vee}),$ 
we can use some geometry to show directly that 
any $\rho \in A$ is contained in some finite dimensional subrepresentation.
For this, choose a compactification $\bar{M}$ of $M$. 
(This can always be done, by a well known result of Nagata \cite{Nagata},
which has been modernized and improved in  \cite{Conrad} by Conrad, following Deligne.)
The vector bundle $T_{(-)^i}M \otimes \exterior{i}V^{\vee}$ on $M$
extends to a coherent sheaf $F$ on  $\bar{M}$. 
Our space  $A=H^0(M,T_{(-)^i}M \otimes \exterior{i}V^{\vee})$
is the union of its finite dimensional subspaces $A_j$ 
consisting of meromorphic sections of $F$ with a pole of order $\leq j$ on the boundary of $\bar{M}$. 
(These $A_j$ may not be $F$-invariant.)
Now given any $\rho \in A$, the $F$-action on $\rho$ gives a map $f: F \to A$.
The inverse images $f^{-1}(A_j)$ are Zariski closed subsets of $F$, 
since each is given locally by the vanishing of a set of Laurent coefficients.
They form an increasing sequence of Zariski closed subsets that cover $F$, 
so one of them must equal $F$. 
(Over an uncountable field, if a variety is the union of a countable collection of its subvarieties, 
then it is also the union of a finite subset of these. Topologically,  this is the Baire category theorem.)
It follows that the span of the $F$-orbit of $\rho$ is contained in some $A_j$, 
hence this span is a finite dimensional $F$-invariant subspace containing $\rho$ as desired.
\end{proof}

\begin{corollary} \label{CorCovering} 
Let $\pi: \wt{S} \to S$ be a finite covering map of supermanifolds. 
If $S$, or $S^{(i)}$ for some $i$, is projected or split, then likewise $\wt{S}$, or $\wt S^{(i)}$, is projected or split.
Conversely (in characteristic zero, as always assumed in this paper), if $\wt{S}$ or $\wt S^{(i)}$ is split, then so is $S$ or $S^{(i)}$.
\end{corollary} 
\begin{proof}The first part is an immediate consequence of the construction in section \ref{Covers} of a covering map of supermanifolds.
In this construction, if $S$ or $S^{(i)}$ for some $i$ is split or projected, then the cover $\wt{S}$ or $\wt{S}^{(i)}$ immediately acquires the
same property.

For the converse part,  if there is a finite group $F$ that acts freely on  $\wt{S}$ with quotient $S$, we can appeal to  Corollary \ref{F}:  the splitting of $\wt{S}$ can be replaced by one which is $F$-invariant, and the latter is a pullback from a splitting of $S$. In general, we can always find a covering $\bar{S} \to  \wt{S}$, a freely acting group $F$, and a subgroup $G \subset F$ such that  $\wt{S}$ is the quotient of $\bar{S} $ by the action of $G$, while $S$ is the quotient of $\bar{S} $ by the action of $F$. To construct $\bar S$, we simply construct a Galois cover $\bar M$ of $M$ lying over $\wt{M}$ (the
reduced space of $\wt S$), and let
$\bar S$ be the corresponding cover of $S$.  The splitness of $\wt{S}$ implies the splitness of  $\bar{S} $ by the first part of this proof, and the splitness of $S$ then follows from Corollary \ref{F}.

\end{proof}

  We do not know whether Corollary \ref{F} and the converse part of Corollary \ref{CorCovering}  hold  with ``split'' replaced by ``projected.''
  
  The special case of Corollary \ref{CorCovering} that we will use later is the following:

\begin{corollary}\label{CoroCovering}Let $\pi:\wt{S}\to S$ be a finite covering map of supermanifolds (in characteristic zero).
If $\omega_2(S)\not=0$, then $\omega_2(\wt{S})\not=0$, so $\wt{S}$ is not projected.
\end{corollary}

\begin{proof} The hypothesis that $\omega_2({S})\not=0$ is equivalent to ${S}^{(2)}$ being non-split.  By Corollary \ref{CorCovering}, therefore,
$\wt{S}^{(2)}$ is not split.  This in turn is equivalent to $\omega_2(\wt{S})\not=0$, and, as $\omega_2$ obstructs a projection, it follows that $\wt{S}$ is not projected.

\end{proof}
One can also give a direct proof of this Corollary. (We thank one of the referees for pointing this out.)
For a local isomorphism $\pi:\wt{S}\to S$, the structure sheaf $\OO_{\wt{S}}$ is the pullback of $\OO_{{S}}$. It follows that $\omega_2({\wt{S}}) = \pi^* \omega_2({{S}}) $.  But the composition of $\pi^*$ with the trace map is multiplication by the degree of $\pi$, so $\pi^*$ is injective and the result follows.

\subsubsection{\it Submanifolds}

Suppose that $S$ is a supermanifold modeled on $M,V$.  Let $M'$ be a submanifold of $M$, and let $V'$ be a subbundle of $V|_{M'}$.
In general, a submanifold  $S'$ of $S$ with reduced space $M'$ and odd tangent bundle $V'$ does not exist.  
To analyze the obstruction to this, we start with the split supermanifold $S(M,V)$, which has $S(M',V')$ as a submanifold.
The sheaf of groups $G$ has a subsheaf $G_*$ consisting of automorphisms of $S(M,V)$ that restrict to automorphisms of $S(M',V')$.  
The condition that a  supermanifold $S$ modeled on $M,V$ contains a submanifold $S'$ modeled on $M',V'$ is that the class in
$H^1(M,G)$ that represents $S$ should be the image of a class in $H^1(M,G_*)$.   As in the discussion of projections,
there would be an easy criterion for this if $G_*$ were normal, but this is not the case.  

However, we can partially reduce to the normal case if we replace $S$ with $S^{(2)}$ and ask whether $S^{(2)}$ contains a subspace $S'{}^{(2)}$
modeled on $S(M',V')^{(2)}$.  For this, we replace $G$ by the sheaf of abelian groups $G/G^3={ \mathit{Hom}}(\exterior{2}T_-M,T_+M)$
and $G_*$ by $G_*/G_*^{3}$, which is the subsheaf of $\mathit{Hom}(\exterior{2}T_-M,T_+M)$ consisting of homomorphisms that, when
restricted to $M'$, map $\exterior{2}T_-{M'}$ to $T_+M'$.  We call this sheaf $\mathit{Hom}^*(\exterior{2}T_-M,T_+M)$.  
The quotient of the two is the sheaf $\mathit{Hom}(
\exterior{2}T_-M',N)$ supported on $M'$, where $N$ is the normal bundle to $M'$ in $M$.
So we get an exact sequence in cohomology
\begin{align}\label{poiu}H^1(M,{\mathit{ Hom}}^*(\exterior{2}T_-M,T_+M))&\to H^1(M,\mathit{Hom}(\exterior{2}T_-M,T_+M) \cr &\stackrel{b}{\to}H^1(M',\mathit{Hom}(
\exterior{2}T_-M',N),\end{align}
leading to a necessary condition for existence of $S'$:

\begin{corollary} \label{m|a2}
A necessary condition for a supermanifold $S$ modeled on $M,V$ to contain a submanifold $S'$ modeled on $M',V'$ is that, in the notation of \eqref{poiu},
\[
 b(\omega_2)=0 \in H^1(M',\mathit{Hom}(\exterior{2}T_-M',N)).
\]
\end{corollary}

 From the definition of $G_*$ in terms of automorphisms of $S(M,V)$ that restrict to automorphisms of $S(M',V')$, there is a tautological
 map from $G_*$ to the sheaf $G'$ of automorphisms of $S(M',V')$ that act trivially on $M'$ and its normal bundle.   Thus, $G_*$ maps
 to both $G$ and $G'$, 
 \begin{equation}\label{surely}\begin{matrix} G_* & \stackrel{u}{\to} G \cr
                                                                       r\downarrow & \\
                                                                        G', & \end{matrix}\end{equation} 
leading to corresponding maps in cohomology
 \begin{equation}\label{purely}\begin{matrix} H^1(M,G_*) & \stackrel{u}{\to} H^1(M,G) \cr
                                                                       r\downarrow & \\
                                                                        H^1(M',G'). & \end{matrix}\end{equation} 

\remark{\label{helpful} As already noted, if $S'$ exists,
the class $x_S\in H^1(M,G)$ that represents $S$ is the image of some $x^*_{S}\in H^1(M,G_*)$.  The map from $G_*$ to $G'$ just
restricts an element of $G_*$ to its action on $S(M',V')$, so
the image of $x^*_S$
in $H^1(M',G')$ is the class $x_{S'}$ that represents $S'$.}

  If we abelianize by replacing $G_*$, $G$, and $G'$ by their quotients by $G_*^{3}$, $G^3$, and $G'^3$, we
 can complete the picture (\ref{surely}) to a commuting square:
 \begin{equation}
   \begin{matrix} \mathit{ Hom}^*(\exterior{2}T_-M,T_+M) & \stackrel{u}{\to} &  \mathit{ Hom}(\exterior{2}T_-M,T_+M)  \cr
          r\downarrow  &    &     \iota\downarrow    \cr
                \left.  \mathit{Hom}(\exterior{2}T_-M',T_+M') \right.      & \stackrel{j}{\to} &  \left.\mathit{ Hom}(\exterior{2}T_-M',T_+M|_{M'}) \right  .      \end{matrix}       \end{equation}  
Here $u$ and $r$ are the linearizations of the corresponding maps in (\ref{surely}).
                  The sheaves on the top row are sheaves on $M$ and the sheaves on the bottom row are sheaves on $M'$.
The vertical maps are restrictions from $M$ to $M'$, followed by restriction from $\exterior{2}T_-M$ to $\exterior{2}T_-M'$.
Finally,  $j$ comes from the inclusion $T_+M'
\subset T_+M|_{M'}$.  In cohomology we get:
   \begin{equation}\label{yelb}
   \begin{matrix} H^1(M, \mathit{ Hom}^*(\exterior{2}T_-M,T_+M)) & \stackrel{u}{\to} &   H^1(M,  \mathit{ Hom}(\exterior{2}T_-M,T_+M))   \cr
          r\downarrow  &    &     \iota\downarrow    \cr
                  H^1(M',\mathit{Hom}(\exterior{2}T_-M',T_+M') )       & \stackrel{j}{\to} & H^1(M',    \mathit{ Hom}(\exterior{2}T_-M',T_+M)) .         \end{matrix}       \end{equation}         
 
 The linearization of Remark \ref{helpful} says that there is a class  $\omega_2^*(S)$ in the upper left corner of the square (\ref{yelb})
 that maps horizontally to $\omega_2(S)$ and vertically to $\omega_2(S')$.  So commutativity of the square implies that
 $j(\omega_2(S'))=\iota(\omega_2(S))$.  This statement is our Compatibility Lemma:
 
 \begin{corollary} \label{mila}
If $S$ is a supermanifold with submanifold $S'$, then the classes $\omega_2(S')$ and $\omega_2(S)$ are compatible in the
sense that $j(\omega_2(S'))=\iota(\omega_2(S))$.
\end{corollary}

Particularly important is the case that the normal sequence decomposes:

\vspace {0.2 in}
\begin{corollary} \label{submanifold splits}
Let $S$ be a supermanifold, and $S'\subset S$ a submanifold, with reduced spaces $M'\subset M$, such that the normal
sequence of $M'$ decomposes: 
$ T{M}|_{M'} =  T_{M'} \oplus N$, where $N$ is the (even) normal bundle. If $\omega_2(S')\not=0$,  then also $\omega_2(S)\not=0$
and $S$ is not projected.\end{corollary}
\begin{proof}
This follows from Corollary \ref{mila} and the fact that, when the normal sequence decomposes, the map $j$ is injective.
\end{proof}
 
Are there analogous results with $\omega_2$ replaced by $\omega_j$, for some $j>2$?  Clearly a necessary hypothesis
is that $\omega_i(S)=\omega_i(S')=0$, $i<j$, so that $\omega_j(S)$ and $\omega_j(S')$ are defined.   This would imply that there are splittings
$S^{(j-1)}\to M$ and $S'^{(j-1)}\to M'$.  One actually needs a stronger hypothesis to make the argument.  
One needs to know that the class in $H^1(M,G)$ associated
to $S$ is a pullback from $G_*^{j}$, which implies that there is a splitting $S^{(j-1)}\to M$
that restricts to a splitting $S'^{(j-1)}\to M'$.  Under this hypothesis, the analog of Corollary \ref{mila} holds, with essentially the same proof.

Rather than an embedding, consider any map of supermanifolds  $f: S' \to S$, with reduced manifolds 
${S}_{{\mathrm{red}}} = M$ and $S'_{{\mathrm{red}}} = M'$. The differentials:
\[
d_{\pm}: T_{\pm}{S'}  \to f^*(T_{\pm}{S})
\]
induce maps on cohomology groups:
\[
\iota: H^1(M,  (T_{\pm}{S}) \otimes \exterior{i}   (T_{-}{S}) ^{\vee})
\to
H^1(M',  (T_{\pm}{S}) \otimes \exterior{i}  (T_{-}{S'} ) ^{\vee})
\]
and
\[
j: H^1(M',  (T_{\pm}{S'} ) \otimes \exterior{i}   (T_{-}{S'} ) ^{\vee})
\to
H^1(M',  (T_{\pm}{S}) \otimes \exterior{i}  (T_{-}{S'} ) ^{\vee}).\]

Perhaps surprisingly, Corollary \ref{mila}   holds for an arbitrary map:
\vspace {0.2 in}

\begin{corollary} \label{X,Y general} For any map  of supermanifolds $f: S' \to S$,
the classes $\omega_2(S')$ and $\omega_2(S)$  are compatible, in the sense that:
\[
 j ( \omega_2(S')  )=\iota( \omega_2(S)).
\]
\end{corollary}
\begin{proof}
The map $f: S' \to S$ factors through the graph embedding $\Gamma: S' \to S \times S'$ and the projection on the first factor. The result is clear for the projection and behaves well under composition, so this corollary follows from the  embedding case which is Corollary \ref{mila}.
\end{proof}

\vspace {0.2 in}

\section{Super Riemann Surfaces}\label{SRSs}

\subsection{Basics}\label{basics}
A {\em Super Riemann Surface}\footnote{The term ``super Riemann surface'' was introduced in \cite{F}.   The same objects
were called ``superconformal manifolds'' in \cite{BS1}.  In our terminology, a super Riemann surface, or SRS for short, 
is a $(1|1)$ complex supermanifold with a superconformal structure.}  is a pair $\S = (S,\D)$ where $S = (C, \OO_S)$ is a complex analytic supermanifold of dimension $(1|1)$ and $\D$ is an everywhere non-integrable odd distribution, $\D \subset TS$.  Recall that the square of an odd vector field $v$ is an even vector field $v^2 = \frac{1}{2} \{v,v\}$. The distribution $\D$ generated by $v$ is said to be {\em integrable} if $v^2 \in \D$, and 
{\em  everywhere non-integrable} if $v^2$ is everywhere independent of $v$.  In the latter case, $v$ and $ v^2$ span the full tangent bundle $TS$, and
thus the nonintegrable distribution $\D$ is actually part of an exact sequence
\begin{equation}\label{exseq} 0\to \D\to TS\to \D^2\to 0. \end{equation}
The everywhere nonintegrable odd distribution $\D$ endows the $(1|1)$ supermanifold $S$ with what is called a superconformal structure.

As is the case for supermanifolds in general, 
interesting phenomena usually concern not a single Super Riemann Surface 
but a family of Super Riemann Surfaces over a base that is itself a supermanifold.
This means a family of complex analytic supermanifold of dimension $(1|1)$
together with a subbundle of the relative tangent bundle, of rank $(0|1)$ and 
everywhere non-integrable.

\begin{lemma}\label{egh}
Locally on a super Riemann surface $S$
one can choose coordinates $x$ and $\theta$, referred to as superconformal coordinates, such that $\D$ is generated by the odd vector field
\begin{equation} \label{SRSD}
v :=\frac{ \partial }{\partial \theta }+  \theta\frac{ \partial }{\partial x}. \end{equation}  More specifically, if $x$ is any even function on $S$ that reduces mod
nilpotents to a local parameter on the reduced space, then there is (locally in the analytic or etale topologies) an odd function $\theta$ on $S$ such that $x|\theta$ are superconformal 
coordinates.
\end{lemma}
\begin{proof}
In general, in any local coordinate system $x|\theta$, an odd vector field is of the form $v=a(x|\theta)\partial_\theta+b(x|\theta)\partial_x$, where $a$ is
an even  function and $b$ is an odd one.  (In general, $a$ and $b$ depend on additional odd and/or even parameters as well as on $x|\theta$.)
As $b$ is odd, it is nilpotent. The definition of an odd distribution is such that $v$ defines an odd distribution where and only where
$a\not=0$.  (On a supermanifold, to say that an even quantity is nonzero is taken to mean that it is invertible.)  The condition that 
$v$ and $v^2$ generate $TS$ implies that if we write $b(x|\theta)=b_0(x)+\theta b_1(x)$, then $b_1\not=0$ everywhere.
Given the conditions that $a$ and $b_1$ are everywhere nonzero, 
it is elementary to find a change of variables that locally puts $v$ in the form given in eqn. (\ref{SRSD}).  The last statement in the lemma
holds because in this change of variables, it is only necessary to change $\theta$, not $x$.
\end{proof}

Applying  the vector field $v$ of \eqref{SRSD} to a function $f(x) + \theta g(x)$ gives $g(x) + \theta f'(x)$, so applying it twice gives $f'(x) + \theta g'(x)$. In other words,
\begin{equation}\label{szx}
v^2= \frac{\partial }{\partial x}.
\end{equation}

Since a super Riemann surface has dimension $(1|1)$, a divisor in it has dimension $(0|1)$.  One way to specify a subvariety of
dimension $(0|1)$ in any complex supermanifold is to give a point $p$ and an odd tangent vector at $p$.  There is then a unique
subvariety of dimension $(0|1)$ that passes through $p$ with the given tangent vector.  In the case of a super Riemann surface $\S$,
let  $p$ be the point defined by equations $x=x_0$, $\theta=\theta_0$.  For this to make sense, 
with $\theta_0\not=0$,
we must work over a ring with odd elements; the constant $\theta_0$ is an odd element of the ground ring.  The fiber at $p$ of
the distribution $\D$ gives an odd tangent direction at $p$, and this determines a divisor $D$ passing through $p$.  
Concretely, with $\D$ generated as in (\ref{SRSD}), 
$D$ is given in parametric form by $x=x_0+\alpha\theta$, $\theta=\theta_0+\alpha$, where $\alpha$ is an odd parameter.
Alternatively, $D$ is given by the equation
\begin{equation}\label{medlo} x=x_0-\theta_0\theta. \end{equation}
In a general supermanifold of dimension $(1|1)$, a divisor that in some local coordinate system is defined by an equation of this kind
is called a minimal divisor.  (Algebraically, one might call it a prime divisor.)  Since the parameters $(x_0,\theta_0)$
in the equation were the coordinates of an arbitrary point $p\in \S$, this construction gives 
a natural 1-1 correspondence between points and minimal divisors on a super Riemann surface $\S$.  On a general complex
supermanifold $S$ of dimension $(1|1)$, there is no such natural correspondence between points and minimal divisors.  (Rather,
to a $(1|1)$ supermanifold $S$, one associates a dual supermanifold $\widehat S$ that parametrizes minimal divisors in $S$; the relationship
between $S$ and $\widehat S$ is actually symmetrical.)

\remark{For a concrete special case of the relationship between points and minimal divisors on a super Riemann surface, 
consider, in local superconformal
coordinates $x|\theta$, the divisor $x=0$.  The equation $x=0$ is of the form (\ref{medlo}) for $x_0=\theta_0=0$, so the distinguished
point associated to the divisor $x=0$ is given by $x=\theta=0$.  On a general  $(1|1)$ supermanifold, one could make an automorphism
$x\to x,$ $\theta\to\theta+\alpha$, with $\alpha$ an odd parameter, and thus there would be no distinguished point on the divisor
$x=0$.  On a super Riemann surface, there is no such automorphism preserving the superconformal structure. 

 The example that we have
just described is typical in the sense that any minimal divisor on a super Riemann surface takes the form $x=0$ in some system $x|\theta$
of local superconformal coordinates.  This assertion follows from the last remark in Lemma \ref{egh}, or alternatively from the fact
that if $x|\theta$ are local superconformal coordinates, then so are $x-a+\eta\theta|\theta-\eta$, where here $a$ and $\eta$ are arbitrary
even and odd parameters.  }\label{funremark}

\subsection{Moduli}\label{moduli}

\mbox { }

In ordinary bosonic algebraic geometry, one encounters the coarse moduli space $\M_g$ of genus $g$ curves 
(or: Riemann surfaces) as well as the moduli stack $M_g$. The stack is characterized by the collection of maps to it: a 
map $\phi: B \to M_g$ from a variety $B$ to the moduli stack $M_g$ is specified by a family of genus $g$ curves 
parametrized by $B$. Such a family determines a map $\bar{\phi}$ of $B$ to the moduli space $\M_g$ too; this is 
just the composition of $\phi$ with the natural map $\pi: M_g \to \M_g$. But not every such map $\bar{\phi}$ 
arises from a family over $B$; it does if and only if it factors through $M_g$. For example, the identity 
map of $\M_g$ to itself does not arise from a family: there is no universal curve over $\M_g$. The problem 
arises from the automorphisms. Every curve of genus $\leq 2$, and some curves in any genus, 
have non-trivial automorphisms. At the corresponding points of moduli, the map $\pi: M_g \to \M_g$ is not a 
local isomorphism: it looks locally like a quotient by the automorphism group.

For super Riemann surfaces,  the 
problem is more severe. Any split super Riemann surface $\S$, e.g. any SRS over a point 
(as opposed to a family of SRS's), has a non-trivial automorphism which is the identity on $\S_{\red}$ and 
acts as $-1$ in the odd direction. In local coordinates, this is $x \mapsto x, ~~\theta \mapsto - \theta$. So super Riemann surfaces with
non-trivial automorphisms are dense in any family of super Riemann surfaces, 
and unlike the case of ordinary Riemann surfaces, one cannot avoid this by taking a finite cover
of the moduli space (for example, by  fixing a level structure).  
So the stacky nature of the moduli problem is more essential for super Riemann surfaces than for ordinary ones and 
what we will loosely call super moduli space and denote as $\SM_g$ must be properly understood not as a supermanifold
but as the supermanifold analog of a stack.  We return to this in section \ref{stacks}, but for now we consider local questions.

By Corollary \ref{m|1}, a $(1|1)$ supermanifold $S$ defined over a field  (as opposed to one that is defined over a ring with odd elements)
is split, so it is specified by a pair $(C,V)$ where $C$ is an ordinary Riemann surface and $V$ is the  analytic line bundle on $C$ underlying $\D$. We refer to the genus of $C$ also as the genus of $\S$. The  calculation \eqref{szx}  shows that $V^{\otimes 2} \sim T_C$, so $V$ (or rather, its dual) is a spin structure  on $C$. An ordinary Riemann surface with a choice of spin structure is often called a spin curve, and thus the reduced space
$C$ of a super Riemann surface $\S$ is a spin curve.  It follows also that the reduced space  of the moduli space $\SM_g$ of super
Riemann surfaces of genus $g$ 
is the moduli space  (or rather stack)  $\sM_g$ of spin curves.  Just like $\sM_g$, $\SM_g$ has two components, corresponding
to even and odd spin structures.

For a super Riemann surface constructed from a pair $(C,V)$, we will denote $V$ as  ${T_C}^{1/2}$, and 
its $n$-th power by either $T_C^{\otimes \frac{n}{2}} = T_C^{ \frac{n}{2}}$ or equivalently $K_C^{\otimes (- \frac{n}{2})} = K_C^{ - \frac{n}{2}}$.

As in ordinary algebraic geometry, to
understand the deformation theory of  a super Riemann surface $\S$, we first must consider the automorphisms.
Locally, the superconformal vector fields, i.e. infinitesimal automorphisms of $\S$, are given by vector fields that preserve the distribution $\D$.  In
superconformal coordinates, a short calculation shows that an even superconformal vector field takes the form
\begin{equation}\label{evensup} f(x)\frac{\partial}{\partial x}+\frac{f'(x)}{2}\theta\frac{\partial}{\partial\theta},\end{equation}
while an odd one takes the form
\begin{equation}\label{enrup} -g(x)\left(\frac{\partial}{\partial\theta}-\theta\frac{\partial}{\partial x}\right).\end{equation}
We stress that $f$ and $g$ are functions of $x$ only, and not $\theta$.
One defines a subsheaf $W$ of the tangent bundle $TS$ whose local sections are of the form (\ref{evensup}) and (\ref{enrup}).
It is called the sheaf of superconformal vector fields.
Since it is defined by the condition of preserving the superconformal structure of $\S$, $W$ is a sheaf of $\Z_2$-graded Lie algebras.
From this point of view, $W$ is not naturally a locally-free sheaf (multiplying a superconformal vector field by a function of $x$ and $\theta$,
or even a function of $x$, does not in general give a new superconformal vector field).  However, forgetting its structure
as a sheaf of graded Lie algebras, $W$ can be given the structure of a locally-free sheaf.  For this, we just think of $W$ as a subsheaf
of the sheaf of sections of $TS$, so that in view of the exact sequence (\ref{exseq}), $W$ can be projected to the sheaf of sections
of $\D^2$.  A short calculation in local superconformal coordinates  shows that this projection is an isomorphism,
so $W$ can be identified with the sheaf of sections of $\D^2$. Indeed, in local superconformal coordinates $x|\theta$,
$TS/\D\cong \D^2$ is generated by $\partial_x$, so a general section of $\D^2$ is $a(x|\theta)\partial_x$ for some function $a(x|\theta)$.  So
we must show that $a(x|\theta)\partial_x$ can in a unique fashion be written as a section of $W$ modulo $\D$.  This follows from the formula
\begin{equation}\label{fondo}\left(f(x)\partial_x+(f'(x)/2)\theta\partial_\theta\right)
-g(x)\left(\partial_\theta-\theta\partial_x\right)=(f(x)+2\theta g(x))\partial_x~~\mathrm{mod}~\D,\end{equation}
which shows that a general section of $W$ is $a(x|\theta)\partial_x$ mod $\D$, with 
$a(x|\theta)=f(x)+2\theta g(x)$.  

Just as in bosonic algebraic geometry, the first-order deformations of $\S$ are given by the first cohomology of $\S$ with values
in the sheaf of infinitesimal automorphisms.  So first-order deformations are given by $H^1(\S,W)$, or equivalently by $H^1(\S,\D^2)$. (Sheaf cohomology on a SRS $\S$ means cohomology of the same sheaf on the underlying supermanifold $S$, which in turn was defined to be the cohomology of the corresponding sheaf on the reduced space $C$.)
This gives the tangent space to the moduli space $\SM_g$ of super Riemann surfaces at the point corresponding to a super Riemann
surface $\S$:
\begin{equation}\label{tangentspace} T_{\S}\SM_g=H^1(\S,W)=H^1(\S,\D^2). \end{equation}
If $\S$ is split (by which we mean that the underlying supermanifold $S$ is split), we can make this more explicit.  For $\S$ split with 
reduced space $C$, the sheaf $W$ of superconformal vector fields is the direct sum of its even and odd parts $W_+$ and $W_-$,
where $W_+$ is the sheaf of sections of $T_C$ and $W_-$ is the sheaf of sections of $T_C^{1/2}$.  This enables us to identify the even
and odd tangent spaces $T_\pm \SM_g$.  Their fibers at the point in $\sM_g$ corresponding to $C$ are
\begin{align}\label{gensp} T_{+,\S}\SM_g &= H^1(C,T_C) \cr T_{-,\S}\SM_g&= H^1(C,T_C^{1/2}). \end{align}

\subsubsection{\it More on the moduli stack}\label{stacks}

A closer examination of the formula (\ref{gensp}) for the odd normal bundle to $\sM_g$ in $\SM_g$ leads to a better understanding
of why the ``stacky'' nature of the moduli problem is more central for  super Riemann surfaces than for ordinary ones.

Suppose that $B$ is an algebraic variety that parametrizes a family of curves with spin structure.
We denote this family as $\pi:X\to B$ and denote a fiber of this fibration as $C$.  By definition, since $B$ parametrizes
a family of curves with spin structure, each $C$ comes with an isomorphism class of line bundle $\L=T_C^{1/2}$ with
an isomorphism $\varphi:\L^2\cong TC$.  We define  $\T\to X$ to be the  relative tangent bundle, i.e. the 
tangent bundle along the fibers of $\pi:X\to B$. 
If there is a line bundle $\L\to X$ with an isomorphism $\varphi:\L^2\cong \T$, and such
that $\L$ restricted to each $C$ is isomorphic to the given $T_C^{1/2}$, then we call $\L$ a relative spin bundle.

However, in general, the existence of a relative spin bundle is obstructed.  The essential reason is that locally, after a line bundle
$\L$ with an isomophism $\varphi:\L^2\cong\T$ is chosen, we are still left with the group of automorphisms $\{\pm 1\}$ acting
on $\L$ without changing $\varphi$.  Thus, locally, the pair $\L$, $\varphi$ is unique up to isomorphism but not up to a unique isomorphism, and this leads to a global
obstruction.  One may cover $B$ with small open sets $\O_\alpha$ and choose for each of them a line bundle $\L_\alpha \to \pi^{-1}
(\O_\alpha)$ with an isomorphism $\varphi_\alpha:\L_\alpha^2\cong \T$.  Since locally the pair $\L_\alpha$, $\varphi_\alpha$
 is unique up to isomorphism,
on each $\O_\alpha\cap\O_\beta$, one can pick an isomorphism $\psi_{\alpha\beta}:\L_\alpha\to \L_\beta$ such that
$\varphi_\alpha=\varphi_\beta\circ(\psi_{\alpha\beta}\otimes \psi_{\alpha\beta})$.  This last condition determines $\psi_{\alpha\beta}$ uniquely
up to sign, but in general there is no natural way to fix the sign. In a triple intersection $\O_\alpha\cap\O_\beta\cap\O_\gamma$, 
set $\lambda_{\alpha\beta\gamma}=\psi_{\gamma\alpha}\psi_{\beta\gamma}\psi_{\alpha\beta}$.  In general, $\lambda_{\alpha\beta\gamma}=
\pm 1$.   If the signs of the local
isomorphisms $\psi_{\alpha\beta}$ can be chosen so that $\lambda_{\alpha\beta\gamma}=+1$ for all $\alpha,\beta,\gamma$, then the
$\L_\alpha$ can be glued together via the isomorphisms $\psi_{\alpha\beta}$ to make a relative spin bundle $\L$.  In general,
the $\lambda_{\alpha\beta\gamma}$ are a 2-cocycle representing an element  $\varrho\in H^2(B,\Z_2)$.  When $\varrho\not=0$,
this obstructs the existence of $\L$; in this situation, we can say that $\L$ exists not as an ordinary line bundle, but as a twisted
line bundle, twisted by the $\Z_2$ gerbe $\hat\varrho$ corresponding to $\varrho$. 
 In general, such an obstruction can arise even if $B$ is simply-connected, so it cannot be removed
by replacing $B$ by an unramified cover.   

Let us consider in this light the case that $B$ is the spin moduli space $\sM_g$ that parametrizes
pairs consisting of a curve $C$ with an isomorphism class of spin structure.   A universal spin curve $\pi:X\to \sM_g$ exists
if one suitably interprets $\sM_g$ as an orbifold or  stack, to account for the possibility that a curve $C$ may have non-trivial automorphisms that preserve
its  spin structure.  If there were a universal relative spin bundle $\L$ in this situation, then we would interpret (\ref{gensp}) to mean
that the odd normal bundle to $\sM_g$ in $\SM_g$ is the vector bundle over $\sM_g$ with fiber $H^1(C,\L|_C)$.  
However, the existence of such an $\L$ is obstructed\footnote{Relatively simple families realizing the obstruction have been described
by J. Ebert and O. Randal-Williams and by B. Hassett, who has also pointed out reference \cite{GV}.} for sufficiently large $g$ \cite{ER,R}.  (We do not know if the obstruction can be
eliminated by endowing $C$ with a suitable level structure or by otherwise taking a finite cover of $\sM_g$.)  
This obstruction means that the ``odd normal bundle'' to $\sM_g$ in $\SM_g$ is not a vector bundle, even in the orbifold sense.
It is better described as a twisted vector bundle, twisted by a $\Z_2$-valued gerbe.  

Thus, to properly understand the moduli stack of super Riemann surfaces, one should think of a curve $C$ with spin structure,
even if $C$ has no non-trivial geometrical automorphisms that preserve its spin structure, 
as having a $\Z_2$ group of automorphisms $\{\pm 1\}$ acting on its spin bundle.
When $C$ is such that there is a nontrivial group $F$ of geometrical automorphisms that preserve the spin structure, the automorphism group
that is relevant in $\SM_g$ is the double cover $\widehat F$ of $F$ that acts on $T_C^{1/2}$; in general, this is a nontrivial double cover of $F$.
(The stacky structure of the Deligne-Mumford compactification of $\SM_g$ is still more subtle because in general there are separate
groups $\{\pm 1\}$ acting on the spin bundles of different components of $C$.)

To fully understand $\SM_g$, one should generalize the theory of supermanifolds to a theory of superstacks and understand $\SM_g$
in this framework.  A very special case of this more general theory is as follows.  Let $\hat\varrho$ be a $\Z_2$-gerbe over an ordinary
manifold (or algebraic variety) $M$.  Then by a $\hat\varrho$-twisted supermanifold $S$ with reduced space $M$, we mean a $\hat\varrho$-twisted
sheaf $\O_S$ of $\Z_2$-graded algebras, such that the even part of $\O_S$ is an ordinary sheaf, the odd part is a $\hat\varrho$-twisted
sheaf, and $\O_S$ is locally isomorphic to the sheaf of sections of $\wedge^\bullet V$, where $V\to M$ is a $\hat\varrho$-twisted vector
bundle.  
In the approximation of ignoring geometrical automorphisms (or eliminating them by picking a level structure), we can view
$\SM_g$ as a $\hat\varrho$-twisted supermanifold with reduced space $\sM_g$, where $\hat\varrho$ is the gerbe associated to the obstruction
$\varrho$ 
to finding a relative spin bundle.   The obstruction class $\omega_2$ to splitting of a $\hat\varrho$-twisted
supermanifold can be defined rather as for ordinary supermanifolds, with similar properties.

For the limited purposes of the present paper, we do not really need to go  in that direction.  Our considerations 
showing that $\SM_g$ is not projected involve concrete 1-parameter families of spin curves, over which a relative spin bundle will be visible. So the supermanifold framework is adequate for our purposes.  We will construct explicit
families of super Riemann surfaces, parametrized by a base $B$ that  in our examples generally will have dimension $(1|2)$,
and show that the  class $\omega_2(\SM_g)$ that obstructs
projection or splitting of $\SM_g $ is non-zero by showing that it has a nonzero
 restriction to $B$ -- or more precisely a nonzero  pullback to $B$, where this more precise statement accounts for the fact that
 some of the spin curves parametrized by $B$ have geometrical automorphisms.

\subsection{Punctures}

\mbox {   }

The notion of a puncture or a marked point on an ordinary Riemann surface has two analogs on a super Riemann surface.
In string theory, they are known as Neveu-Schwarz (NS) and Ramond punctures, respectively.

An NS puncture in a super Riemann surface $\S$
is the obvious analog of a puncture in an ordinary Riemann surface.  It is simply the choice of a point in
$\S$, given in local coordinates $x|\theta$ by $x=x_0$, $\theta=\theta_0$, for some $x_0,\theta_0$.  As we have 
learned in section \ref{basics}, such an NS puncture determines and is determined by a minimal divisor on $\S$.  Just as in the classical case,
deformation theory in the presence of an NS puncture at a point $p$  is obtained by restricting the sheaf $W$ of superconformal vector fields
to its subsheaf $W_p$ consisting of superconformal vector fields that leave fixed the point $p$.  By this definition, $W_p$ is a subsheaf of the sheaf
of sections of $TS$, but not a locally free subsheaf.  However, rather as we explained in the absence of the puncture, $W_p$ can be given 
a natural structure of a locally-free sheaf; in fact, it is isomorphic to $\D^2(-F)$, where $F$ is the minimal divisor that corresponds to the point $p$
in the correspondence described in section \ref{basics}.  To see this, we use local coordinates $x|\theta$, and take $p$ to be the point $x=\theta=0$.
The condition for a superconformal vector field to vanish at $p$ is then that 
 $f(0)=g(0)=0$ in eqns. (\ref{evensup}) and (\ref{enrup}).  This means precisely that $a(x|\theta)=f(x)+2\theta g(x)$ vanishes at $x=0$ in the computation described
 in eqn. (\ref{fondo}). But the divisor $F$ corresponding to the point $p$ is defined by $x=0$, as explained in Remark \ref{funremark}.  So the condition
 that $a(0|\theta)=0$ precisely means that $a(x|\theta)\partial_x$ is a section of $\D^2(-F)$.

If $\SM_{g,1}$ is the moduli space of super Riemann surfaces with a single NS puncture, then its reduced space is $\sM_{g,1}$,
which parametrizes spin curves of genus $g$ with a single puncture $p$.  The analog of eqn. (\ref{gensp}) is
\begin{align}\label{genspo} T_{+,S}\SM_{g,1} &= H^1(C,T_C(-p)) \cr T_{-,S}\SM_{g,1}&= H^1(C,T_C^{1/2}(-p)), \end{align}
where the twisting by $\O(-p)$ reflects the conditions $f(0)=g(0)=0$.  Eqn. (\ref{genspo}) has an obvious generalization for the moduli
space $\SM_{g,n}$ of super Riemann surfaces with any number $n$ of NS punctures.

We also are interested in the case of a $(1|1)$ supermanifold $\S$ that is endowed with a superconformal structure that degenerates
along a divisor in the following way.  We assume that the underlying supermanifold
$S = (C, \OO_S)$ is still smooth, but the  odd distribution $\D \subset T_-S$ is no longer  everywhere non-integrable: the local form \eqref{SRSD} for a generator is replaced by 
\begin{equation} \label{SRSD wRP}
v := \frac{\partial }{\partial \theta} +  x^k \theta \frac{\partial }{\partial x}.
\end{equation}
For negative $k$, such $v$ is meromorphic, and $\D$ fails to be a distribution along the divisor $x=0$.  (One can multiply by $x^{-k}$ to
remove the pole, but the resulting vector field $x^{-k}\partial_\theta+\theta\partial_x$ vanishes modulo nilpotents at $x=0$ and does
not define a distribution there; the notion of odd distribution is explained  for instance in the proof of Lemma \ref{egh}.)
For $k \ge 1$, such a $\D$ is a distribution but the non-integrability fails along the divisor $x=0$, with multiplicity $k$. 
In fact, $v^2 = x^k \partial /\partial x$ vanishes along the divisor $x=0$ to order $k$.
We say that $\S = (S,\D)$ is a {\em SRS with a  parabolic structure}  of order $k$ at the divisor $x=0$ (and we use this definition also for negative $k$).  The basic case $k=1$ is called a Ramond puncture by string theorists. The local form is:
\begin{equation} \label{SRSDR}
v :=\frac{ \partial }{\partial \theta }+  x \theta \frac{\partial }{\partial x}.
\end{equation}
It can be shown that in the presence of a Ramond puncture, the sheaf of superconformal vector fields is still isomorphic to $\D^2$, but
$\D^2$ is no longer isomorphic to $TS/\D$; rather, $\D^2\cong TS/\D\otimes \OO(-F)$, where now $F$ is the divisor on which the superconformal
structure degenerates (thus, in the example (\ref{SRSDR}), $F$ is the divisor $x=0$).

 For our purposes, the importance of parabolic structures  is that they arise naturally
in branched coverings, as we will see in section  \ref{SRSbranched}.   
In string theory, and also in the context of the Deligne-Mumford compactification of supermoduli space, there is a close analogy
between NS and Ramond punctures, and it is natural to define moduli spaces -- or rather stacks -- that parametrize super Riemann surfaces
with a specified number of punctures of each type.
However, in this paper, we limit ourselves to the moduli spaces $\SM_{g,n}$ of super Riemann surfaces
of genus $g$ with $n$ NS punctures.  

The moduli space $\SM_{g,1}$ can be interpreted as the total space of the universal genus $g$ super Riemann surface parametrized by 
$\SM_g=\SM_{g,0}$. It has dimension $(3g-2|2g-1)$ and comes with a morphism to $\SM_{g}$ whose fibers are the $(1|1)$ supermanifolds $S$ underlying the genus $g$ super Riemann surfaces parametrized by $\SM_g$.  In the stacky sense, it is  not necessary to correct
this statement for small $g$ to take into account the generic automorphisms of a Riemann surface.  In our applications, we  will always be pulling back
the relevant structures on moduli stacks to concrete families of Riemann surfaces or super Riemann surfaces and again we need not worry
about the automorphisms.

\subsection{Effects of geometric operations}

\subsubsection{\it Effect of a branched covering} \label{SRSbranched}
Let $\pi: \wt{C} \to C$ be a branched covering of ordinary Riemann surfaces. There are some ramification points $\wt{p}_j \in \wt{C}$ whose images in $C$ are the branch points $p_j = \pi(\wt{p}_j)$, where the covering map $\pi$ has local degree 
$k_j \ge 2$. Let $\wt{S} \to S$ be a branched covering of $(1|1)$-dimensional complex supermanifolds whose reduction is $\pi$, as in section \ref{covers}. Corresponding to each $\wt{p}_j$ there is now a ramification divisor $R_j \subset \wt{S}$ sitting over a branch divisor $B_j \subset S$.   Let $\S = (S,\D)$  be a SRS with underlying supermanifold $S$. The structure induced on $\wt{S}$ is not that of a SRS, but rather a SRS with parabolic structure of (negative) order $1-k_j$ along the $R_j$. Indeed, if $v$ is given near $p_j \in C$ by \eqref{SRSD}, and the local coordinate  $w$ near $\wt{p}_j \in \wt{C}$  satisfies $w^{k_j}=x$, then the induced (meromorphic) vector field upstairs is $\wt{v} =  \partial /\partial \theta +  \theta \partial /\partial x =
\partial /\partial \theta +  \frac{1}{k_j} w^{1-k_j} \theta \partial /\partial w$.

The above remains true, of course, when $k_j=1$, except that $p_j$ is no longer a branch point and $\wt{v}$ no longer has a pole. A little more generally,  if $\S = (S,\D)$  is a SRS with underlying supermanifold $S$ and parabolic structure of order $m_j$ at minimal divisor $B_j$, and  $\pi: \wt{C} \to C$ is a branched covering of ordinary Riemann surfaces with local degree $k_j \ge 1$ at $\wt{p}_j$, then a branched covering  $\wt{S} \to S$ inherits the structure of a SRS with parabolic structure
at each $R_j$ of order $k_j  (m_j -1) +1$. This follows immediately from equation \eqref{SRSD wRP} and the local form of the branched cover, as given in section \ref{covers}.

\subsubsection{\it Effect of a blowup}\label{SRS blowup}

\mbox {   }
\def\wh{\widehat}
Let $S$ be a $(1|1)$-dimensional complex supermanifold, with local coordinates $x|\theta$ defining the point, or codimension 
$(1|1)$ submanifold $p: \{x = \theta = 0 \}$. In section \ref{blowup} we described the blowup $\wh{S}$ of $S$ 
at $p$: it is again a  $(1|1)$-dimensional complex supermanifold, with the same reduced manifold as $S$, 
but with a new (and ``larger'') structure sheaf. It has local coordinates $\wh{x}|\wh{\theta}$ such that the map 
$\wh{S} \to S$ sets $\wh{x} = x, \wh{\theta} = \theta/x$, replacing the point $p$ by the minimal divisor $\wh{x} = 0$.

Now let ${\S} = ({S},{\D})$ be a SRS with local superconformal coordinates $x|\theta$,
and thus with the distribution $\D$ generated by $v=\partial_\theta+\theta\partial_x$.
As in section \ref{SRSbranched}, let $\pi:\wt \S\to \S$ be a branched cover of $\S$,
with $k$-fold ramification along the divisor $x=0$.  Concretely, $\wt S$ is described by local coordinates $y|\theta$, where $y^k=x$,
and the distribution is generated by $v=\partial_\theta+(\theta/k y^{k-1})\partial_y$, with parabolic structure of degree $-(k-1)$.
In $\S$, the divisor $x=0$ has the distinguished point $x=\theta=0$ (see Remark \ref{funremark}), and this can be pulled
back to the point $y=\theta=0$ in $\wt S$.   Then we can blow up this point, to get a new complex supermanifold $\widehat S$ with
local coordinates $\widehat y=y$ and $\wh \theta=\theta/y$.    The generator of the distribution becomes $ v= (1/y)\partial_{\wh\theta}
+(1/k y^{k-2}) \wh\theta\partial_x$.  Away from $y=0$, the same distribution is generated by $\wh v=yv=\partial_{\wh\theta}+(1/ky^{k-3})
\wh\theta\partial_y$.

For $k=3$, which actually will be the basic case in our later applications, this simple blowup has eliminated the parabolic
structure, and $\wh S$ has an ordinary super Riemann surface structure near $y=0$.    More generally, for any $k$,
the effect of a blowup has been to increase the order of a  parabolic structure by $2$.

Actually, we can make a blowup to  increase the order of a parabolic structure by any positive even number. This is not achieved by a series of blowups as above; the problem is that after the first blowup, there is no distinguished point to blow up inside the exceptional divisor. Instead, we need to blow up a more complicated multiple point, specified by a certain sheaf of ideals.
We describe this ideal as follows.  Downstairs, the distinguished point $x=\theta=0$ on the divisor $x=0$ determines the ideal
$I$ generated by $x$ and $\theta$.  The ideal $\pi^*(I)$ is generated by $y^k$  and $\theta$.  On the other hand, upstairs, we have the
ideal $J$ generated by $y$ and hence also its $\ell^{th}$ power $J^\ell$, generated by $y^\ell$. We let $I_{\ell}=(J^\ell,\pi^*I)$ be the
ideal generated by $J^\ell$ and $\pi^*I$. For $k \geq \ell$ it is generated, in local coordinates, by $y^\ell$ and $\theta$.
 Blowing up along this  ideal has the effect
of replacing $\theta$ by $\theta' := \theta/y^{\ell}$ and thus  increasing the order of the parabolic structure by $2\ell$, i.e. from 
$1-k$ to $1-k+2\ell$. If $k$ is odd, we can thus produce a SRS with no parabolic structure.  If $k$ is even, we can reduce to the case
of parabolic structure of degree 1.  These are the two cases (no parabolic structure and parabolic structure of degree 1, corresponding
to a Ramond puncture) that are usually
considered in string theory.

In short, we can eliminate the parabolic structure caused by a branched covering of local degree $k$ if and only if $k$ is odd, in which case we need to blowup the sheaf $I_\ell$, with $\ell=(k-1)/2$. Recalling the result in section \ref{covers}, we conclude{\footnote {A related assertion has been made on p. 61 of \cite{Natanzon}.}} :

\begin{proposition}\label{branched covers of SRS}
To a family $f: \S  \to \SM$ of super Riemann surfaces parametrized by a super manifold $\SM$, with underlying  super manifold $S \to \SM$, together with a branched covering map $\pi_{{\mathrm{red}}}: \wt{C} \to C$ of the reduced space $C := S_{{\mathrm{red}}}$ with all its local degrees {\em odd}, and a divisor $D \subset S$ whose reduced manifold is the branch locus $B$ of $\pi_{{\mathrm{red}}}$, there is naturally associated a family of super Riemann surfaces $\wt{f}: \wt{S}  \to \SM$ which factors: $\wt{f} = f \circ \pi$ through a branched covering map 
$\pi: \wt{\S} \to \S$   whose reduced version is the given $\pi_{{\mathrm{red}}}$.
\end{proposition}

Although we do not need this in the sequel, we note that the above has an interesting converse: any minimal divisor $\wt{x} = 0$ on a $(1|1)$-dimensional complex supermanifold $\wt{S}$, with
local coordinates $\wt x|\wt \theta$,  can be blown down to a point $p: \{x = \theta = 0 \}$ on a $(1|1)$-dimensional complex supermanifold $S$ with the same  reduced space, the same structure sheaf away from $p$, and  coordinates $x=\wt{x}, \quad \theta =x \wt{\theta}$ near $p$. This process is natural, and can be described independently of the choice of coordinates. Given a $(1|1)$ supermanifold $S$ with minimal divisor $D$,
we blow down $D$ by allowing only functions that are constant along $D$.   If more generally $S \to  B$ is a family of $(1|1)$ supermanifolds
with minimal divisor $D$, we only allow functions on $S$ whose restriction to $D$ is a pullback from $B$.
In bosonic algebraic geometry, one could make the same definition by allowing only functions that are constant
along a given divisor, but in general
this would not give the sheaf of functions on an algebraic variety. For a $(1|1)$ supermanifold, 
the computation in the local coordinates $x|\theta$
does show that the blowdown works.
So for example, let $\wt{\S} = (\wt{S},\wt{\D})$ be a SRS that has along the minimal divisor $\wt{x}=0$ a parabolic structure of multiplicity $m + 2$. As we have just seen, this divisor can be blown down, reducing the multiplicity of the  parabolic structure to $m$.  The blowdown
process can be repeated, further reducing the multiplicity.

Everything that we have said in this section works naturally in families. Start with a family of SRS parametrized by some base supermanifold $B$ with  parabolic structure along a relative minimal divisor (i.e. a divisor intersecting each fiber in a minimal divisor). We can blow the divisor down, and thereby reduce the order of the family of parabolic structures by 2. Or we can blow up the relative point corresponding to the given divisor, and thereby increase the order of the family of parabolic structures by 2, or blow up a more subtle sheaf of ideals, increasing
it by $2\ell$ for any integer $\ell$.

\subsection{A non-split supermanifold} \label{local non-splitness}

\mbox {   }

In this section we exhibit a particular non-split supermanifold $X_\eta $. It has dimension $(1|2)$, and is fibered over the odd line $\CC^{0|1}$. The fibers are super Riemann surfaces. In fact, we interpret $\CC^{0|1}$ as an odd tangent line to  $\SM_{g}$, and build our $X_\eta $ by restricting the universal super Riemann surface $\SM_{g,1}$ to this line. This example will serve as a crucial ingredient in our proof of non-projectedness of $\SM_{g,1}$ in section \ref{marked non-splitness}.

Let $\S = (S,\D)$ be a split  SRS of genus $g$, $C:=S_{{\mathrm{red}}}$ the underlying Riemann surface.
Let $\eta  \in H^1(C, T_C^{1/2})$ be an odd tangent vector. It determines a map
$f_\eta : \CC^{0|1} \to \SM_{g}$. By pulling back the universal super Riemann surface $\SM_{g.1} \to \SM_{g}$, we obtain a $(1|2)$-dimensional supermanifold $X_\eta  := f_\eta ^* (\SM_{g,1})$. By definition, $X_\eta $ comes with a projection $\pi_\eta : X_\eta  \to \CC^{0|1}$.
\begin{proposition} 
 $X_\eta $ is projected if and only if $\eta =0$, in which case it is actually split. 
\end{proposition} 
\begin{proof}
First we remark that when $\eta =0$, the map $f_\eta $ is constant, so $X_\eta $ is the product $S \times \CC^{0|1}$ so in particular it is split.
In general, we are in the situation of Corollary \ref{m|2}, so $X_\eta $ is determined by its first obstruction:
\[
\omega := \omega(X_\eta ) \in H^1((X_\eta )_{{\mathrm{red}}}, \Hom( \exterior{2} T_{-} X_\eta , T_{+} X_\eta )). 
\]
We can identify:
\[(X_\eta )_{{\mathrm{red}}} = C\]
\[ T_{+} X_\eta  = TC\]
\[
\exterior{2} T_{-} X_\eta  = T_{-}S \otimes T_-{\CC^{0|1}} = T_C^{1/2} \otimes \OO = T_C^{1/2}.
\] 

So $\omega$ lives in 
$$H^1((X_\eta )_{{\mathrm{red}}}, \Hom( \exterior{2} T_{-} X_\eta , T_{+} X_\eta )) = H^1(C,T_C^{1/2}).$$

Our proposition follows from:
\begin{lemma} \label{Id}
Under the natural identifications, $\omega(X_\eta ) = \eta $.
\end{lemma}

We will give a very pedestrian explanation.  Since $X_\eta$ has odd dimension 2,
the class in $H^1(C,G)$ associated to $X_\eta$ is a 1-cocycle on the split model $S(M,V)=S\times \C^{0|1}$ valued in
vector fields of the form
\begin{equation}\label{modno} w(x)\eta\theta\frac{\partial}{\partial x}, \end{equation}
where $x$ is a local coordinate on $S_{\mathrm{red}}=C$ and $\theta$ is a local odd coordinate on $S$ that vanishes along $C$.
We can view this as a 1-cocycle that deforms $S\times \C^{0|1}$ away from being split.  $S\times \C^{0|1}$ has other first-order
deformations, but they do not affect its splitness.

On the other hand, to deform the super Riemann surface $\S$ in an $\eta$-dependent fashion, leaving it fixed at $\eta=0$,
we use a one-cocycle valued in odd superconformal vector fields on $\S$, multiplied by $\eta$.  Given the form (\ref{enrup}) of
an odd superconformal vector field, the one-cocycle is valued in vector fields of the form
\begin{equation}\label{lumbo} - g(x)\eta\left(\frac{\partial}{\partial\theta}-\theta\frac{\partial}{\partial x}\right)    .       \end{equation}
If we forget the superconformal structure and simply view this as a one-cocycle that we use to deform the complex structure of
$S\times \C^{0|1}$, it is a sum of two terms, namely $-g(x)\eta\partial_\theta$ and $g(x)\eta\theta\partial_x$, that can be considered
separately.  The first term does not affect the splitness of $S\times \C^{0|1}$, but the second does.  Indeed if we set $w=g$,
the second term coincides
with the cocycle (\ref{modno}) that characterizes $X_\eta$. 
The value that the cocycle has in one interpretation is the same as the value that it has in the other interpretation, since
either way the vector field $w(x)\eta\theta\partial_x$ or $g(x)\eta\theta\partial_x$ can be naturally identified with a section over $C$ of 
$T_C^{1/2}$ and concretely the cocycles under discussion  represent  elements of $H^1(C,T_C^{1/2})$.

\end{proof}

\section{Non-projectedness of $\SM_{g,1}$} \label{marked non-splitness}

\mbox {   }

We are now ready to prove Theorem \ref{Main2}, which says that the even spin component of the moduli space $\SM_{g,1}$ of super
Riemann surfaces with one NS puncture is non-projected. The result follows from the more precise:

\begin{proposition}
The first obstruction to the splitting of $\SM_{g,1}$:
\[
\omega:=\omega_2 \in H^1(\sM_{g,1}, \Hom(\exterior{2}T_-,T_+))
\]
does not vanish for $g \geq 2$ (and even spin structure), so the supermanifold $\SM_{g,1}$ is non-projected. 
\end{proposition}

Here and in the rest of this section, $T_{\pm}$ refer to $T_{\pm}\SM_{g,1}$. Our proof here is based on some of  the general results we obtained about supermanifolds and their obstructions. In the sequel to this paper \cite{DWtwo} we give a different proof which
relies on a cohomological interpretation of the obstruction class.

\begin{proof}[Proof]

Fix a spin curve $(C,T_C^{1/2}) \in \sM_g^+$ and an odd tangent vector $\eta \in H^1(C,T_C^{1/2})$. 
We have already seen in section \ref{local non-splitness} an example of 
a family $X_\eta$ of super Riemann surfaces with a  non-projected total space.
We identify this $X_\eta$ as a submanifold of  $\SM_{g,1}$. In terms of the 
projection $\pi: \SM_{g,1}  \to \SM_g,$
it is $\pi^{-1}({\CC}^{(0|1)})$, where ${\CC}^{(0|1)}$ is embedded in $\SM_g$ via the odd tangent vector $\eta$.
We wish to apply Corollary \ref{mila}, with $S = \SM_{g,1}, \quad S' = X_\eta, \quad M = \sM_{g,1}, \quad M' = C.$ We note that $T_{-}X_\eta $ is a rank 2 vector bundle on $C$, in fact it is an extension:
\[
0 \to T_C^{1/2}  \to T_{-}X_\eta  \to {\CC}\eta \to 0,
\]
where the third term stands for the trivial bundle with fiber ${\CC}\eta$. The choice of $\eta$ therefore identifies 
$\exterior{2}T_{-}X_\eta$ with $T_C^{1/2}$.
The maps $\iota, j$ appearing in Corollary \ref{mila} can therefore be written explicitly in our case:

$ \iota: H^1(C, \Hom(\exterior{2}T_-,T_+)) \to 
H^1(C,\Hom(T_C^{1/2}, T_+))$

$j: H^1(C,\Hom(T_C^{1/2}, T_C)) \to
 H^1(C,\Hom(T_C^{1/2}, T_+)).$

According to Corollary \ref{mila}, $\iota(\omega_{|C}) =j(\omega(X_\eta))$. We already know by Lemma \ref{Id} that $\omega(X_\eta) =\eta \neq 0$, so in order to show that $\omega \neq 0$, it suffices to show that $j$ is injective. Unfortunately, Lemma \ref{submanifold splits} does not apply. Instead, we note that $j$ fits into an exact sequence. We start with the  short exact sequence of sheaves on $C$:
\begin{equation}\label{even}
0 \to T_C   \to    i^* T_+   \to   i^* \pi^*T_{+}\SM_g   \to   0
\end{equation}
induced on (even)  tangent spaces by the fibration $\pi: \SM_{g,1}  \to \SM_g$ and the inclusion $i: C \to \SM_{g,1}$. Note that the third term is isomorphic to the trivial sheaf $W \otimes \OO_C$ with $W$ the tangent space $T_{+,C}\SM_g$ to $\SM_g$ at the point $[C]$. We apply
the functor  $\Hom(T_C^{1/2},\cdot)$ to this sequence;  the cohomology sequence of the resulting exact sequence reads in part
\[
W \otimes H^0(C,  K_C^{1/2}) \to H^1(C,\Hom(T_C^{1/2}, T_C))  \stackrel{j}{\to}
 H^1(C,\Hom(T_C^{1/2}, T_+)).
\]
 For generic choice of the even spin curve $(C,T_C^{1/2})$, we have $H^0(C,  K_C^{1/2}) =0$, so $j$ is injective, completing the proof.
\end{proof}

\begin{remark} Our proof fails for the odd component; in fact the map $j$ above is identically zero in this case, because $H^0(C,  K_C^{1/2})$ is generically $1$-dimensional  rather than $0$. In more detail, the vanishing of $j$ follows from surjectivity of the above map
\[
 W \otimes H^0(C,  K_C^{1/2}) \to H^1(C,\Hom(T_C^{1/2}, T_C)),  
\]
which can be written explicitly as
\[
H^1(C,T_{C})  \otimes H^0(C,  K_C^{1/2})   \to H^1(C,T_{C}^{1/2}),  
\]
whose surjectivity is equivalent (in the generic case when $H^0(C,  K_C^{1/2})$ is $1$-dimensional) to injectivity of the trasposed map
\[
H^0(C,  K_C^{1/2}) \otimes H^0(C,  K_C^{3/2})  \to H^0(C,  K_C^{2}),
\]
but the latter map is just multiplication with a fixed non-zero section of $H^0(C,  K_C^{1/2})$, which is indeed injective.
\end{remark}

\section{Compact families of curves and non-projectedness of $\SM_{g}$}

To show that $\SM_g$ is non-projected, it suffices to show that its first obstruction 
$\omega := \omega_2(\SM_g) \in H^1(\sM_g, \Hom( \exterior{2} T_{-}, T_{+} ))$
does not vanish. 
Here $\sM_g = (\SM_g)_{{\mathrm{red}}}$ is the moduli space of ordinary Riemann surfaces with a spin structure.
$T_{\pm}$ stand respectively for the even and odd tangent bundles of $\SM_g$. These are vector bundles on $(\SM_g)_{{\mathrm{red}}}$;
their fibers at  $C \in \sM_g$ are  $H^1(C,T_C)$ and $H^1(C,T_C^{1/2})$, respectively.

The basic idea of the proof is to construct a compact curve $Y\subset \sM_g$, or more precisely a family of spin curves parametrized
by a compact curve $Y$, and to show that the pullback of the class $\omega$ to $Y$ is nonzero.  The families of genus $g$ curves
that we study are constructed in a standard
fashion to parametrize a family of ramified covers of a fixed curve of genus less than $g$.  The construction is such that we can prove
that both components of $\SM_g$ are non-projected.

\subsection{Examples of compact families of curves} \label{examples}

\mbox {   }

One general way to produce compact curves in ${\M}_g$ depends on the existence of a small compactification. We briefly review this approach.

The Satake compactification $\overline{\A_g}$ of the moduli space $\A_g$ of abelian varieties is a projective variety whose boundary is $\overline{\A_{g-1}}$, hence of codimension $g$.  The  closure $\overline{\M_g}$ of the Abel-Jacobi image of $\M_g$ does not meet this boundary transversally: it meets the boundary in the compactification $\overline{\M_{g-1}}$ of ${\M_{g-1}}$, which for $g \geq 3$ has codimension 3 in $\overline{\M_g}$. (Contrast this with the Deligne-Mumford compactification, whose boundary has codimension 1.) The difference $\overline{\M_g} \setminus \M_g$ consists of this boundary plus a locus in the interior of $\A_g$, namely the closure of the locus of reducible curves consisting of two components meeting at  a point. The genera $g_1,g_2 > 0$ of these components add up to $g$. Most components of this locus also have codimension 3. But for every genus, there is one component whose codimension is 2: this happens when one of the $g_i$ equals 1. (When both $g_1=g_2=1$,  the codimension is just 1; but this happens only for $g=2$.) 

We can embed the projective variety $\overline{\A_g}$, and hence also its subvariety $\overline{\M_g}$,  in a large projective space $\PP^N$. Consider the 1-dimensional  intersection of $\overline{\M_g}$ with a generic linear subspace in $\PP^N$ of the appropriate codimension, which is $3g-4$.  The above dimension count  showed that when $g \geq 3$, the complement 
$\overline{\M_g} \setminus \M_g$ has codimension at least 2 in $\overline{\M_g}$. It follows that our generic $1$-dimensional intersection is contained in the interior $\M_g$. This provides a large but non-explicit family of compact curves in $\M_g$ for $g \geq 3$.
(On the other hand, when $g\leq 2$, it is known that $\M_g$ is an affine variety; hence it can contain no compact curves.)

In addition, several explicit constructions are known.  Kodaira, Atiyah and Hirzebruch constructed examples \cite{K,A,H} of surfaces $X$ with smooth maps $\pi: X \to B$ to a smooth curve $B$ of genus $g'$.  All the fibers of such a  map are smooth curves $C$ of some genus $g$. In fact, their surfaces $X$ are certain branched covers of the product of two curves, so they can be fibered as above in two distinct ways. In their smallest example, the base genus is 129 and the fiber genus is 6, for one fibration; and the base genus is 3 and the fiber genus is 321, for the other fibration. More efficient constructions are known, producing examples of lower genera. The construction in \cite{BD}  gives
 base genus 9   and fiber genus 4, for one projection, and base genus 2   and fiber genus 25, for the other.

A fibered surface $X$ as above  determines a map from $B$ to $\M_g$. When this map is non-constant, the signature of $X$ is non-zero. 
Conversely, if a universal curve over ${\mathcal C} \to\M_g$ existed, every map from $B$ to $\M_g$ would determine such a fibered surface $X$. Actually this fails, since a universal curve over $\M_g$ does not exist near curves with extra automorphsms. For example, it is easy to see that the base $B$ of a fibered $X$ must have genus $g' \geq 2$, since the universal cover of $B$ must map to the Siegel half space by the period map. However, we will see in the Appendix an example of a genus 0 curve in $\M_5$. There is no family of genus $g=5$ curves fibered over this $\PP^1$, but there is such a family over a certain cover $B$ of $\PP^1$,
with genus $g' = 19$. Earlier examples of compact genus 0 curves in moduli spaces appeared in \cite{O}.

\subsection{Covers with triple ramification}\label{triple covers}

The strategy behind the explicit constructions mentioned at the end of the previous section is to start with a curve $C$ of some lower genus $g_0$ and a branched cover $\pi: \wt{C} \to C$  having branch divisor $B \subset C$. Then  keep the curve $C$ fixed, and allow $B$ to move in a $1$-parameter family $P$.
The crucial condition is that the points of $B$ should never collide, i.e. the cardinality of $B=B_p$ must remain constant as $p$ varies in $P$. This means that the topology of the pair $(C,B)$ remains constant as $B$ varies, and therefore as $p$ varies locally in $P$, the cover  $\pi: \wt{C_p} \to C$ deforms along with $B$. Globally there may be monodromy: there are usually many branched covers with a specified branch divisor $B$.  So this construction produces a family of covers $\pi: \wt{C_p} \to C$  parametrized by points $p$ of some cover $\wt{P}$ of $P$. 

One way to guarantee compactness of our parameter space $\wt{P}$, or to enforce the condition that points of $B$ should never collide, is to start with a $C$ which admits a free action of a finite group $G$, and to take the branch divisors $B_p$ to be the orbits, parametrized by the quotient $P:=C/G$. This is the basic idea behind the constructions in \cite {K, A, H, BD}, where the simplest examples use double covers. For our purpose in this work, we will avoid collisions by taking $B$ to consist of a single point: $B_p = \{ p \}$.

 It is well known that a double cover cannot have just one branch point. But this {\em is} possible for degree $\ge 3$. 
One way to see the impossibility in degree $2$ is  to note that this would produce a branched cover $\wt{C}$ of odd Euler characteristic. The same argument excludes a single, total branch point in all even degrees, but allows  a single, total branch point in any odd degree; in any degree $d$, odd or even, it does not exclude a single branch point of various non-total types, e.g. the $(3,1, \ldots, 1)=(3,1^{d-3})$ pattern works for all $d \ge 3$, as we will see shortly. We settle then on the following version of the construction:

Fix a curve $C$ of genus $g_0$ and an integer $d\geq 3$. Consider the family $\wt{P}$ of all branched covers $\pi: \wt{C} \to C$ of degree $d $ having a {\em single} ramification point $\wt{p} \in \wt{C}$, of local ramification degree $3$, over a branch point $p \in C$. The fiber $\pi^{-1}(p)$ consists of $3$ times $\wt{p}$ plus $d-3$ other points $\wt{p}_i, \quad i=1, \ldots, d-3$.

We can easily see that such branched covers indeed exist. Consider the fundamental group of the complement, $\pi_1(C \setminus p)$. This is generated by a standard symplectic basis $\alpha_i, \beta_i, \quad   i=1, \ldots, g_0$, plus a loop $l$ around the point $p$. The only relation is $l r = 1$, where $r :=\prod_{i=1}^{g_0}  [ \alpha_i, \beta_i ]$. There is thus a short exact sequence
\begin{equation}\label{pi1}
0 \to K \to \pi_1(C \setminus p) \to \pi_1(C)  \to 0
\end{equation}
where the kernel $K$ is generated by $l$ (or equivalently, $r$). Now a $d$-sheeted cover of $C$ which is not branched except possibly at $p$ is specified by a subgroup $S \subset \pi_1(C \setminus p)$ of index $d$, and this cover is unbranched at $p$ if and only if $K \subset S$. 
Every subgroup of index 2 is normal, so it is the kernel of a homomorphism $ \pi_1(C \setminus p) \to \ZZ/2$ and must contain all commutators, hence must contain $K$. 
But for $d \ge 3$,  subgroups of index $d$ which do not contain $K$ do exist. For example, when $d=3$, we can map $\pi_1(C \setminus p)$  to the symmetric group  $S_3$ by sending $\alpha_1 \to (23), \quad \beta_1 \to (123), \quad $ and $\alpha_i, \beta_i \to 1$ for $i \ne 1$, so $r \to (123)$. We then take $S$ to be the inverse image of the non-normal subgroup $S_2 \subset S_3$ stabilizing one of the three permuted objects. This $S$ has index 3 in $ \pi_1(C \setminus p)$ but does not contain the element $r \in K$. More generally, for any $d \ge 3$, 
we can map $\pi_1(C \setminus p)$  onto the symmetric group  $S_d$ by sending $\alpha_1 \to (23), \quad \beta_1 \to (12 \dots d), \quad $ and $\alpha_i, \beta_i \to 1$ for $i \ne 1$, so again $r \to (123)$. We then take $S$ to be the inverse image of the non-normal subgroup $S_{d-1} \subset S_d$ fixing $1$ and permuting $2,\dots,d$. This $S$ has index $d$ in $ \pi_1(C \setminus p)$, but the loop $l=r^{-1}$ goes to a 3-cycle, so the resulting cover $\wt{C} \to C$ has the desired $(3,1^{d-3})$ pattern.

The genus of any such $\wt{C}$ is easily seen to be $g=d (g_0 -1) +2$. We may as well take $g_0=2$ so $g=d+2$, and thus by
varying $d$, we get all values $g\geq 5$.
In the appendix we will give a detailed, algebro-geometric description of the $g_0=2, d=3, g=5$ case, in which the branching is total. This example is attributed in \cite{HM} (a few paragraphs above theorem 2.34) to Kodaira.

We saw that for each $p \in C$, we get a finite number of covers $\wt{C} \to C$.  The parameter space $\wt{P}$ of our branched covers of $C$ is therefore itself a certain cover of $C$. If we now allow the curve $C$ to vary through the moduli space $\M_{g_0}$, we get a family of covers $\pi: \wt{C} \to C$ parametrized by a certain cover 
$\wt{\M}_{g_{0},1} := \wt{\M}_{g_{0},1} ^d$ of the universal curve ${\M}_{g_{0},1}$, and a morphism 
$\wt{\M}_{g_{0},1} \to \M_g$ sending the isomorphism class of a cover $\pi: \wt{C} \to C$ to the isomorphism class of $\wt{C}$.

\subsection{Maps from $\wt{\sM}_{g_{0},1}$  to $\sM_{{g}}$} \label{spin maps}

The construction in the previous section  \ref{triple covers} gives a map of moduli spaces, from  $\wt{\M}_{g_{0},1}$  to $\M_{{g}}$. In the present section we extend this to a map of the spin moduli spaces, from  $\wt{\sM}_{g_{0},1}$  to $\sM_{{g}}$. In the next section we discuss the further extension to a map  (in the sense of stacks) of supermoduli spaces, from  $\wt{\SM}_{g_{0},1}$  to $\SM_{{g}}$.

For fixed $d \ge 3$, let $\wt{\M}_{g_{0},1} := \wt{\M}_{g_{0},1} ^d$ denote the moduli space of all branched covers $\pi: \wt{C} \to C$ as in the previous section: of degree $d$ over a curve $C$ of genus $g_0$, having a single ramification point $\wt{p} \in \wt{C}$, of local ramification degree $3$, over a branch point $p \in C$. 
(Recall from Proposition \ref{branched covers of SRS} that in order to lift a family of branched coverings of Riemann surfaces to branched coverings of super Riemann surfaces we need all the local ramification degrees to be odd. This is satisfied in our situation, where these local degrees equal 3 or 1.)

There are several forgetful morphisms:
\begin{itemize}
\item  
$P_1: \wt{\M}_{g_0,1}  \to  \M_g$   sends
$(\pi: \wt{C} \to C)$ to $\wt{C}$ .
\item
$P_2: \wt{\M}_{g_0,1}  \to \M_{g_0,1} $, sending  $(\pi: \wt{C} \to C) \mapsto (C,p)$, is a finite covering. 
\item
$P_3: \wt{\M}_{g_0,1}  \to \M_{g_0} $ sends  $(\pi: \wt{C} \to C)$ to $C$.
\end{itemize}

Our goal in this section and the next is to construct super versions of the space $\wt{\M}_{g_{0},1}$ and the morphism $P_1$. 
We do this  in several steps, designed to match the input of Proposition \ref{branched covers of SRS}:

\begin{enumerate}
\item Start with a typical branched covering map $\pi: \wt{C} \to C$ as in section \ref{triple covers}.
\item Put it into a universal family $\Pi: \wt{U} \to U$, parametrized by $\wt{\M}_{g_0,1} $.
\item Add spin:  $S\Pi: S\wt{U} \to SU$ over $S\wt{\M}_{g_0,1} $.
\item Construct the super space $\wt{\SM}_{g_{0},1}$ and a universal genus $g_0$ SRS ${\mathcal U}$ over it, 
${\mathcal F}:  {\mathcal U} \to \wt{\SM}_{g_{0},1}$, whose reduced spaces are $S\wt{\M}_{g_{0},1}$ and $SU$, respectively.
\item Construct the divisor ${\mathcal B} \subset {\mathcal U} $ whose reduced version is the branch divisor $SB$ of $S\Pi$.
\item By Proposition \ref{branched covers of SRS} we then get a  family of super Riemann surfaces
\newline $\wt{\mathcal{F}}: \wt{{\mathcal U}} \to \wt{\SM}_{g_{0},1}$ of genus $g$, along with a branched covering map 
${{\varPi}}: \wt{{\mathcal U}} \to {{\mathcal U}} $
whose reduced version is  $S\Pi: S\wt{U} \to SU$ and which satisfies $\wt{\mathcal{F}} = {\mathcal{F}} \circ {{\varPi}}$.
\item By the universal property of moduli spaces, this gives a morphism of supermanifolds, from $\wt{\SM}_{g_{0},1}$  to $\SM_{{g}}$, whose reduced version is the spin lift of the above forgetful map $P_1: \wt{\M}_{g_0,1}  \to  \M_g$.
\end{enumerate}

In fact, we need quite a bit less than this for the proof of our main results: we will use the construction only in the case $g_0=2$, and only in the vicinity of one fiber of $\wt{\SM}_{g_{0},1}$  to $\SM_{{g_0}}$. In fact, for $g_0=2$ one can give an elementary construction of a  relative spin line bundle $L = K ^{1/2}$ on the universal curve over a particular double cover of $\sM_2^+$, the one that parametrizes triples of Weierstrass points on the (hyperelliptic) genus 2 curve. So in this case one could prove a stronger result, about a map between moduli spaces (rather than stacks). Since we will need this only in the vicinity of a single curve, we will not work out the details of this improvement.

In the remainder of this section we will fill in the details of steps (2) and (3), leading to the map of spin moduli spaces.  The subsequent
steps are treated in section \ref{maps}.

(2) Over $\wt{\M}_{g_{0},1}$ we construct two universal spaces $a: U \to \wt{\M}_{g_{0},1}$ and $\wt{a}: \wt{U}  \to \wt{\M}_{g_{0},1}$, with an intertwining map $\Pi: \wt{U} \to U$ satisfying $a \circ \Pi = \wt{a}$. Here $U$ is just the universal genus $g_0$ curve over $\wt{\M}_{g_{0},1}$, meaning that it is the pullback of the universal curve ${\M}_{g_{0},1} \to {\M}_{g_{0}}$ via the above forgetful map $P_3: \wt{\M}_{g_0,1}  \to \M_{g_0} $:
\[
U = \wt{\M}_{g_{0},1} \times_{ {\M}_{g_{0}}} {\M}_{g_{0},1}.
\]
Similarly, $\wt{U}$ is the pullback of the universal genus $g$ curve ${\M}_{g,1} \to {\M}_{g}$ via the forgetful map $P_1: \wt{\M}_{g_0,1}  \to \M_{g}$:
\[
\wt{U} = \wt{\M}_{g_{0},1} \times_{ {\M}_{g}} {\M}_{g,1}.
\]
Note that $\Pi$ is a branched covering map. Its restriction to the fibers of $U, \wt{U}$ above a point of $\wt{\M}_{g_{0},1}$ is just the map  $\pi: \wt{C} \to C$ classified by that point. The branch divisor $B \subset U$ of $\Pi$ is the pullback via 
\[
P_2 \times 1: U = \wt{\M}_{g_{0},1} \times_{ {\M}_{g_{0}}} {\M}_{g_{0},1} \to {\M}_{g_{0},1} \times_{ {\M}_{g_{0}}} {\M}_{g_{0},1} =: \bar{U}
\] 
of the diagonal 
\[
\Delta \cong  {\M}_{g_{0},1}   \subset {\M}_{g_{0},1} \times_{ {\M}_{g_{0}}} {\M}_{g_{0},1} = \bar{U}.
\]

(3) Next, we add a spin structure. 
The space $\sM_{g_0} $ parametrizes genus $g_0$ curves with a spin structure. 
By pulling  $\M_{g_0,1} $ and $\wt{\M}_{g_{0},1}$ back to 
the cover $\sM_{g_0}  \to \M_{g_0} $,
we get  $\sM_{g_0,1} $, parametrizing genus $g_0$ curves with a spin structure and a marked point, and $\wt{\sM}_{g_{0},1}$, parametrizing branched covers $\pi: \wt{C} \to C$ with a single branch point $p$ as above and with a spin structure  $L$ on $C$, $L^2 \cong K_C$. 
Similarly, we let $SU,  S\wt{U}, S\Pi$ and $S\bar{U}$ denote the pullbacks of $U,  \wt{U}, \Pi$ and $\bar{U}$ 
from $\wt{\M}_{g_0,1} $ to $\wt{\sM}_{g_{0},1}$  (or equivalently, from $\M_{g_0} $ to $\sM_{g_0} $).
We note that $K_{ \wt{C} } \cong \pi^*K_C(2p)$,
so the spin structure $L$ on $C$ determines a spin structure $\wt{L}:=\pi^*L(p)$ on $\wt{C}$. (This would fail if one or more of the ramification points of $\pi$ had even local degree over $C$.)

	
\subsection{Maps from $\wt{\SM}_{g_{0},1}$  to $\SM_{{g}}$} \label{maps}

Our goal here is to lift the map of spin moduli spaces constructed in the previous section to a map of supermoduli spaces. We continue filling in the details of steps (4)-(7) outlined above.

(4)  Start with  $\SM_{g_{0},1}$ whose reduced space is  $\sM_{g_{0},1}$. From this we build $\wt{\SM}_{g_{0},1} $: 
since $\wt{\M}_{g_{0},1} \to { {\M}_{g_{0},1}} $ is a covering map, so is
 $\wt{\sM}_{g_{0},1} \to {S {\M}_{g_{0},1}} $, and
we get a covering supermanifold  $\wt{\SM}_{g_{0},1} \to \SM_{g_{0},1} $ as in section \ref{Covers}. 
Similarly, we can start with 
$\bar{\mathcal U} := \SM_{g_{0},1} \times_{\SM_{g_{0}}}   \SM_{g_{0},1}$, whose reduced space is $S\bar{U}$. 
Since $U \to \bar{U}$, hence also $SU \to \bar{SU}$, are covering maps, we again get a covering supermanifold 
${\mathcal U} \to \bar{{\mathcal U}}$, as in section \ref{Covers}.
In fact, we can identify this explicitly as:
 ${\mathcal U} =  \wt{\SM}_{g_{0},1} \times_{\SM_{g_{0}}}   \SM_{g_{0},1}$.  This has reduced space $SU$  and comes with a  map 
${\mathcal F}:  {\mathcal U} \to \wt{\SM}_{g_{0},1}$ which is the projection onto the first factor.

(5) To be able to apply the results of \ref{covers} and of Proposition \ref{branched covers of SRS}, we need a divisor 
${\mathcal B} \subset {\mathcal U}$ whose reduction is the branch divisor $SB \subset SU$, which as we saw above is the pullback via $SU \to U \to \bar{U}$ of the diagonal $\Delta \subset \bar{U}$.
So we need an appropriate divisor $\bar{{\mathcal B}} \subset \bar{{\mathcal U}}= \SM_{g_{0},1} \times_{\SM_{g_{0}}}   \SM_{g_{0},1}$. The first guess might be to consider the diagonal; but this is not a divisor, it is a submanifold of codimension $(1|1)$. Instead, we need to invoke the duality of section \ref{basics} for SRS's, which converts points to divisors on a SRS and the diagonal to a divisor $\bar{{\mathcal B}} \subset \bar{{\mathcal U}}$. We then define 
${{\mathcal B}} \subset {{\mathcal U}}$ as the inverse image of $\bar{{\mathcal B}}$.

(6)  and (7): We have now constructed all the input needed for Proposition \ref{branched covers of SRS}: a family of super Riemann surfaces, a branched covering of the reduced space with odd local degrees, and a thickening of the branch divisor.
 So we  get a  family of super Riemann surfaces that are branched covers of the original family, as claimed.
By the universal property of moduli spaces, this gives the desired morphism of supermanifolds, from $\wt{\SM}_{g_{0},1}$  to $\SM_{{g}}$.

\subsection{Components} \label{Components}

{\mbox{     }}
The spin moduli space $\sM_{g_0}$ has  two components $\sM_{g_0}^{\pm}$, distinguished by the parity of the spin structure. Therefore, the same holds for the supermoduli space ${\SM}_{g_0}$. Related spaces such as $\wt{\SM}_{g_{0},1}$ inherit at least two components. We will see here that there are actually more components than these obvious two:  

\begin{proposition}
$\wt{\SM}_{g_{0},1}^{+}$ has at least two components. Under the restriction  
$\wt{\SM}_{g_{0},1}^{+}  \to \SM_{{g}}$ of the map constructed in the previous section, these map to the two components 
$\SM_{{g}}^{\pm}$ of $\SM_{{g}}$.
\end{proposition}

\begin{proof}

We can see this very explicitly in case $g_0=1, d=3, g=2$. Here the base curve $C=E$ is elliptic. The branch point $p \in E$ determines a degree 2 map $E \to {\PP}^1$. 
Let $p_0, p_1, p_2, p_3=p$ be its four ramification points. We choose a coordinate $z$ on $\PP^1$ which takes the values $0,z_1,z_2, \infty$ at $p_0, p_1, p_2, p_3$, and write the equation of $E$ as 
$y^2 = z(z-z_1)(z-z_2)$.
The cover $\pi: \wt{C} \to E$ is essentially unique. It is given topologically as in section \ref{triple covers} by  mapping $\pi_1(E \setminus p)$  to the symmetric group  $S_3$ by sending $\alpha_1 \to (23), \quad \beta_1 \to (123)$. We can also describe it algebraically. Let $\pi_0: \wt{{\PP}^1} \to \PP^1$ be a triple cover which is totally branched over $\infty$ and has simple branching over $z_1,z_2$. This is accomplished by setting $z$ to be a cubic polynomial in the coordinate $w$ on $\wt{{\PP}^1}$:
$z=(w-w_3)(w-w_4)(w-w_5),$ 
whose critical values are  $z_1,z_2$ (and $\infty$, with multiplicity 2). For $i=1,2$, let $w_i, w'_i$, respectively, be the values of $w$ at the unramified and ramified points above $z_i$, so $z-z_i = (w-w_i)(w-w'_i)^2$.
We then take $\wt{C}$ to be the  normalization of the fiber product of $E$ and $\wt{\PP^1}$ over $\PP^1$. Explicitly, the fiber product  has the equation 
$y^2 = z(z-z_1)(z-z_2)$, and its normalization has the equation
$y^2 = (w-w_1)(w-w_2)(w-w_3)(w-w_4)(w-w_5)$.
In particular, this $\wt{C}$ is a genus 2 hyperelliptic curve, a double cover of $\wt{\PP^1}$ ramified at points $q_l, \quad l=1 \dots 5$ and $q$  over the  branch points $w_l$ and $\infty$. 
Note that for $i=1,2$, the inverse image of $w_i$ is $q_i$ (with multiplicity 2), but the inverse image of $w'_i$ consists of two distinct points $q'_i,q_{i}''$ forming a hyperelliptic pair.
The map to $E$ is totally ramified at $q$, whose image is the unique branch point $p$.

The three even spin structures on $E$ are $L_j=p_0-p_j$ for $j=1,2,3$. Our recipe for lifting to a spin structure on $\wt{C}$ is $\wt{L} := \pi^*(L)(q)$. If we start with $L_1=p_0-p_1$, we get 
$\wt{L_1} = {\OO}(q_3 +  q_4 + q_5 - q_1 - q_1' -q_1'' + q) = \OO(q_2),$ which is an odd spin structure on $\wt{C}$:
$h^0(\wt{L_3}) =1.$
Similarly, $\wt{L_2} =  \OO(q_1).$ But 
$\wt{L_3} = {\OO}(q_3 +  q_4 + q_5 -3q + q) = \OO(q_3 + q_4 -q_5)$ is an even spin structure on $\wt{C}$: $h^0(\wt{L_3}) =0.$

To extend to higher genera $g_0$ and degrees $d$, we need some basic facts about spin structures on {\em singular} algebraic curves. 
These facts are important in string theory. A mathematical version was obtained by Cornalba \cite{Cornalba}, who has constructed a compactified moduli space $\overline{\sM_g}$ of stable spin curves. This has been studied in \cite{Caporaso}, and a useful review is in \cite{Farkas}. The compactification has two components $\overline{\sM_g}^{\pm}$ which are compactifications of
${\sM_g}^{\pm}$ respectively. Each component maps onto the Deligne-Mumford compactified moduli space $\overline{\M_g}$ of stable curves. 
A spin structure on a stable curve $C$ consists of data $C',L',\beta$.
Here $C'$ is a `blow up' of $C$ along some subset $\Delta$ of the nodes of $C$:
take the partial resolution $N$ of $C$ at $\Delta$, and
for each $p_i \in \Delta$ attach to $N$ a smooth rational curve $R_i$
meeting $N$ transversally at the two branches above $p_i$.
(Such curves $C'$ are called decent or quasi stable.)
The remaining data consist of a line bundle $L'$ on $C'$
whose  total degree is $g-1$ and whose degree on each $R_i$ is 1, 
and a homomorphism $\beta: {L'}^{\otimes 2} \to \omega_{C'}$ 
that vanishes on the $R_i$ but on no other component of $C'$. 
Equivalenty and perhaps more intuitively, 
a spin structure on $C$ is specified by the torsion free sheaf
(not necessarily a line bundle)
$L := \nu_*L'$ on $C$, where $\nu: C' \to C$
is the map that collapses each $R_i$ to $p_i$.
This has rank 2 at the points of $\Delta$, and is locally free elsewhere.

The behavior is very simple over points of compact type in $\overline{\M_g}$, 
i.e. curves whose dual graph is simply-connected (a tree). 
In this case $\Delta$ must consist of all the singular points of $C$.
If such a curve $C$ is the union of irreducible components $C_i$ whose intersection pattern is determined by the dual graph, a spin structure  $L$ on $C$ is uniquely specified by a collection of spin structures $L_i$ on the $C_i$. The corresponding spin structure $L$ is the direct sum of the direct images of the $L_i$: it has rank 2 (i.e. fails to be a line bundle) at every node of $C$.  The parity of $L$ is the sum of the parities of the $L_i$.
(The extra complications for curves not of compact type arise from the possibility of the spin structure having rank 1 at some of the nodes. For
example,  for an irreducible $C$ with a single node, let $C'$ be its normalization and $p,q$ the points above the node. Then any square root of $K_{C'}(p+q)$ can be glued -- in two ways -- to give a line bundle spin structure on $C$.  The two cases that the spin structure has rank 2 or rank 1 at
a node correspond in string theory to a degeneration of NS or Ramond type. Ramond degenerations are not possible if the dual graph
is a tree because an irreducible curve always has an even number of Ramond punctures. The fact that Ramond degenerations occur in the natural compactification of the
spin moduli space $\sM_g$ -- and therefore also of the super Riemann surface moduli space $\SM_g$ -- can be regarded as the reason that it is necessary to study
Ramond punctures in string theory.)

We can now prove the proposition for triple covers of curves of  genus $g_0>1$. Let $C$ be a reducible curve consisting of an elliptic component $E$ meeting a genus $g_0 -1$ curve $C'$ in a single point $a$. For the triple cover $\wt{C} \to C$ we start with the triple cover of $E$ constructed above (but let us rename it $\wt{E}$), and glue it to three copies of $C'$ at the three inverse images of $a$. We choose even spin structures $L_E,L_{C'}$ on $E,C'$ and lift them to the four components $\wt{E}, C_j$ of $\wt{C}$. On the $C_j$ we get even spin structures, but on $\wt{E}$ we can get either even or odd spins. By going to smooth deformations of $C,\wt{C}$ we conclude that $\wt{\sM}_{g_0}^+$ is reducible, with components mapping to both components $\sM_g^{\pm}$ of $\sM_g$.
So $\wt{\SM}_{g_{0},1}^{+}$ has at least two components. Under the restriction  
$\wt{\SM}_{g_{0},1}^{+}  \to \SM_{{g}}$ of the map constructed in the previous section, these map to the two components 
$\SM_{{g}}^{\pm}$ of $\SM_{{g}}$ as claimed.

We need one more modification to allow arbitrary degrees $d \ge 3$. 
We again start with  a reducible curve $C$ consisting of an elliptic component $E$ meeting a genus $g_0 -1$ curve $C'$ in a single point $a$. We build the $d$ sheeted covering $\wt{C}$ by gluing the following pieces. Over $E$ we take the previous triple cover with single ramification point, $\wt{E}$, plus $d-3$ disjoint copies $E_i$ of $E$. Over $C'$ we take a $d-2$-sheeted unramified covering $\wt{C'}$ plus $2$ disjoint copies $C'_j$ of $C'$. We glue the $d-2$ points of a fiber of $\wt{C'} \to C'$ to one point on each of the $E_i$ and to just one point on $\wt{E}$. The two remaining points of $\wt{E}$ are glued to the two $C'_j$. This produces a simply connected dual graph, so we again specify spin structures on $C$ and $\wt{C}$ by specifying even spin structures $L_E, L_{C'}$ on $E$ and $C'$. Now the effect of switching our choice of $L_E$ from what we were calling $L_1$ to $L_3$ is to switch the parity on $\wt{E}$ without changing anything else. We conclude as before that in this general case too, $\wt{\SM}_{g_{0},1}^{+}$ has at least two components, which map to the two components $\SM_{{g}}^{\pm}$ of $\SM_{{g}}$ as claimed.
\end{proof}

\subsection{The normal bundle sequence} \label{decomposition}

\mbox {   }

Fix a $d$-sheeted branched cover of curves $\pi: \wt{C} \to C$, with branch divisor
$B=\sum_{i=1}^n p_i$ 
in $C$ whose points are distinct and labelled from $1$ to $n$, and ramification divisor
$R=\sum_{i,j}a_{i,j} \wt{p}_{i,j}$ 
in $\wt{C}$, with
$\sum_j a_{i,j} =d$ and $a_{i,j} \ge 1$. The  $a_{i,j} $ specify the {\em ramification pattern} of $\pi$:  $a_{i,j}$ is the number of sheets that come together at a ramification point $\wt{p}_{i,j}$. Let $g, \wt{g}$ be the genera of $C, \wt{C}$ respectively.

Now we allow the continuous parameters, i.e. the curves $ \wt{C},  C$  and hence
the map $\pi$ (and in particular also  the branch divisor $B$)  to vary, holding fixed the discrete data of the ramification pattern, i.e. the   $a_{i,j} $ and the genera $g, \wt{g}$. There is a moduli space $\wt{\M}_{g,n}$ parametrizing these covers. (Sometimes this is called a Hurwitz scheme.) It comes with a forgetful map
\begin{align*}
\wt{\M}_{g,n} &\rightarrow \M_{g,n}\\
(\pi: \wt{C} \to C) &\mapsto (C,B).
\end{align*}
This map is a local isomorphism, in fact an unramified finite cover: given $\pi: \wt{C} \to C$ and a small deformation of $(C,B)$, there is a unique lift to a deformation of $\pi: \wt{C} \to C$. (Note that the points of $B$ are not allowed to collide.) Let ${\mathcal C} \to \wt{\M}_{g,n} $ be the universal curve.

There is also a map
\begin{center}
\[
\iota: \wt{\M}_{g,n} \rightarrow  \M_{\wt{g}}
\]
\[
(\pi: \wt{C} \to C) \mapsto  \wt{C}.
\]
\end{center}
When $\wt{g}>1$, a curve $\wt{C}$ can have at most finitely many maps to curves $C$, and $\iota$ is a local embedding away from curves with extra automorphisms. The main result of this section is:
\begin{proposition}\label{decomposition of sequence} 
The normal bundle sequence of $\iota$:
\[
0 \to T{\wt{\M}_{g,n}} {\stackrel{\iota_*}{\longrightarrow}}  {T{ \M_{\wt{g}}}}_{|\wt{\M}_{g,n}} \longrightarrow N  \to 0
\]
splits.
\end{proposition}

\begin{proof}
First consider the special case when the map  $\pi: \wt{C} \to C$ is a Galois cover 
at one point of $\wt{\M}_{g,n} $, hence at all such points. 
In other words, assume there is a finite group $G$ which acts faithfully on $ \wt{C}$ with quotient $C =  \wt{C} / G$.
For now, assume also that $G$ acts on 
the universal curve ${\mathcal C} \to \wt{\M}_{g,n} $.
The action of $G$  on $ \wt{C}$ induces actions on $\pi_*\OO_{ \wt{C} }$ 
and on other natural objects such as   $\pi_*K_{ \wt{C} }^\ell$ (for various integer $\ell$) 
and $T_{ \wt{C} }\M_{\wt{g}} $.
Similarly, the action of $G$  on ${\mathcal C} \to \wt{\M}_{g,n} $ turns 
$ {T{ \M_{\wt{g}}}}_{|\wt{\M}_{g,n}}$  into an equivariant $G$-bundle whose typical fiber is 
$T_{ \wt{C} }\M_{\wt{g}} $. Therefore 
$ {T{ \M_{\wt{g}}}}_{|\wt{\M}_{g,n}}$ 
decomposes as a direct sum:
\[
 {T{ \M_{\wt{g}}}}_{|\wt{\M}_{g,n}} = \bigoplus_{  \rho \in G^{\vee}} V_{\rho} \otimes \rho.
\]
Here $\rho$ runs over the irreducible representations of $G$, and 
$ V_{\rho}  := \Hom^G(\rho, {T{ \M_{\wt{g}}}}_{|\wt{\M}_{g,n}})$
is the multiplicity bundle of $\rho$. 
Since the summand with $\rho = {\mathbf 1}$ is $V_{  {\mathbf 1}  } = T{\wt{\M}_{g,n}}$, 
we have a decomposition
\[
 {T{ \M_{\wt{g}}}}_{|\wt{\M}_{g,n}}  = T{\wt{\M}_{g,n}}  \oplus N
\]
where $N$ is the sum of the  remaining summands.

In the above we assumed that the action of $G$ on $\wt{C}$ extended to an action of $G$  on 
the universal curve ${\mathcal C} \to \wt{\M}_{g,n} $. This may not be the case: monodromy around 
$\wt{\M}_{g,n} $ may take the action of $G$ on $\wt{C}$ to another, conjugate action. 
The effect is that several of the $V_{\rho}$ may have to be combined. 
Nevertheless, our argument goes through since the trivial representation $\rho = {\mathbf 1}$ is not conjugate to any other.

In the general case, we can replace  $\pi: \wt{C} \to C$ by its Galois closure  $\hat{\pi}: \hat{C} \to C$:
Away from the branch locus, i.e. over the open subset $C_0 := C \setminus B$, let $\overline{C}_0$ be the $n!$-sheeted unramified cover whose fiber over $p \in C_0$ consists of the $n!$ ways of ordering the $n$ points of $\pi^{-1}(p)$. 
This $\overline{C}_0$ may be disconnected, so let $\hat{C}_0$ be a connected component. It is an unramified Galois cover of 
 $C_0$, with Galois group a subgroup of the symmetric group $S_n$.
 The complete curve $\hat{C}$ is the unique (non-singular)  compactification of $\hat{C}_0$. It is still Galois over  $C$ with the same $G$, but of course it is ramified. Denote its genus by $\hat{g}$.

In this situation, $G$ does not act on $\wt{C}$, nor on  ${T{ \M_{\wt{g}}}}_{|\wt{\M}_{g,n}}$. Nevertheless we can still make sense of the decomposition. 
We have an action of $G$ on $\hat{C}$ with quotient $C$. Let $H \subset G$ be the stabilizer of an unramified point of $\wt{C}$. Then 
$\hat{C}$ is also a Galois cover of $\wt{C} = \hat{C} / H $, with group $H$. 
For each  irreducible representation $\rho$ of $G$ there is an invariant subspace $\rho^H \subset \rho$, and we have compatible decompositions:
\begin{align*}
T_{\hat{C}}\M_{\hat{g}}  &= \bigoplus_{  \rho \in G^{\vee}} V_{\rho} \otimes \rho\\
T_{\wt{C}}\M_{\wt{g}}  &= \bigoplus_{  \rho \in G^{\vee}} V_{\rho} \otimes \rho^H. 
\end{align*}
Under these decompositions, the corresponding tangent space of ${\wt{\M}}_{g,n}$ is just 
\newline $ V_{\mathbf 1} \otimes {\mathbf 1} ^H=  V_{\mathbf 1} \otimes {\mathbf 1}$, so 
we have again exhibited  $T{\wt{\M}}_{g,n}$ as a direct summand of the restriction of $T\M_{\wt{g}}$.
\end{proof}

\subsection{Non-projectedness of $\SM_{g}$ and $\SM_{g,n}$}

\mbox {   }

We can now prove our main result, Theorem \ref{Main1}:  the non-projectedness of super moduli space $\SM_g$  for $g \geq 5$. 

\begin{proof} 

We do this by reducing from  $\SM_g$  to its submanifold  $\wt{\SM}_{2,1}$ constructed above. In fact, we now have in place all the ingredients for this reduction:

\begin{itemize}
\item the non-vanishing of the first obstruction class
$\omega_2(\wt{\SM}_{2,1})$. (The non-vanishing of $\omega_2({\SM_{2,1}})$ is Theorem \ref{Main2},  proved in  section \ref{marked non-splitness}. The lifting to the covering space $\wt{\SM}_{2,1}$ was seen in Corollary \ref{CoroCovering}.)
\item the inclusion of supermanifolds $\wt{\SM}_{2,1} \to \SM_g$,  proved in section \ref{maps}; and 
\item the decomposition of the restriction  to $\wt{\sM}_{2,1} $ of $T_{+}\SM_g$ into its tangential and normal pieces, proved in proposition \ref{decomposition of sequence} in section \ref{decomposition}.
\end{itemize}

These three ingredients are precisely the inputs of the Proposition \ref{submanifold splits}. We conclude the non-vanishing of the first obstruction class:
\[
\omega := \omega_2(\SM_g)  \neq 0 \in H^1(\sM_g, \Hom( \exterior{2} T_{-} \SM_g, T_{+} \SM_g)),
\]
and hence the non-projectedness of $\SM_g$.

\end{proof} 

We conclude with a proof of Theorem \ref{Main3}:

\begin{proof} 

We start with the extreme case $n=g-1$. Consider the space
$\wt{\SM}_{2,1}$ which parametrizes pairs $(C,p) \in \SM_{2,1}$ plus an unramified cyclic $n$-sheeted cover $\pi: \wt{C} \to C$. Note that the genus of $\wt{C}$ is $g=n+1$. There is a natural embedding $i: \wt{\M}_{2,1} \to \M_{g,n}$ sending the above data to the curve $\wt{C} $ with the $n$ marked points $\pi^{-1}(p)$. A priori, these $n$ points are only cyclically ordered, so there are $n$ distinct ways to order them. However, the cyclic automorphism group $\ZZ/n$ of $\wt{C}$ over $C$ permutes these $n$ orderings transitively, so we get a well defined image point in $\M_{g,n}$. In fact, the entire $\M_{g,n}$ admits an action of the cyclic group $\ZZ/n$ which cyclically permutes the $n$ marked points. Our locus $\wt{\M}_{2,1}$ is a component of the fixed locus of this action. It follows that the normal bundle sequence for the embedding of $\wt{\M}_{2,1}$ in $\M_{g,n}$ is split: the tangent bundle is the $+1$-eigenbundle, while the normal bundle is the sum of the remaining eigenbundles. Finally, the embedding of $\wt{\M}_{2,1}$ in $\M_{g,n}$ lifts, as in section \ref{maps}, to an 
embedding of $\wt{\SM}_{2,1}$ in $\SM_{g,n}$. As in the proof of Theorem \ref{Main1} above, we now have the three ingredients needed for
Proposition \ref{submanifold splits} to apply. We conclude the non-vanishing of the first obstruction class:
\[
\omega := \omega_2(\SM_{g,g-1})  \neq 0,
\]
and hence the non-projectedness of $\SM_{g,g-1}$.

For lower values of $n$, we map the image of $\wt{\M}_{2,1}$ to 
$\M_{g,g-1}$ as above, and then project to $\M_{g,n}$ by the map that 
preserves the first $n$ of the $g-1$ marked points and omits the rest. Again, there is no problem in lifting to a map of supermoduli spaces. We can no longer interpret $\wt{\M}_{2,1}$ as the fixed locus of a group of automorphisms, but nevertheless we can conclude the splitting of the normal bundle sequence:

Very generally, let $f: Y \to Z$ be a fibration, and $i: X \to Y$ an immersion, such that
$f \circ i : X \to Z$ is also an immersion.

We map the normal bundle sequence of bundles on $X$:
\begin{equation}\label{1}
0 \to T_X \to i^*{T_Y} \to N_{X,Y} \to 0        
\end{equation}

onto the normal bundle sequence:

\begin{equation}\label{2}
0 \to T_X \to (f \circ i)^*{T_Z} \to N_{X,Z} \to 0,       
\end{equation}

and note that the kernel sequence is trivial:

$$0 \to 0 \to i^*T_{Y/Z} \to i^*T_{Y/Z} \to 0.$$

Now if sequence \eqref{1} is split by a map

$$N_{X,Y} \to i^*{T_Y}$$

we get an induced splitting

$$N_{X,Z} = N_{X,Y}/i^*T_{Y/Z}  \to  i^*{T_Y}/i^*T_{Y/Z} = (f \circ i)^*{T_Z}$$

of sequence \eqref{2}.

We want to apply this to:

* $X = \wt{\M}_{2,1}$ which parametrizes pairs $(C,p) \in \M_{2,1}$ plus an unramified $(g-1)$-sheeted cover $\pi: \wt{C} \to C$. 

* $Y = \M_{g,g-1}$

* $Z = \M_{g,n}$, for $g-1  \geq n \geq 1$, with $f: Y \to Z$ preserving the first $n$ of the $g-1$ marked points and omitting the rest.

For this, we need to check that $i: X \to Y$ and the induced $f \circ i : X \to Z$ are immersions. Write $\pi^{-1}(p) = p_1 + \dots + p_{g-1}$. Then we need injectivity of the maps on tangent spaces:
$$ H^1(C,T_C (-p)) \stackrel{d(i)}{\to} H^1(\wt{C},T_{\wt{C}} (-\pi^{-1} p)) \stackrel{d(f)}{\to} H^1(\wt{C},T_{\wt{C}} (-(p_1 + \dots + p_n ))). 
$$
This  commutes with the map on tangents of the moduli spaces of curves without marked points:
$$ H^1(C,T_C) \to H^1(\wt{C},T_{\wt{C}}),
$$
and the latter is injective. These maps fit together into a commutative diagram with exact columns:

 \begin{equation*}
   \begin{matrix}  
H^0({T_C}_{|p})  & \to & \oplus_{j=1}^{g-1} H^0({T_{\wt{C}}}_{| p_j })  & \to & \oplus_{j=1}^n H^0({T_{\wt{C}}}_{| p_j })  \cr
\downarrow  &  & \downarrow &  &     \downarrow    \cr
H^1(C,T_C (-p))  & \stackrel{d(i)}{\to} & H^1(\wt{C},T_{\wt{C}} (-\pi^{-1} p)) & \stackrel{d(f)}{\to} & H^1(\wt{C},T_{\wt{C}} (-(p_1 + \dots + p_n )))  \cr
 \downarrow  &  & \downarrow &  &     \downarrow    \cr
 H^1(C,T_C) &  \to & H^1(\wt{C},T_{\wt{C}}) & = &  H^1(\wt{C},T_{\wt{C}}).   \end{matrix}       \end{equation*} 

injectivity of the bottom map implies that $\Ker (d(f \circ i ))$ must come from the vertical direction $H^0({T_C}_{|p}) = \Ker (H^1(C,T_C (-p)) \to H^1(C,T_C))$. But for each $j$, the map
$H^0({T_C}_{|p}) \to H^0({T_{\wt{C}}}_{|p_j})$ is an isomorphism, so $\Ker (d(f \circ i ))$ vanishes, showing that $f \circ i $ is indeed an immersion.

This shows that the normal sequence for $\wt{\M}_{2,1}$ in $\M_{g,n}$  splits for $g-1  \geq n \geq 1$, so the non-vanishing of the obstruction for $\wt{\M}_{2,1}$ implies the same for $\M_{g,n}$. The theorem now follows as before from Proposition \ref{submanifold splits}.   (Note that the argument fails for $n=0$, because the natural map $f \circ i : \wt{\M}_{2,1} \to \M_g$ factors through $\wt{\M}_2$ and is therefore not an immersion.)
\end{proof} 

In stating this argument, we have ignored the fact that particular genus 2 curves with a marked point have exceptional
automorphisms.  To justify what we have asserted, one may either develop the theory for orbifolds, or restrict
from $X$ to an open subset of $X$ that contains a fiber of $\wt{\M}_{2,1}\to \M_{2}$.

\section{Acknowledgments}
We thank Pierre Deligne for a careful reading of an earlier version of the manuscript, and for correcting several statements and proofs, especially in section 2.
We are grateful to Gavril Farkas,  Dick Hain, Brendan Hassett, Sheldon Katz, Igor Krichever, Dimitry Leites, Yuri Manin, Tony Pantev, Albert Schwarz,  Liza Vishnyakova and Katrin Wendland for helpful discussions.
RD acknowledges partial support by NSF grants DMS 0908487 and RTG 0636606. EW acknowledges partial support by NSF Grant PHY-0969448.

\appendix

\section{A detailed example in genus 5}

In this appendix, we give an elementary construction of a family of triple covers $\wt{C} \to C$ with a single branch point, where $g(\wt{C})=5, ~  g(C)=2$. We analyze the parameters on which this construction depends and find, somewhat surprisingly, that the parameter space of such covers with fixed $C$ is a {\em rational} curve in $\M_5$. This curve has some orbifold points, so it maps to the moduli space but not to the moduli stack. It has a cover of genus 19 over which a family of genus 5 triple covers exists.
Finally, we examine the effect of adding spin structures to our curves, and verify that even spin structures on the genus 2 curve $C$ can lead to both even and odd  spin structures on the genus 5 curve $\wt{C}$.

\subsection{The Galois closure}
A triple cover $\rho: \wt{C} \to C$ with a single branch point $p \in C$ cannot be cyclic, so its Galois group is the symmetric group $S_3$ of permutations of $\{1,2,3\}$. Its Galois closure is therefore a  smooth
curve   ${\overline{\wt{C}}} $  on which $S_3$ acts. 
(One way to obtain ${\overline{\wt{C}}} $ explicitly is by taking  the self product  $\wt{C} \times_C \wt{C} $, removing the diagonal $\wt{C} $, and taking the unique smooth compactification.)
The quotient by $S_3$ is the original $C$, and the quotient by the subgroup $S_2$ of permutations of $\{1,2\}$ is the original $\wt{C}$. This subgroup is not normal: there are three conjugate  subgroups $({S_2})_i, ~ i=1,2,3$, and corresponding quotient curves ${\wt{C}}_i$. These are isomorphic to each other:  if $\{i,j,k\}$ is a permutation of $\{1,2,3\}$, the involution $\overline{\tau}_k$ on ${\overline{\wt{C}}}$ induced by the transposition $(ij)$ exchanges ${\wt{C}}_i$ and ${\wt{C}}_j$. But if we divide by the alternating subgroup $A_3 \subset S_3$, we get a new intermediate curve $\overline{C} $ which is a double cover  of $ C$ and a triple quotient $\bar{\rho}: {\overline{\wt{C}}} \to \overline{C} $
of ${\overline{\wt{C}}} $. It comes with an involution $\tau$, induced by any of the $\overline{\tau}_k$.
A point of ${\overline{\wt{C}}} $ over some $q \neq p$ in $C$ can be thought of as a labeling or ordering of the 3 points 
in $\rho^{-1}(q)$. The action of $S_3$ permutes the labels, and a point of the quotient $\overline{C}$ above $q$ can be thought of as an orientation, or cyclic ordering, of that fiber. An easy monodromy argument shows that above $p$ there is a single point $\wt{p}_i$ in each ${\wt{C}}_i$,  two points 
${\overline{\bar{p}}},  \quad {\overline{\bar{p}}}'=\tau {\overline{\bar{p}}} \in \overline{C}$, and two points 
${\overline{\wt{p}}},  {\overline{\wt{p}}}'  \in    {\overline{\wt{C}}}$. (Each of the transpositions $\tau_k = (ij)$ exchanges
${\overline{\wt{p}}}, {\overline{\wt{p}}}' $, while the 3-cycles preserve them.) It follows that  ${\overline{\wt{C}}} $ is an unramified double cover of each ${\wt{C}}_i$, as is $\overline{C} $ over $C$. We display these curves, their genera and the maps between them in the following snapshot:

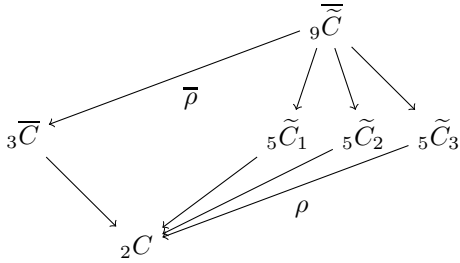
\begin{center}

\begin{tikzpicture}[node distance=2cm, auto]
 \node (C) {${}_2 C$};
  \node (C-) [left of=C, above of=C, node distance=1.5cm] {${}_3{\overline{{C}}}$};
  \node [above of=C-, node distance=1.5cm] (G) {};
  \node [right of=G, node distance=4cm] (C1-) {${}_9{\overline{\wt{C}}} $};
  \node [right of=C1-, below of=C1-,  node distance=1.5cm] (C3) {${}_5{\wt{C}}_3$};
  \node [left of=C3,   node distance=1cm] (C2) {${}_5{\wt{C}}_2$};
  \node [left of=C2,   node distance=1cm] (C1) {${}_5{\wt{C}}_1$};
  \draw[->] (C1) to node {} (C);
  \draw[->] (C2) to node {} (C);
  \draw[->] (C3) to node {$\rho$} (C);
  \draw[->] (C-) to node {} (C);
  \draw[->] (C1-) to node {$\bar{\rho}$ } (C-);
  \draw[->] (C1-) to node {} (C1);
  \draw[->] (C1-) to node {} (C2);
  \draw[->] (C1-) to node {} (C3);
\end{tikzpicture}

\end{center}

\subsection{The construction}\label{sedrftgyhu}

We now reverse the above analysis, obtaining a direct construction of the  triple covers $\rho: \wt{C} \to C$ with a single branch point $p \in C$. This will allow us to describe the parameter spaces on which the construction depends.

Start with a genus 2 curve $C= ~ _2C$ and an unramified double cover $_3\overline{C} \to C$, with fixed point free involution $\tau: \overline{C} \to \overline{C} $. (Since $C$ has genus 2, the genus of $\overline{C} $ is $3= 2 \times 2 - 1$, which we indicate with the left subscript.)   Given a point $\bar{p}    \in  \overline{C} $ and some additional data, we construct a cyclic triple cover $\bar{\rho}: {\overline{\wt{C}}}  \to\overline{C} $ which is totally ramified over $\bar{p}$ and (with the opposite orientation) over $\tau \bar{p}$, with a deck transformation
$\sigma: {\overline{\wt{C}}}  \to {\overline{\wt{C}}}$.
The extra data consists of a line bundle $L \in \bar{J} :=\Pic^0 ({\overline{{C}}})$ with an isomorphism 
\begin{equation} \label{17}
L^{\otimes 3} \cong \OO_{\overline{C} } (\bar{p}-\tau \bar{p}) ,
\end{equation}
which we interpret as a 3 to 1 map from the total space of $L$ to the total space of 
$\OO_{\overline{C} } (\bar{p}-\tau \bar{p}) $. 
Now $\OO_{\overline{C} } (\bar{p}-\tau \bar{p}) $ has a meromorphic section  corresponding to the section 1 of $\OO_{\overline{C} }$, and we let ${\overline{\wt{C}}}$ be its inverse image in  the total space of $L$, a cyclic triple cover of  ${\overline{{C}}}$. The automorphism $\sigma$ is induced from multiplication by a cubic root of unity on $L$.

The involution $\tau: \overline{C} \to \overline{C} $ lifts to an involution  
$ \wt{ \tau}:   {\overline{\wt{C}}} \to {\overline{\wt{C}}}$ if and only if $\tau^* L$ is isomorphic either to $L$ or to $ L^{-1}$. In the former case,  $ \wt{ \tau}$ commutes with $\sigma$, so the resulting 
$ {\overline{\wt{C}}}$ is a Galois cover of $C$ with Galois group $\ZZ/6$. 
In the latter case,  $\wt{ \tau}$ commutes  $\sigma$ to $\sigma^{-1}$, so the resulting 
$ {\overline{\wt{C}}}$ is a Galois cover of $C$ with Galois group the symmetric group $S_3$. 
(In general, the line bundle $\tau^* L^{-1}$ also satisfies condition \eqref{17} and therefore defines  another cyclic triple cover ${\overline{\wt{C'}}}  \to\overline{C} $ with the same ramification pattern as ${\overline{\wt{C}}}$ . The involution $\tau: \overline{C} \to \overline{C} $ always lifts to an isomorphism ${\overline{\wt{C}}} \to {\overline{\wt{C'}}} $, in fact to three of them, and when 
\begin{equation} \label{18}
\tau^* L^{-1} \cong L\end{equation}
these give three involutions of ${\overline{\wt{C}}}$, each conjugating $\sigma$ to $\sigma^{-1}$.)

\subsection{A rational curve in ${\M}_5$}

We will now analyze the parameters on which the construction of the previous section depends. Perhaps surprisingly, we find that the compact curve in moduli space parametrizing our triple covers (of a fixed curve $C$, and corresponding to a specified double cover $\overline{C} \to C$) is actually {\em rational}. 

 We need to choose a point $\bar{p} \in \overline{C}$ and a cubic root $L$ of $\OO_{\overline{C} } (\bar{p}-\tau \bar{p}) $ as in \eqref{17}. The set of those roots is a coset of the subgroup of points of order 3 in the Jacobian $\bar{J} = \Pic^0 ({\overline{{C}}})$, which is isomorphic to $(\ZZ/3)^6$. So our data seems to live in  a cover of $\overline{C}$ of degree $3^6$. But this cover turns out to be reducible, and our additional condition \eqref{18} picks out a subcover of degree 9. In order to see this, we need to review some general results on Prym varieties of unramified double covers.

\subsubsection{Pryms} \label{wesdr1}
Condition \eqref{18} is equivalent to $L \in \Ker(1+\tau^*)$, where $1+\tau^*$ is the endomorphism of $\bar{J}$ sending $L$ to $L \otimes \tau^* L$. For a general unramified double cover $\pi: \overline{C} \to C$ with involution $\tau$, Mumford \cite{Mumford} described $\Ker(1+\tau^*)$. It consists of four cosets of the {\em Prym} variety 
\begin{equation}\label{Prym}
P:= (1-\tau^*)\bar{J}.
\end{equation}
It is convenient to use the {\em Norm} map ${\text{Nm}}(\pi): \bar{J} \to J$, which is the homomorphism that sends a line bundle $\OO_{\overline{C}}(\bar{D})$ to $\OO_{{C}}({D})$ where $D:=\pi(\bar{D})$ is the image of the divisor $D$. Note that  $\pi^* \circ {\text{Nm}}(\pi) = 1 + \tau^*$, so  $\Ker(1+\tau^*) $ contains 
 $\Ker({\text{Nm}}(\pi)) $. In fact,
$\Ker(1+\tau^*) = ( {\text{Nm}}(\pi))^{-1} (K) $, where 
\[
K := \Ker(\pi^* : J \to \bar{J}) \cong \ZZ/2 .
\] 
Finally, 
$\Ker({\text{Nm}}(\pi)) =P \cup P'$ where $P'$ is the coset: 
\begin{equation}\label{P'}
P':= (1-\tau^*)\Pic^1(\overline{C})
\end{equation}
of $P$. (More generally, 
$1-\tau^*$ maps divisors of even degree on $\overline{C}$ to $P$, and divisors of odd degree to $P'$.)  So all in all we have:
\[\
\Ker({\text{Nm}}(\pi)) =P \cup P' = (1-\tau^*) \Pic( \overline{C}  )
\]
and 
\[
0 \to \Ker({\text{Nm}}(\pi)) \to \Ker(1+\tau^*) \to K \to 0.
\]

We now return to the conditions on our line bundle $L$. Condition \eqref{18} says that $L$ is in $\Ker(1+\tau^*)$ which, as we have just seen, consists of four components. The map $L \to L ^{\otimes 3}$ sends each of these components to itself. Since the right hand side of \eqref{17} is in $P'$ by \eqref{P'}, we see that $L$ must be in $P'$ too.

\subsubsection{Hyperelliptic Pryms} \label{wesdr2}
In our case, we can make everything more explicit. Our genus 2 base curve $C$ is hyperelliptic, a double cover of $\PP^1$ branched at the 6 points $0, \infty, e_0=1, e_1, e_2, e_3$. The double cover $\overline{C}$ is determined by a $4+2$ partition of these 6 branch points: say $0, \infty$ vs. $e_0, e_1, e_2, e_3$. The double cover of $\PP^1$ branched at $0, \infty$ is a rational curve $R=  {{_0}R}$, and the double cover of $\PP^1$ branched at $e_0, e_1, e_2, e_3$ is a genus 1 curve $E=  {{_1}E}$. The cover $\overline{C}$ has three involutions $\tau_0 ,\tau_1,\tau_2=\tau$, sitting in
a symmetry group $\ZZ/2 \times \ZZ/2$, with quotient $\PP^1$ and intermediate quotients $R, E, C$. We can include these in our snapshot:

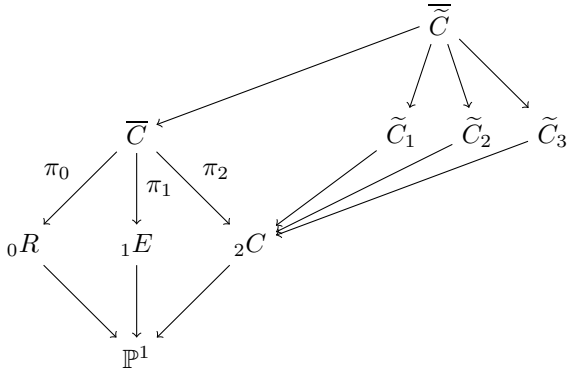
\begin{center}

\begin{tikzpicture}[node distance=2cm, auto]
 \node (C) {${{_2}C}$};
  \node (C-) [left of=C, above of=C, node distance=1.5cm] {${\overline{{C}}}$};
  \node [above of=C-, node distance=1.5cm] (G) {};
  \node [right of=G, node distance=4cm] (C1-) {${\overline{\wt{C}}} $};  
  \node [right of=C1-, below of=C1-,  node distance=1.5cm] (C3) {${\wt{C}}_3$};
  \node [left of=C3,   node distance=1cm] (C2) {${\wt{C}}_2$};
  \node [left of=C2,   node distance=1cm] (C1) {${\wt{C}}_1$};
  \node (P1) [left of=C, below of=C, node distance=1.5cm] {${\PP}^1$};
  \node (R) [left of=P1, above of=P1, node distance=1.5cm] {${{_0}R}$};
  \node (E) [ above of=P1, node distance=1.5cm] {${{_1}E}$};
  \draw[->] (C) to node {} (P1);
  \draw[->] (E) to node {} (P1);
  \draw[->] (R) to node {} (P1);

  \draw[->] (C1) to node {} (C);
  \draw[->] (C2) to node {} (C);
  \draw[->] (C3) to node {} (C);
  \draw[->] (C-) to node {$\pi_2$} (C);
  \draw[->] (C-) to node {$\pi_1$} (E);
  \draw[->,swap] (C-) to node {$\pi_0$} (R);
  \draw[->] (C1-) to node {} (C-);
  \draw[->] (C1-) to node {} (C1);
  \draw[->] (C1-) to node {} (C2);
  \draw[->] (C1-) to node {} (C3);
\end{tikzpicture}

\end{center}

In order to avoid clutter, we will show only one of the 3 quotients ${\wt{C}}_i$, which we rename ${\wt{C}}$:

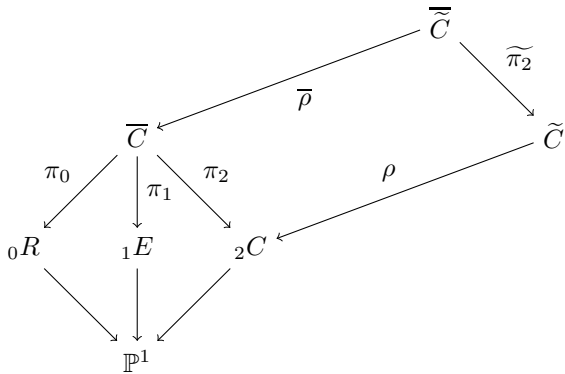
\begin{center}

\begin{tikzpicture}[node distance=2cm, auto]
 \node (C) {${{_2}C}$};
 \node (C-) [left of=C, above of=C, node distance=1.5cm] {${\overline{{C}}}$};
 \node [above of=C-, node distance=1.5cm] (G) {};
  \node [right of=G, node distance=4cm] (C1-) {${\overline{\wt{C}}} $};  
  \node [right of=C1-, below of=C1-,  node distance=1.5cm] (C3) {${\wt{C}}$};
 
  \node (P1) [left of=C, below of=C, node distance=1.5cm] {${\PP}^1$};
  \node (R) [left of=P1, above of=P1, node distance=1.5cm] {${{_0}R}$};
  \node (E) [ above of=P1, node distance=1.5cm] {${{_1}E}$};
  \draw[->] (C) to node {} (P1);
  \draw[->] (E) to node {} (P1);
  \draw[->] (R) to node {} (P1);
  \draw[->,swap] (C3) to node {$\rho$} (C);
  \draw[->] (C-) to node {$\pi_2$} (C);
  \draw[->] (C-) to node {$\pi_1$} (E);
  \draw[->,swap] (C-) to node {$\pi_0$} (R);
  \draw[->] (C1-) to node {$\bar{\rho}$} (C-);
  \draw[->] (C1-) to node {$\wt{\pi_2}$} (C3);
\end{tikzpicture}

\end{center}

Since any degree 0 line bundle on $\overline{C}$ can be written as the sum of pullbacks from the quotients:
\[
\Pic^0(  \overline{C} ) =(\pi_1)^* (\Pic^0(E))  +  (\pi_2)^* (\Pic^0 (C))  
\]
and $1-\tau^*$ kills $\Pic^0 (C) $, we get from\eqref{Prym} an isomorphism:
\[
P = (1-\tau^*)\bar{J} \cong (1-\tau^*) (\pi_1)^* \Pic^0(E).
\]
But since $(1-\tau^*) (\pi_1)^*  = 2  (\pi_1)^*$ on $\Pic^0(E)$ and $(\pi_1)^*$ is injective, we see that 
\[
(\pi_1)^*: \Pic^0(E) \to P
\]
is an isomorphism. Similarly we find that $E= \Pic^1(E)$ can be naturally identified with $P'$ via a translate of $(\pi_1)^*$. The map is:
\begin{equation}\label{zesdrvbh}
e \mapsto L_e := (\pi_1)^*\OO_E(e) \otimes (H_{\overline{C}})^{-1},
\end{equation}
where $H_{\overline{C}} := (\pi_0)^*\OO_{R}(1)$ is the hyperelliptic line bundle on $\overline{C}$, and it needs to be inserted in the above formula in order to yield a divisor of degree 0 on $\overline{C}$.

\subsubsection{The parameter space}

We can now describe the parameter space for our covers $\wt{C}$. Originally, we wanted pairs $(L,\overline{p})$ satisfying conditions  \eqref{17}, \eqref{18}. The space of eligible line bundles $L$ was identified in section \ref{wesdr1} with the shifted Prym $P'$. In our hyperelliptic setting this was translated in section \ref{wesdr2} to $E = P'$, the isomorphism being given by \eqref{zesdrvbh}. Putting these together, we need to parametrize pairs $(e,\overline{p})$ satisfying the condition:
\begin{equation}\label{zesdh}
 (\pi_1)^*\OO_E(3e) \cong  \OO_{\overline{C} } (\bar{p}-\tau \bar{p}) \otimes (H_{\overline{C}})^{3}.
\end{equation}
This is an equation in the Picard of $\overline{C}$. Keeping in mind that $ (\pi_1)^*$ is injective, this is also equivalent to the equation in $E$:
\begin{equation}\label{zesh}
3e \sim \pi_1\overline{p} + H_E,
\end{equation}
where $\sim$ means linear equivalence on $E$, and we have used that for any $\bar{p} \in \bar{C}$, 
\begin{equation}\label{esfcv}
H_{\overline{C}} \cong \OO(\tau_2\bar{p} + \tau_1\bar{p}).
\end{equation}

In order to parametrize solutions of \eqref{zesh}, consider another copy of $E$, say $E_r$. We think of it as parametrizing cubic {\em roots} of points of $E$. The curves $E_r, E$ are  isomorphic, but we find it convenient to keep the distinction. Let 
$m_1: E_r \stackrel{\cong}{\to} E$ be the isomorphism,
and $m_3$  the multiplication by 3 map:
\[
m_3: E_r \to \Pic^1(E) = E, \quad  e \mapsto \OO_E(3e) \otimes (H_E)^{-1},
\]
where $H_E$ is the hyperelliptic line bundle on $E$, pullback of $\OO_{\PP^1}(1)$.  (Note that we have not chosen a base point in $E$, only a degree 2 line bundle $H_E$, or equivalently the map to $\PP^1$. The  cubing map $m_3$ is nevertheless well defined.) We see that the natural  parameter space for our triple covers $\wt{C} \to C$ is the fiber product:
\[
 \overline{C}_r  := E_r \times_E \overline{C}.
\]
This is a 9-sheeted unramified cover of $\overline{C}$, so its genus is $19=1+9 \times (3-1)$. Locally over 
$ \overline{C}_r $ we can construct the family of triple covers ${\overline{\wt{C}}}  \to\overline{C} $ and their quotients $\wt{C}$. Since the generic curve $\wt{C}$ has no non-trivial automorphisms, these local families automatically glue to a family of triple covers $\wt{C} \to C$ parametrized by $ \overline{C}_r $.

The resulting map $\sigma: \overline{C}_r \to \M_5$ is clearly not an embedding. For one thing, the pairs $(L,\bar{p})$ and 
$(L^{-1}=\tau^*L, \tau\bar{p})$ give isomorphic covers. To understand the quotient, we note that the involution on $E$ with quotient $\PP^1$ lifts to an involution on $E_r$, so let $\PP^1_r$ be the quotient. It is  a 9 sheeted branched cover of the original $\PP^1$,  with ramification pattern $(2^4,1)$ over each of the 4 branch points $e_i$ of $E$ over $\PP^1$,  and $E_r$ is recovered as the normalization of the fiber product $ \PP^1_r  \times_{\PP^1} E$. We can complete this into a commutative box:

\begin{center}

\begin{tikzpicture}[
  back line/.style={densely dotted},
  cross line/.style={preaction={draw=white, -,line width=6pt}}]
  \node (A) {$E_r$};
  \node [right of=A, node distance=4cm] (B) {$E$};
  \node [below of=A, node distance=2cm] (C) {$\PP^1_r$};
  \node [right of=C, node distance=4cm] (D) {$\PP^1$};
 
  \node (A1) [right of=A, above of=A, node distance=1cm] {$\overline{C}_r$};
  \node [right of=A1, node distance=4cm] (B1) {$\overline{C}$};
  \node [below of=A1, node distance=2cm] (C1) {$C_r$};
  \node [right of=C1, node distance=4cm] (D1) {$C$};
 
  \draw[back line] (D1) -- (C1) -- (A1);
  \draw[back line] (C) -- (C1);
  \draw[cross line] (D1) -- (B1) -- (A1) -- (A) ;
  \draw (D) -- (D1) -- (B1) -- (B);
  \draw[->] (A) to node {} (B);
  \draw[->] (A) to node {} (C);
  \draw[->] (C) to node {} (D);
  \draw[->] (B) to node {} (D);
 \draw[->] (A1) to node {} (A);
 \draw[->] (A1) to node {} (B1);
  \draw[->] (B1) to node {} (B);
  \draw[->] (B1) to node {} (D1);
  \draw[->] (D1) to node {} (D);
\end{tikzpicture}

\end{center}

where the horizontal maps have degree 9, the others have degree 2. We see that the quotient of $\overline{C}_r$ by the above involution is what we have now labeled $C_r := \PP^1_r \times_{\PP^1} C $, a 9-sheeted unramified cover of $C$, hence of genus 10. However, there is a further symmetry: 
the hyperelliptic involution of $C$. Dividing by that, we see that the map of our family $\overline{C}_r$ to $\M_g$ factors through the rational curve $\PP^1_r$, as claimed. Pictorially, the hyperelliptic involutions of $C$ and $E$ generate a group $\ZZ/2 \times \ZZ/2$ which acts on the entire box. In particular it acts on $\overline{C}_r$ which parametrizes the triple covers $\wt{C}$, and the map $\sigma: \overline{C}_r \to \M_5$ is invariant under this action, so 
the curves $\wt{C}$ over points in an orbit of $\ZZ/2 \times \ZZ/2$ are isomorphic.

\subsection{The family}

We are going to construct universal curves over $\overline{C}_r,$ i.e. surfaces 
\[
\overline{\wt{\mathcal{C}}},   ~ \wt{\mathcal{C}}
\]
and fibrations
\[ 
\overline{\wt{\mathcal{C}}} \to \wt{\mathcal{C}} \to \overline{C}_r,
\]
such that the fibers over each point of $\overline{C}_r$ are isomorphic to the corresponding curves
$ \overline{\wt{{C}}} \to \wt{{C}}$ constructed above.
More precisely:

\begin{proposition}
 There is a commutative diagram:

\begin{center}

\begin{tikzpicture}[ auto]
  \node (P) {$\overline{\wt{\mathcal{C}}}$};
  \node (B) [node distance=3cm,right of=P] {$\wt{\mathcal{C}}$};
  \node (A) [node distance=2cm,below of=P] {$ \overline{C}_r \times  \overline{C}$};
  \node (C) [node distance=2cm,below of=B] {$\overline{C}_r \times {C},$};
  \draw[->] (P) to node {$\wt{\pi_2}$} (B);
  \draw[->] (P) to node [swap] {$\bar{\rho}$} (A);
  \draw[->] (A) to node [swap] {$1 \times \pi_2$} (C);
  \draw[->] (B) to node {$\rho$} (C);
\end{tikzpicture}

\end{center}

where $\overline{\wt{\mathcal{C}}},   ~ \wt{\mathcal{C}}$ are smooth surfaces,
the vertical maps are branched triple covers, and the fibers of $\overline{\wt{\mathcal{C}}},   ~ \wt{\mathcal{C}}$ over points of $\bar{C}_r$ are the triple covers $\bar{\tilde{C}} \to \bar{C}, \quad \wt{C} \to C$, with two (respectively one) total ramification points, constructed in section \ref{sedrftgyhu}.
\end{proposition}

\begin{proof}

To construct $\overline{\wt{\mathcal{C}}}$ we need a line bundle $\mathcal{L}$ on $ \overline{C}_r \times  \overline{C}$ satisfying the analog of \eqref{17}:
\begin{equation}\label{lmkij}
{\mathcal{L}} ^{\otimes 3} \cong p_{23}^* \OO(\Delta - \Delta') =: \mathcal{D},
\end{equation}
where $p_{23}$ is the projection:
\[
p_{23} : \overline{C}_r \times  \overline{C} = E_r \times_E \overline{C} \times  \overline{C} \to 
\overline{C} \times  \overline{C}, 
\]

$\Delta \subset \overline{C} \times  \overline{C}$ is the diagonal, and 
$ \Delta'$ is the graph of $\tau_2: \overline{C} \to  \overline{C}$.

We claim that \eqref{lmkij} is satisfied by the choice:
\[
{\mathcal{L}} := \OO(\Gamma) \otimes pr_2^* \pi_0^*\OO_R (-1) \otimes pr_1^* \pi_{0,r}^*\OO_{R_r} (-1),
\]
where the maps are:
\[
\overline{C}_r \times  \overline{C} 
\stackrel{pr_2} {\longrightarrow}
\overline{C} \stackrel{\pi_0} {\longrightarrow} R
\]
and
\[
\overline{C}_r \times  \overline{C} 
\stackrel{pr_{1}} {\longrightarrow}
\overline{C}_r \stackrel{\pi_{0,r}} {\longrightarrow} R_r,
\]
while $\Gamma \subset  \overline{C}_r \times  \overline{C} = E_r \times_E \overline{C} \times  \overline{C} $ is the effective divisor:
\[
\Gamma := \quad \{   (e,\bar{p},\bar{q})  \quad | \quad \pi_1 \bar{p} = m_3 e, ~ \pi_1 \bar{q} = e  \}.
\]

To prove \eqref{lmkij}, it suffices to verify that it holds when restricted to each horizontal curve $\overline{C}_r \times \{\bar{q} \} $ and  each vertical curve $\{ (e,\bar{p}) \}  \times  \overline{C} $.  Indeed, on the vertical curve ${\mathcal{L}}$ becomes 
\[
\pi_1^*\OO_E(e)  \otimes \pi_0^*\OO_R(-1) 
\]
while $\mathcal{D}$ becomes
\[
\OO_{\bar{C}}(\bar{p} - \tau_2\bar{p}),
\]
so the equality is just the condition \eqref{zesdh}. For the horizontal curves we need to work a little harder. It is convenient to focus on the diagram:

\begin{center}
\begin{equation}\label{pojhg}
\begin{tikzpicture}[ auto]
  \node (P) {$\bar{C}_r = E_r \times_{E}\bar{C}$}; 
  \node (B) [node distance=3cm,right of=P] {$\bar{C}$};
  \node (A) [node distance=2cm,below of=P] {${E}_r $};
  \node (C) [node distance=2cm,below of=B] {$E,$};
  \draw[->] (P) to node {${p_2}$} (B);
  \draw[->] (P) to node [swap] {$p_1$} (A);
  \draw[->] (A) to node [swap] {$m_3$} (C);
  \draw[->] (B) to node {$\pi_1$} (C);
\end{tikzpicture}
\end{equation}

\end{center}
and to recall that for any $e \in E$, 
\begin{equation}\label{pog}
m_3^*\OO_E(e) \cong m_1^*\OO_E(3e)  \otimes H_{E_r}^3.
\end{equation}

Now on the horizontal curve  $\overline{C}_r \times \{\bar{q} \} $, the line bundle
${\mathcal{L}}$ becomes 
\begin{equation}\label{dwtresy}
p_1^* m_1^*\OO_E(\pi_1\bar{q}) \otimes H_{\bar{C}_r}^{-1},
\end{equation}
while $\mathcal{D}$ becomes
\begin{align*}\label{dsses} 
p_2^* \OO_{\bar{C}}(\bar{q} - \tau_2 \bar{q})  &=  \quad \quad  \text{by identity \eqref{esfcv}}  \\
p_2^* (\OO_{\bar{C}}(\bar{q} + \tau_1 \bar{q}) \otimes H_{\bar{C}}^{-1}  ) &= \quad \quad \\
p_2^* \OO_{\bar{C}}(\pi_1^{-1} \pi_1\bar{q}) \otimes H_{{\bar{C}}_r}^{-9}  &=  \\
p_2^* \pi_1^*\OO_{ E  }(\pi_1\bar{q}) \otimes H_{{\bar{C}}_r}^{-9}  &= \quad \quad  \text{by commutativity of \eqref{pojhg}} \\
p_1^* m_3^*\OO_E(\pi_1\bar{q}) \otimes H_{\bar{C}_r}^{-9} &= \quad \quad  \text{by  \eqref{pog}} \\
p_1^* m_1^*\OO_E(3\pi_1\bar{q}) \otimes H_{\bar{C}_r}^{-3}, 
\end{align*}
showing the needed equality of the vertical restrictions of ${\mathcal{L}}^{\otimes 3}$ and $\mathcal{D}$.
We therefore have the right line bundle $\mathcal{L}$, so we get the desired triple cover 
$\bar{\rho}: \overline{\wt{\mathcal{C}}} \to \overline{C}_r \times \bar{C}.$

Next, we want to lift the involution $\tau = \tau_2: \bar{C} \to \bar{C}$ 
to an involution $\wt{\tau}:  \overline{\wt{\mathcal{C}}} \to  \overline{\wt{\mathcal{C}}}$ which would allow us to construct the quotient
$\wt{\pi_2}: \overline{\wt{\mathcal{C}}} \to \wt{\mathcal{C}}.$ 
For this we need to know that $\mathcal{L}$ satisfies the global analog of condition \eqref{18} as well. Again, it suffices to check this on horizontal and vertical curves. On horizontal curves, this follows immediately from \eqref{dwtresy}. On vertical curves, this is the original condition \eqref{18}. This completes the construction of  the universal curves  
$\overline{\wt{\mathcal{C}}},   ~ \wt{\mathcal{C}}$
over $\overline{C}_r.$

\end{proof}

Recall that we have an action of the group $\ZZ/2 \times \ZZ/2$ on 
$\overline{C}_r$ with quotient $\PP^1_r$.
This action lifts to $\overline{\wt{\mathcal{C}}},   \wt{\mathcal{C}}$, but it has fixed points there, so the smooth family does not descend to one over $\PP^1_r$. (In fact, the base curve $B$ of any non-locally trivial  family of smooth curves must have genus $g(B) \ge 2$, since the period map lifts to a non-constant map from the universal cover of $B$ to a bounded domain.)

\subsection{Adding spin}

We have constructed  families $\bar{\tilde{\mathcal{C}}}$ and $\wt{\mathcal{C}}$ of curves $\bar{\tilde{C}}$ and $\wt{C}$, parametrized by $\bar{C}_r$. We want to promote these to families of spin curves $(\bar{\tilde{C}},\bar{\tilde{N}})$ and $(\wt{C},\wt(N))$, making the construction of section \ref{spin maps} explicit. The further promotion to
super Riemann surfaces then follows Section \ref{maps}. We will in particular recover the result of Section \ref{Components} which asures us that both spin components do arise.

Recall that our genus 2 base curve $C$ is hyperelliptic, a double cover of $\PP^1$ branched at the 6 points $B = \{0, \infty, e_0=1, e_1, e_2, e_3 \}$. Let $p_i \in C$ denote the corresponding Weierstrass points, for $i \in B$. The double cover $\overline{C}$ is determined by a $4+2$ partition of these 6 branch points: say $0, \infty$ vs. $e_0, e_1, e_2, e_3$. The kernel of the pullback $\pi_2^*: J(C) \to J(\bar{C})$ is isomorphic to $\ZZ_2$, and we let $\mu$ be the non-zero element. It is a line bundle on $C$, given explicitly by $\OO_C(p_0 - p_{\infty})$. Its pullback  $\wt{\mu}:=\rho^*\mu$ is a line bundle of order 2 on $\wt{C}$, and is the non-trivial element in the kernel of $\wt{\pi}_2^*$.

Fix a point $\bar{p}$ of $\bar{C}$. The involution  $\tau = \tau_2$ on $\bar{C}$ takes it to  $\tau\bar{p}$, and both map by $\pi_2$ to the same point $p \in C$. Due to the total ramification, these three points have unique lifts to points $\bar{\tilde{p}}, {\wt{\tau}} \bar{\tilde{p}} \in \bar{\tilde{C}}$ and ${\wt{p}} \in {\wt{C}}$.

\begin{proposition}
Each of the 16 spin structures ${N}$ on $C$ determines a relative spin structure (i.e. a square root of the relative canonical bundle) for the family $\wt{\mathcal C}  \to \bar{C}_r $, so in particular a lift $\sigma_N: \bar{C}_r \to \sM_5$ 
of our $\sigma: \bar{C}_r \to \M_5$. Of the 10 even spin structures $N$, 6 give odd spin structures on the $\wt{C}$ and 4 give even spin structures on the $\wt{C}$.
\end{proposition}
\begin{proof}

Let $N$ be a spin structure on $C$, i.e. a line bundle satisfying $N^2 \cong K_C$. It induces spin structures 
\[
\bar{N} := \pi_2^* N, \quad   
\wt{N} := \rho^*N({\wt{p}} ), \quad   
\bar{\tilde{N}} := \bar{\rho}^*\pi_2^*{N}(\bar{\tilde{p}} + {\wt{\tau}} \bar{\tilde{p}} ) 
        = \bar{\rho}^*\bar{N}(\bar{\tilde{p}} + {\wt{\tau}} \bar{\tilde{p}} ) 
        = \wt{\pi}_2^*\wt{N} 
\]
on $\bar{C}$, $\bar{\tilde{C}}$, and ${\wt{C}}$ respectively. We have two other natural spin structures: $N \otimes \mu$ on $C$ and $\wt{N} \otimes \wt{\mu}$ on $\wt{C}$, which pull back to the same $\bar{N}, \bar{\tilde{N}}$. The choice of $\wt{N}$ clearly gives a lift $\sigma_N: \bar{C}_r \to \sM_5$. Moreover, we get a global line bundle $\wt{\mathcal{N}}$
on $\wt{\mathcal{C}}$ whose square is the  relative canonical bundle. We still need to compare the parities.

\begin{lemma}
The spin structures $\bar{N}, \bar{\tilde{N}}$ have the same parity.  
\end{lemma}
\begin{proof}
Since $\bar{\rho}$ is cyclic,  the direct image $\bar{\rho}_*(\bar{\tilde{N}})$ decomposes under $\ZZ_3$:
\[
\bar{\rho}_*(\bar{\tilde{N}}) = \bar{N} \oplus (\bar{N} \otimes L)   \oplus (\bar{N} \otimes L^{-1}),  
\]
where $L \in Pic(\bar{C})$ is the defining line bundle of the cyclic cover, satisfying \eqref{17}, \eqref{18}.
But the two bundles $\bar{N} \otimes L, ~ \bar{N} \otimes L^{-1}$ have the same parity (they are each other's Serre duals, and of zero Euler characteristic). So
\[
h^0(\bar{\tilde{C}}, \bar{\tilde{N}}) = h^0(\bar{{C}}, \bar{\rho}_*\bar{\tilde{N}}) = 
h^0(\bar{{C}}, \bar{{N}})  \quad \quad  (\text{mod} ~2).
\]
\end{proof}

\begin{lemma}
The spin structures $\wt{N}, ~ N \otimes \mu$ have opposite parity. (As do $\wt{N} \otimes \wt{\mu}, ~ N$.)
\end{lemma}
\begin{proof}
These parities remain constant over connected families, so we may as well specialize to a convenient cover ${\wt{C}}$. We take it to be one of the orbifold points in the image $\sigma(\bar{C}_r) \subset \M_5$. 
Namely, we take the point $(e,\bar{p}) \in \bar{C}_r = E_r \times_E \bar{C}$ where $e$ is one of the four ramification points of $E_r$ over $\PP^1_r$, so $m_3 e$ is one of  the four ramification points of $E$ over $\PP^1$, say the one over $1 \in \PP^1$, and $\bar{p} $ is one of the two points in $\pi_1^{-1}(e)$, with image $p_1 \in C$. The corresponding $\wt{C}$ is Galois over $\PP^1$, with group $S_3$: it is the fiber product
\[
\wt{C} \cong C \times_{\PP^1} \wt{\PP^1}, 
\]
where $\wt{\PP^1}$ is the triple cover of $\PP^1$ with total ramification over $1 \in \PP^1$ and simple ramification over  $0, \infty \in \PP^1$. (Such $\wt{\PP^1}$ is uniquely specified by the above.) We see that in this special case where the cover is parametrized by one of the orbifold points,  our previous snapshot can be extended:

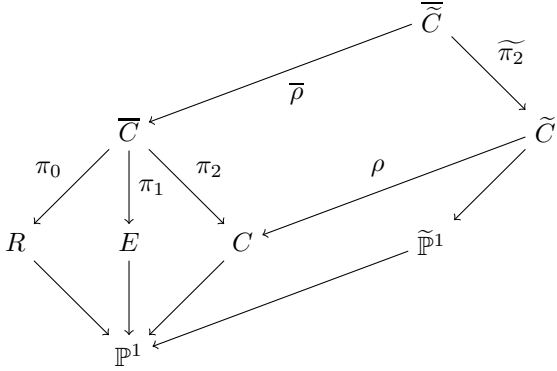
\begin{center}

\begin{tikzpicture}[node distance=2cm, auto]
 \node (C) {$C$};
  \node (C-) [left of=C, above of=C, node distance=1.5cm] {${\overline{{C}}}$};
 \node [above of=C-, node distance=1.5cm] (G) {};
  \node [right of=G, node distance=4cm] (C1-) {${\overline{\wt{C}}} $};  
  \node [right of=C1-, below of=C1-,  node distance=1.5cm] (C3) {${\wt{C}}$};
  \node (P1) [left of=C, below of=C, node distance=1.5cm] {${\PP}^1$};
  \node (R) [left of=P1, above of=P1, node distance=1.5cm] {$R$};
  \node (E) [ above of=P1, node distance=1.5cm] {$E$};
  \node (P1t) [left of=C3, below of=C3, node distance=1.5cm] {$\wt{\PP}^1$};

  \draw[->] (C) to node {} (P1);
  \draw[->] (E) to node {} (P1);
  \draw[->] (R) to node {} (P1);
  \draw[->] (P1t) to node {} (P1);
  \draw[->,swap] (C3) to node {$\rho$} (C);
  \draw[->] (C3) to node {} (P1t);
  \draw[->] (C-) to node {$\pi_2$} (C);
  \draw[->] (C-) to node {$\pi_1$} (E);
  \draw[->,swap] (C-) to node {$\pi_0$} (R);
  \draw[->] (C1-) to node {$\bar{\rho}$} (C-);
  \draw[->] (C1-) to node {$\wt{\pi_2}$} (C3);
\end{tikzpicture}

\end{center}

The advantage of this choice is that now $\wt{C}$ is hyperelliptic as is $C$, so we know everything about spin structures on them and we can check the claim directly. On a hyperelliptic curve of genus $g$ with hyperelliptic line bundle $H$, the spin structures are of the form 
$\OO(D) \otimes H^{(g-1-\ell)/2}$
where $D$ is a subset of cardinality $\ell$ of the set of Weierstrass points, and $\ell \equiv g-1 ~ (\text{mod} ~ 2)$. The parity of this spin structure is then $(g+1-\ell)/2$.

Thus the 6 odd spin structures on our  $_2C$ are the 
\[
N=\OO_C(p_i), ~ i \in B = \{0, \infty, e_0=1, e_1, e_2, e_3 \}, ~
\]
and the 10 even ones are of the form $N=\OO_C(p_i + p_j - p_k), \quad i,j,k \in B$. The line bundle $\mu$ is 
$\mu= \OO_C(p_0  - p_{\infty})$. 

Our $_5\wt{C}$ has one Weierstrass point $\wt{p}_i, ~ i=0, \infty, 1$ above the corresponding $p_i$, and 3 Weierstrass point $\wt{p}_j^a, ~ j=e_1,e_2,e_3, ~ a=1,2,3$ above the corresponding $p_j$, for a total of  12 Weierstrass points. These satisfy:
\[
\rho^*\OO_C(p_i) \cong \OO_{\wt{C}}(\wt{p}_i) \otimes H, \quad i=0, \infty, 1
\]
and
\[
\rho^*\OO_C(p_j) \cong \OO_{\wt{C}}(\wt{p}_j^1 + \wt{p}_j^2 + \wt{p}_j^3) , \quad j=e_1,e_2,e_3.
\]
It is therefore natural to write
\[
\ell = \ell_0 + \ell_1 + \ell_2  \equiv 1 ~ (\text{mod} ~ 2)
\]
where $\ell_0, \ell_1, \ell_2$ are the numbers of points of $D$ from the subsets $\{0, \infty \}, \{ 1 \}, \{ e_1,e_2,e_3 \} $ respectively. The corresponding partition for $\wt{N} = \rho^*{N}(\wt{p}_1)$ is therefore
\[
\wt{\ell} = \wt{\ell}_0 + \wt{\ell}_1 + \wt{\ell}_2
\]
with
\[
\wt{\ell}_0= \ell_0, ~  \wt{\ell}_1= 1-\ell_1, ~ \wt{\ell}_2=3\ell_2
\]
so
\[
\text{parity}(\wt{N}) = \frac{5-\ell_0+\ell_1-3\ell_2}{2}.
\]
On the other hand, for $N \otimes \mu$ we have 
\begin{align*}
\ell_{\mu} &= \ell_{\mu,0} + \ell_{\mu,1}+ \ell_{\mu,2}\\
                 &=(2-\ell_{0}) + \ell_{1} + \ell_{2}
\end{align*}
so
\begin{align*}
\text{parity}({N \otimes \mu}) &= \frac{1+\ell_0-\ell_1-\ell_2}{2} \\
                                    &= \text{parity}(\wt{L}) - (2+2\ell_1 - \ell)\\
                                    &\equiv \text{parity}(\wt{L}) -1 ~ (\text{mod} ~ 2).
\end{align*}
\end{proof}

To complete the proof of the Proposition, we therefore have to count the even spin structures $N$ for which $N \otimes \mu$ is odd. In the notation of the previous proof, the condition is that $\ell_0$ should be even. There are indeed four of these: $\OO_{C}(p_k) \otimes \mu$, where $k$ is one of $e_0, e_1, e_2, e_3$. 

\end{proof}


\begin{thebibliography}{99}

\bibitem{Green}
P. Green, ``On Holomorphic Graded Manifolds,''
PAMS v.85 (1982).

\bibitem{Berezin}
F. Berezin, \emph{Introduction to Superanalysis},
Reidel Publishing (1987).

\bibitem{Manin}
Y.  Manin, \emph{Gauge Field Theory and Complex Geometry},
Grundlehren 289, Springer (1988).

\bibitem{Vaintrob}
A.  Vaintrob, ``Deformations of Complex Structures on Supermanifolds,''
Functional Analysis and Its Applications, Vol. 18, pp 135-136 (1984).


\bibitem{Rothstein}
M. Rothstein,   \emph{Deformations of Complex Supermanifolds},
PAMS v.95 (1985).

\bibitem{Onishchik}
A.L. Onishchik, ``Non-Split Supermanifolds Associated with the Cotangent Bundle,'' (1986),
published in  \emph{Lie Groups, Geometric Structures and Differential Equations, One Hundred Years after Sophus Lie}, Surikaisekikenkyusho Kokyuroku No. 1150 (2000), 45-53. 

\bibitem{BS1}
A. Baranov and A. Schwarz, ``Multiloop Contribution To String Theory,''
 JETP Lett. 42 (1985) 419-421.

\bibitem{BS2}
A. Baranov and A. Schwarz, ``On The Multiloop Contribution To The String Theory,''
Int. J. Mod. Phys. A2 (1987) 1773.

\bibitem{F}
D. Friedan, ``Notes On String Theory And Two Dimensional Conformal Field Theory,''
in: M. B. Green et. al., eds. \emph{Unified String Theories} (World-Scientic, 1986), pp. 162-213.

\bibitem{M2}
Yu. I. Manin, ``Critical Dimensions of String Theories and the Dualizing Sheaf on the Moduli
Space of (Super) Curves,''
Funct. Anal. Appl. 20 (1987) 244.

\bibitem{RSV}
A. A. Rosly, A. S. Schwarz, and A. A. Voronov, ``Geometry Of Superconformal Manifolds,''
Comm. Math. Phys. {\bf 119} (1988) 129-152.


\bibitem{SRSnotes}
E. Witten,  ``Notes On Super Riemann Surfaces And Their Moduli,''
arXiv:1209.2459.

\bibitem{DRS}
S. N. Dolgikh, A. A. Rosly, and A. S. Schwarz, ``Supermoduli Spaces,'' Commun. Math. Phys. {\bf 135} (1990) 91-100.

\bibitem{DPh}
E. D'Hoker and D. H. Phong, ``The Geometry Of String Perturbation Theory,'' Rev. Mod. Phys. {\bf 60} (1988) 917-1065.

\bibitem{Revisited}
E. Witten,  ``Superstring Perturbation Theory Revisited,''
arXiv:1209.5461.

\bibitem{DPhgold}
E. D'Hoker and D. H. Phong, ``Lectures On Two-Loop Superstrings,'' Adv. Lect. Math. {\bf 1} 85-123, 
hep-th/0211111.




\bibitem{Nagata}
M. Nagata, ``Imbedding of an abstract variety in a complete variety", J. Math. Kyoto {\bf 2 (1)} (1962) 1–10.

\bibitem{Conrad}
B. Conrad, ``Deligne’s notes on Nagata compactifications," Journal of the Ramanujan Math. Soc. {\bf 22}  (2007)
205–257.







\bibitem{ER}
J. Ebert and O. Randal-Williams,  ``Stable Cohomology Of The Universal Picard Varieties And The Extended Mapping Class Group,''
arXiv:1012.0901.

\bibitem{R}
O. Randal-Williams, ``The Picard Group Of The Moduli Space Of $r$-Spin Riemann Surfaces,'' arXiv:1102.0715.

\bibitem{GV}
 S. Gorchinskiy and F. Viviani,  ``A note on families of hyperelliptic curves," arxiv:0802.0635.
 
\bibitem{Natanzon}
S. M. Natanzon, ``Moduli Of Riemann Surfaces, Real Algebraic Curves, And Their Superanalogs'',
AMS Translations of Mathematical Monographs (2004); Vol. 225, ISBN-10: 0-8218-3594-7.

\bibitem{DWtwo}
R. Donagi and E. Witten, 
``Super Atiyah classes and obstructions to splitting of supermoduli space'',
Special Issue: In memory of Andrey Todorov, 
Pure and Applied Math Quarterly 9 (2013), 739-788,
arXiv:1404.6257 

\bibitem{SMnotes}
E. Witten,  ``Notes On Supermanifolds and Integration,''
arXiv:1209.2199

\bibitem{K}
K. Kodaira,   ``A Certain Type of Irregular Algebraic Surfaces,'' 
J. Analyse Math. 19 (1967) 207–215.

\bibitem{A}
M. Atiyah, ``The Signature of Fibre-Bundles,''  in \emph{Global Analysis (Papers
in Honor of K. Kodaira)}, Univ. Tokyo Press, Tokyo (1969) 73-84.

\bibitem{H}
F. Hirzebruch, ``The Signature of Ramified Coverings,'' in \emph{Global Analysis
(Papers in Honor of K. Kodaira)}, Univ. Tokyo Press, Tokyo (1969) 253-265.

\bibitem{BD}
J. Bryan and R. Donagi,   ``Surface Bundles Over Surfaces of Small Genus,''
Geometry and Topology 6 (2002) 59-67.

\bibitem{O}
F. Oort, ``Subvarieties of Moduli Spaces,''
Inventiones Math. 24, 95 - 119 (1974).


\bibitem{HM}
J. Harris, I. Morrison,  \emph{Moduli of Curves}
Graduate Texts in Mathematics 187, Springer (1998).

\bibitem{Cornalba}
M. Cornalba,   ``Moduli of Curves and Theta Characteristics,''
in: Lectures on Riemann surfaces (Trieste, 1987),
World Sci. Publ., 560-589.


\bibitem{Caporaso}
L. Caporaso and C. Casagrande,   ``Combinatorial Properties of Stable Spin Curves,''
Communications in Algebra 359 (2007), 3733-3768, arXiv:math/0209018.


\bibitem{Farkas}
G. Farkas,   ``Theta Characteristics and Their Moduli,''
Milan journal of mathematics 80.1 (2012): 1-24,
arXiv:1201.2557.

\bibitem{Mumford}
D. Mumford (1974), "Prym varieties. I", in Ahlfors, Lars V.; Kra, Irwin; Nirenberg, Louis et al., Contributions to analysis (a collection of papers dedicated to Lipman Bers), Boston, MA: Academic Press, pp. 325–350.


\end{thebibliography}
\end{document}